\documentclass[11pt]{article}
\usepackage{amssymb}
\usepackage{graphicx}
\usepackage{amsmath}
\usepackage{makeidx}
\usepackage{indentfirst}
\usepackage[T1]{fontenc}
\usepackage[utf8]{inputenc}
\usepackage{authblk}

\newcounter{resultnum}[section]
\newcounter{conclusionnum}[section]
\newcounter{conditionnum}[section]
\newcounter{conjecturenum}[section]
\newcounter{examplenum}[section]
\newcounter{exercisenum}[section]
\newtheorem{lemma}{Lemma}[section]

\newcounter{lemmanum}[section]
\newcounter{notationnum}[section]
\newtheorem{theorem}{Theorem}[section]

\newcounter{theoremnum}[section]
\newtheorem{definition}{Definition}[section]

\newcounter{definitionnum}[section]
\newtheorem{corollary}{Corollary}[section]

\newcounter{corollarynum}[section]
\newcounter{remarknum}[section]
\newtheorem{proposition}{Proposition}[section]

\newcounter{propositionnum}[section]
\newcounter{acknowledgementnum}[section]
\newcounter{algorithmnum}[section]
\newcounter{axiomnum}[section]
\newcounter{casenum}[section]
\newcounter{claimnum}[section]
\newcounter{summarynum}[section]
\newcounter{problemnum}[section]
\newenvironment{proof}[1][]{\textbf{Proof.} }{}

\begin{document}

\title{Nonholonomic Jet Deformations, Exact Solutions\\
for Modified Ricci Soliton and Einstein Equations}
\date{September 23, 2016}
\author{
{Subhash Rajpoot}\\
{\small\it California State University at Long Beach,} \\
{\small\it Long Beach, California, USA}  \\
{\small\it email: Subhash.Rajpoot@csulb.edu}
${}$
\vspace {.1 in}  \\
{ Sergiu I. Vacaru }\\
{\small \textit{Quantum Gravity Research; 101 S. Topanga Canyon Blvd \#
1159. Topanga, CA 90290, USA}}  {\small and }\\
{\small \textit{University "Al. I. Cuza" Ia\c si, Project IDEI }} \\
{\small \textit{18 Pia\c ta Voevozilor bloc A 16, Sc. A, ap. 43, 700587 Ia\c
si, Romania }} \\
{\small \textit{email: sergiu.vacaru@gmail.com}}
}
\maketitle

\begin{abstract}
Let $\mathbf{g}$ be a pseudo--Riemannian metric of arbitrary signature on a
manifold $\mathbf{V}$ with conventional $n+n$ dimensional splitting, $\
n\geq 2,$ determined by a nonholonomic (non--integrable) distribution $%
\mathcal{N}$ defining a generalized (nonlinear) connection and associated
nonholonomic frame structures. We work with an adapted linear metric
compatible connection $\widehat{\mathbf{D}}$ and its nonzero torsion $%
\widehat{\mathcal{T}}$, both completely determined by $\mathbf{g}$. Our
first goal is to prove that there are certain generalized frame and/or jet
transforms and prolongations with $(\mathbf{g},\mathbf{V})\rightarrow (%
\widehat{\mathbf{g}},\widehat{\mathbf{V}})$ into explicit classes of
solutions of some generalized Einstein equations $\widehat{\mathbf{R}}%
\mathit{ic}=\Lambda \widehat{\mathbf{g}}$, $\Lambda =const$, encoding
various types of (nonholonomic) Ricci soliton configurations and/or jet
variables and symmetries. The second goal is to solve additional constraint
equations for zero torsion, $\widehat{\mathcal{T}}=0$, on generalized
solutions constructed in explicit forms with jet variables and extract
Levi--Civita configurations. This allows us to find generic off--diagonal
exact solutions depending on all space time coordinates on $\mathbf{V}$ via
generating and integration functions and various classes of constant jet
parameters and associated symmetries. Our third goal is to study   how such
generalized metrics and connections can be related by the so--called
"half-conformal" and/ or jet deformations of certain sub--classes of
solutions with one, or two, Killing symmetries. Finally, we
present some examples of exact solutions constructed as
nonholonomic jet prolongations of the Kerr metrics, with possible Ricci
soliton deformations, and characterized by nonholonomic jet structures and
generalized connections.

\vskip0.2cm


\vskip0.2cm

\end{abstract}



\section{Introduction}

Various results and methods of the theory of nonholonomic manifolds, jets
and connections can be combined and applied to the study of symmetries of
systems of nonlinear partial differential equations, PDEs, and constructing
exact and approximate solutions. In modern physics, such fundamental field
and evolution equations are related to the Ricci soliton geometry,
mathematical relativity, particle physics and geometric mechanics \cite%
{saunders,torreand,torre,vexsol1}. For instance, a jet space technique was
elaborated upon to analyze special features of the vacuum Einstein equations in
general relativity, GR, that allows to define certain generalized
symmetries and conservation laws. In a more general context, a Lagrangian
formalism was elaborated  on the jet--gauge and jet--diffeomorphism groups with
the aim of unifying gravity with internal gauge symmetries \cite{aldaya}.
Another direction related to Finsler--Lagrange geometry and nonholonomic
mechanics was considered by the authors of papers \cite%
{balanstavr,neagu,atanasiu,balan}, where certain generalizations of Einstein
equations were formulated on jet spaces endowed with nonlinear connection
structures.

In recent years,  series of works have been devoted to elucidating geometric
methods  that allow for the  decoupling of (modified) Einstein equations for certain
"auxiliary" connections with respect to adapted nonholonomic frames, and
constructing generic off--diagonal solutions\footnote{%
which can not be diagonalized by coordinate transforms in a finite spacetime
region, for instance, in GR} depending on all spacetime coordinates, see
reviews of results in \cite{vexsol1,vexsol2,gheorghiu}. Following the
so--called anholonomic frame deformation method, AFDM, the solutions are
generated in explicit forms via formulae determined by generating and
integration functions, and various commutative and noncommutative parameters.
Such solutions may exhibit Killing, non--Killing solitonic and/or other type of
symmetries, which for respective boundary/ initial / source conditions can
be with nontrivial spacetime topology. The solutions may also describe evolution and/or
dynamical processes, or result in stochastic behaviour. We can extract
Levi--Civita configurations with zero torsion if we impose additional
nonholonomic constraints on certain classes of generalized solutions. It
should be noted that because such systems are nonlinear it is important to
consider the restrictions via integration / generation functions and
constants, symmetry / boundary / initial conditions "at the end", on some
defined integral varieties. By prescribing from the very beginning only some
special ansatz for the metrics and connections  may result in a simplified system of
equations (for instance, to transform it into a nonlinear system of ordinary
differential equations), that may not  decouple the PDEs in a general form. This can
result in  reducing the number of  the bulk of nonlinear off--diagonal multi-variables.

The goal of this work is to study the basic properties of nonholonomic Ricci
soliton and (modified) Einstein equations with metrics and (generalized)
connections generated by jet prolongations of exact solutions. We
study also the constraints under which various classes of solutions with
generalized jet variables and symmetries are transformed into standard
Einstein metrics with jet parametric dependence of generic off--diagonal
metrics. Readers are referred to monographs \cite{saunders,kolar} on main
results on jets and jet bundle geometry. The literature on nonholonomic jet
manifolds and bundles is less popular and more sophisticated than that on
holonomic jets. Experts on mathematical relativity and PDEs are less
familiar with the geometry of nonholonomic manifolds elaborated as in the Vr\v{a}%
nceanu--Horak approach \cite{vranceanu1,vranceanu2,vranceanu3,horak}, see
recent results and applications in \cite{bejancu,vjmp14}. We cite here some
important works on generalized connections developed by different schools of
differential geometry on nonholonomic jets, quasi-jets and the theory of higher
order connections, see \cite{ehresmann54,pradines,tomas,dekret,virsik}. We will
 sketch a few essential notions and necessary results using recent
approaches formulated in Refs. \cite{kures98,kures01}.

In this work, we follow three explicit goals motivated and stated in section %
\ref{ssgoals}. The first goal is to develop the anholonomic frame deformation
methods, (AFDM, see reviews of results in \cite{vexsol1,vex3,vexsol2}) in
such a form that will allow us to decouple the nonholonomic $r$--jet
deformations of the Ricci soliton and Einstein equations, and integrate
such equations for general classes of generic off--diagonal metric and
nonlinear connection structures. The second goal is to show how we can
extract from extra dimensional jet configurations the Levi--Civita connections (in
particular, physically important solutions in GR with jet parameters) by
solving the nonholonomic constraints for zero torsion conditions. Finally, the third
goal is to  analyze explicit examples of exact solutions depending on
jet parameters defining nonholonomic deformations of black hole solutions
and gravitational solitonic waves. We study how nonholonomic and/or $r$--jet
deformations of the Kerr metric may model the physical effects of Ricci
solitons in massive gravity and other modified gravity models \cite%
{capoz,odints1,odints2,drg1,drg3,hr1,hr2,kour,stavr,mavr,gheorghiu}.

The article is organized as follow. In section \ref{nhmj}, we recall basic
facts and definitions concerning nonholonomic manifolds and jets and
elaborate on the concept of generalized connection structure. We provide an
introduction to the geometry of nonholonomic manifolds and bundles endowed
with nonlinear connection structures. There are outlined main results and
stated respective denotations on nonholonomic maps and jets of (non)
holonomic manifolds. \textit{The first important result is formulated in
Theorem \ref{tcandist}} stating that there is a canonical distinguished
connection structure on $r$--jet prolongation of (modified) Einstein
manifolds which will allow to prove the main results  in the next section.
Then, we elaborate in details the formalism of nonholonomic $r$--jet
prolongation of Ricci soliton and (generalized/modified) equations. This
allows us to formulate and prove \textit{the second important result, i.e.
Theorem \ref{th2.4}}, which provides the N--adapted equations for gradient
canonical Ricci jet-solitons and generalizes the jet-extensions of Einstein
manifolds.

In section \ref{s3}, we formulate and prove \textit{the main theorems (the
first two main results of this work) on decoupling, see Theorem \ref%
{tdecoupling}, and integration, see Theorem \ref{th3.2},} of (modified)
Ricci soliton and Einstein equations. The approach consists of the  generalization
of the results for nonholonomic jet prolongations of fundamental geometric and
physical objects in generalized/ modified gravity theories and further
developments of AFDM. The key idea is to consider nonholonomic 2+2+2+...
splitting with two dimensional (2-d) shells of jet coordinates and adapting
the geometric constructions for such nonholonomic spacetime and jet
distributions.

Section \ref{s4} is devoted to explicit examples of exact solutions
depending on jet coordinates, jet parameters, symmetries,  Killing
and non-Killing symmetries, deformations by Ricci soliton configurations,
modified gravity contributions, mimicking massive gravity terms with
effective cosmological constant and gravitational polarizations. \textit{The
third main result of this paper, Theorem \ref{th4.1}}, is related to Ricci
soliton modifications and $r$--jet prolongations of the Kerr metric which
play an essential role in the physics of black holes. Such black hole
metrics can be extended to generic off--diagonal forms for various classes
of modified gravity theories with extra dimensions \cite%
{vexsol1,vex3,vexsol2,gheorghiu}. The AFDM even allows us to construct very
general integral varieties for such gravitational and geometric evolution
like nonlinear systems of PDEs. For $r$--jet configurations, it is clear
that new classes of gravitational and matter field equations at least
possess certain jet type local symmetries and possible association with  nonlinear
gauge interior degrees of freedom. We show that such solutions can be
constructed both with zero or non--zero canonical torsion, with possible
rotoid symmetries for Kerr -- de Sitter configurations and other classes of
vacuum and non-vacuum jet prolongations.

In Appendix, we provide a summary of the   most important and necessary N--adapted
coefficient formulas and provide technical details of  some theorems.

\vskip5pt

\textbf{Acknowledgments:\ } The work was partially supported by the Program
IDEI, PN-II-ID-PCE-2011-3-0256 and visiting research programs at CERN and
visiting DAAD fellowships. The results were communicated at the Marcel Grossman
Conference in Rome in 2015.

\section{Nonholonomic Manifolds, Jet Bundles and Generalized Connections}

\label{nhmj}We start by recalling a few basic definitions on the geometry of
nonholonomic manifolds and bundles, related jet spaces and theory of
generalized (nonlinear) connections \cite{saunders,kolar,kures98,kures01}.
The geometric approach is generalized in a form to unify both the concepts
of nonholonomic manifolds \cite%
{vranceanu1,vranceanu2,vranceanu3,horak,bejancu,vjmp14} and that of
nonholonomic jet spaces \cite{ehresmann54,pradines,kures01}. \

We shall work in the category of $n+m$ dimensional nonholonomic manifolds $%
\mathcal{V},$ with $n,m\geq 2,$ of necessary smooth class (for instance, of
class $\mathcal{C}^{\infty }$), Hausdorff, finite dimensional and without
boundaries. The solutions of certain systems of nonlinear partial
differential equations (PDE) can be topologically nontrivial, with
singularities and various type of Killing and non--Killing symmetries. Such
PDEs, nonholonomic constraints\footnote{%
equivalently, anholonomic (nonholonomic), i.e. non--integrable} and their solutions are for
geometric models of (modified) gravity theories and Ricci soliton equations
defined as certain stationary configurations in a nonholonomic geometric
evolution system, with possible Wick rotations (for small deformations) and
frame transformations  between Lorentzian and Euclidean signatures of metrics.

\subsection{Holonomic jets}
\label{ssholj}
Jets are certain equivalence classes of smooth maps between two manifolds $%
M,\dim M=n,$ and $Q,\dim Q=m,$ when the maps are represented by Taylor
polynomials. One writes this as $f,g:M\rightarrow Q$ and says that a $r$-jet
is determined at a point $u\in M$ if there is a $r$-th order contact at $u.$
The idea is formalized mathematically using the concept as the $r$-th order
contact of two curves on a manifold.

\begin{definition}
-\textbf{Lemma:} Two curves $\gamma ,\delta :$ $\mathbb{R}\rightarrow V$
have the $r$-th contact at zero if for every smooth function $\varphi $ on \
$M$ the difference $\varphi \circ \gamma -\varphi \circ \delta $ vanishes to
$r$-th order at $0\in \mathbb{R}.$ In this case, we have an equivalence
relation $\gamma \sim _{r}\delta $ when $r=0$ implies $\gamma (0)=\delta (0).$
If $\gamma \sim _{r}\delta ,$ then $f\circ \gamma \sim _{r}f\circ \delta $
for every map $f:b\rightarrow Q$
\end{definition}

Two maps $f,g:V\rightarrow Q$ are said to determine the same $r$--jet at $%
x\in M,$ if for every curve $\gamma :$ $\mathbb{R}\rightarrow V$ with $%
\gamma (0)=x$ the curves $f\circ \gamma $ and $g\circ \gamma $ have the $r$%
-th order contact at zero. In such a case, we write $j_{x}^{r}f=j_{x}^{r}g,$
or $j^{r}f(x)=j^{r}g(x).$ An equivalence class of this relation is called an
$r$-jets of $M$ into $Q.$

\begin{definition}
\label{drjet}The set of all $r$--jets of $M$ into $Q$ is denoted by $%
J^{r}(M,Q)$; for an element $X=j_{x}^{r}f\in J^{r}(M,Q),$ the point $%
x:=\alpha X$ is the source of $X$ and the point $f(x)=:\beta X$ is the
target of $X.$
\end{definition}

One denotes by $\pi _{s}^{r},0\leq s\leq r$ the projection $%
j_{x}^{r}f\rightarrow j_{x}^{s}f$ of $r$--jets into $s$--jets. All $r$--jets
form a category, the units of which are the $r$--jets of the identity maps
of manifolds.

By $J_{x}^{r}(M,Q),$ or $J_{x}^{r}(M,Q)_{y}$ we mean the set of all $r$-jets
of $x$ onto $Q$ with source $x\in M,$ or tangent $y\in Q,$
\begin{equation*}
J_{x}^{r}(M,Q)_{y}=J_{x}^{r}(M,Q)\cap J_{x}^{r}(M,Q)_{y}\mbox{ and }%
L_{n,m}^{r}=J_{0}^{r}(\mathbb{R}^{n},\mathbb{R}^{m})_{0}
\end{equation*}%
In local coordinates $x^{i},$ the value $\partial _{\check{i}}f:=\frac{%
\partial ^{|\check{i}|}f}{(\partial x^{1})^{i_{1}}...(\partial x^{n})^{i_{n}}%
}$ is the partial derivative of a function $f:U\subset \mathbb{R}%
^{n}\rightarrow \mathbb{R},$ with a \textbf{multi-index} $\check{i}$ of
range $n,$ which is a $m$--tuple $\check{i}=(i_{1},...,i_{n})$ of
non-negative integers. We write $|\check{i}|=i_{1}+...+i_{n},$ with $\check{i%
}!=i_{1}!i_{2}!...i_{n}!,$ $0!=1,$ and $x^{\check{i}%
}=(x^{1})^{i_{1}}...(x^{n})^{i_{n}}$ for $x=(x^{1},x^{2},...,x^{n})\in
\mathbb{R}^{n}.$\footnote{%
Our definition of multi-index derivative $\partial _{\check{i}}f$ \ is
similar to $D_{i}f$ used in \cite{kolar}. For clarity, we need to modify the system of
notations in order to elaborate a geometric method of constructing exact
solutions of PDEs with jet variables.}

Consider a local coordinate system $x^{i}$ on $M$ and a local coordinate
system $y^{a}$ on $Q.$ Two maps $f,g:M\rightarrow Q$ satisfy $%
j_{x}^{r}f=j_{x}^{r}g$ if and only if all the \ partial derivatives up to
order $r$ of the components $f^{a}$ and $g^{a}$ of their coordinate
expressions coincide at $x.$ In this a case the chain rule implies $f\circ
\gamma \sim _{r}g\circ \gamma .$ For the curves $x^{i}=\zeta ^{i}t$ with
arbirtary $\zeta ^{i},$ these conditions read $\sum\limits_{|\check{i}%
|=k}(\partial _{\check{i}}f^{a}(x))\zeta ^{\check{i}}=\sum\limits_{|\check{i}%
|=k}(\partial _{\check{i}}g^{a}(x))\zeta ^{\check{i}}$, for $k=0,1,...,r.$

The elements of $L_{n,m}^{r}$ can be identified with the $r$-th order Taylor
expansions of the generating maps, i.e. with  $m$-tuples of polynomials
of degree $r$ in $m$ variables without the absolute term. Such an expression $%
\sum\limits_{1\leq |\check{i}|\leq r}\zeta _{\check{i}}^{a}x^{\check{i}}$ is
the polynomial representative of a $r$--jet. Hence $L_{n,m}^{r}$ is a
numerical space of the variables $\zeta _{\check{i}}^{a}.$ Standard
combinatorics yields $\dim L_{n,m}^{r}=m\left[ \left(
\begin{array}{c}
n+r \\
n%
\end{array}%
\right) -1\right] .$ The coordinates on $L_{n,m}^{r}$ are sometimes denoted
more explicitly by $\zeta _{i}^{a},\zeta _{ij}^{a},...,\zeta
_{i_{1}...i_{r}}^{a},$ symmetric in all subscripts. The projection $\pi
_{s}^{r}:$ $L_{m,n}^{r}\rightarrow L_{m,n}^{s}$ consists in suppressing all
terms of degree $>s.$

The set of all invertible elements of $L_{n,m}^{r}$ with the jet composition
is a Lie group $G_{m}^{r}$ called the $r$-th differential group of the $r$%
-th jet group in dimension $m.$ For $r=1,$ the group $G_{m}^{1}$ is
identified with $GL(m,\mathbb{R}).$

Let $p:Q\rightarrow M$ be a fibered manifold.

\begin{definition}
\label{dprolj}A map $j^{r}f:M\rightarrow J^{r}(M,Q)$ is called a $r$-th jet
prolongation of $f:M\rightarrow Q.$ The set $J^{r}Q$ of all $r$--jets of the
local sections of $Y$ is called the $r$--th jet prolongation of $Q$ and $%
J^{r}Q\subset J^{r}(M,Q)$ is a closed submanifold.
\end{definition}

We note that if $Q\rightarrow M$ is a vector bundle, then $J^{r}Q$ is a also
a vector bundle.

\subsection{Nonholonomic manifolds and nonlinear connections}

\label{ssnm}The concept of nonholonomic jet is elaborated in Refs. \cite%
{ehresmann54,pradines,kures01}, when multi--indices are not symmetric and
the jet spaces are subject to certain non--integrable conditions.
Nonholonomic structures with non--integrable constraints can be defined not only on
the space of jets but also on the 'prime', $M,$ and 'target', $Q,$
manifolds. In our approach, we shall elaborate a geometric formalism
encoding nonholonomic geometric structures both on manifolds and maps, i.e.
on $M,Q$ and $J^{r}(M,Q).$

By definition, a \textit{nonholonomic manifold} $\mathbf{V}$ \textit{is a
manifold endowed with a nonholonomic distribution.} In this work we follow
the approach elaborated by G. Vr\u{a}nceanu \cite%
{vranceanu1,vranceanu2,vranceanu3} and Z. Horak \cite{horak}, see reviews
\cite{bejancu,vjmp14}. For our purposes (to construct jet-generalizations of
the Einstein equations and physically relavant solutions), it is enough to
consider a nonholonomic distrubution determined by a \textit{nonlinear
connection} (N--connection) structure $\mathbf{N}=\{N_{i}^{a}(x,y)\}.$ Such
a N--connection can be introduced as a Whitney sum%
\begin{equation}
\mathbf{N:\ }T\mathbf{V}:=h\mathbf{V}\oplus v\mathbf{V,}  \label{whitney}
\end{equation}%
where $T\mathbf{V}$ is the tangent bundle of $\mathbf{V}$ and $h\mathbf{V}$
and $v\mathbf{V}$ are, respectively the horizontal (h) and vertical (v)
subspaces for a nonholonomic fibration.\footnote{%
Local coordinates are with a conventional 2+2 splitting, $u^{\alpha
}=(x^{i},y^{a}),$ with $i,j,...=1,2$ and $a,b,...=3,4,...;$ in brief, $%
u=(x,y)\in \mathbf{V}$ for any point and its coordinates. We shall use
boldface symbols in order to emphasize that certain spaces and/or geometric
objects are provided with or adapted to a N--connection structure.} N-connections
were used in coordinate form by E. Cartan in his model of Finsler geometry
\cite{cartanf} by considering $\mathbf{V}=TM$ as a tangent bundle to a manifold
$M.$ In a similar form, we can work with a vector bundle, $V=$ $E,$ on $%
M,\dim E=n+m,\dim M=n$ (for $n,m\geq 2)$ instead of $TM.$ The global
definition of N--connection is due to C. Ehresmann \cite{ehresmann50}. In
\cite{kolar}, such connections are studied for fiber bundles and are called
generalized (Ehresmann) connections.\footnote{%
We refer the readers to this monograph for a modern approach to differential
geometry and main results on jets, Weil bundles and generalized connections.}
We will follow a different system of notations that were elaborated upon and
used in the theory of nonholonomic (non) commutative Ricci flows,
nonholonomic Dirac operators and Clifford bundles, and,  deformations and
quantization of generalized geometries and gravity theories \cite%
{vnrf,vncrfvnrf,vjmp14}.

Any $\mathbf{N}$ defines a N--adapted frame structure $\mathbf{e}_{\alpha }=(%
\mathbf{e}_{i},e_{a})$, on $T\mathbf{V,}$ and co--frame structure $\mathbf{e}%
^{\beta }=(e^{j},\mathbf{e}^{b}),$ on the dual tangent bundle $T^{\ast }\mathbf{%
V,}$
\begin{equation}
\mathbf{e}_{\alpha }=(\mathbf{e}_{i}=\partial _{i}-N_{i}^{b}\partial
_{a},e_{a}=\partial _{a})\mbox{\ and\ }\mathbf{e}^{\beta }=(e^{j}=dx^{j},%
\mathbf{e}^{b}=dy^{b}+N_{i}^{b}dx^{i}),  \label{dderdif}
\end{equation}%
where the Einstein summation convention is applied on  repeated indices and $%
\partial _{i}=\partial /\partial x^{i}$ and $\partial _{a}=\partial
/\partial y^{b}.$ In general, such local bases are nonholonomic, i.e. $%
\mathbf{e}_{\alpha }\mathbf{e}_{\beta }-\mathbf{e}_{\beta }\mathbf{e}%
_{\alpha }=W_{\alpha \beta }^{\gamma }\mathbf{e}_{\gamma },$ with nontrivial
nonholonomy coefficients $W_{\alpha \beta }^{\gamma }.$ \ In this work, we
take these frame structures to be canonical in the sense that they
are linear on N--connection coefficients and admit (see Theorem \ref%
{tdecoupling}) the decoupling of (modified) Einstein equations in general form.
Here we note that although there are different canonical N--connection structures in
different models of Finler-Lagrange geometry, Hamilton geometry etc., see
details in \cite{cartanf,vnrf,vncrfvnrf,vjmp14},  those constructions can
not be used for constructing exact solutions in gavity theories. We call
certain geometric objects to be distinguished objects (d--objects), for
instance d--tensors, d--vectors if they are determined by the coefficients in
N-adapted form\footnote{%
this mean that certain geometric constructions are adapted to a horizontal $h$- and vertical $v$%
--splitting stated by a N--connection distribution (\ref{whitney})}, i.e.
with respect to N--elongated (co) bases (\ref{dderdif}) and their tensor
products. \ For instance, a vector $X\in T\mathbf{V}$ can be written in a
"non N-adapted" coordinate form, $X=X^{\alpha }\partial _{\alpha },$ or as a
d--vector, $\mathbf{X}=hX\oplus vX=\mathbf{X}^{\alpha }\mathbf{e}_{\alpha
}=X^{i}\mathbf{e}_{i}+X^{a}e_{a}.$

Two important characteristics of a N--connection are 1) the almost
complex  structure $\mathbf{J},$ where $\mathbf{J(e}_{i}\mathbf{)=-}e_{2+i}$
and $\mathbf{J(}e_{2+i}\mathbf{)=e}_{i},$ with $\mathbf{J}$ satisfying the symplectic relation $\mathbf{J\circ J=-}\mathbb{I}%
,$ where  $\mathbb{I}$ is the unit matrix and 2) the Neijenhuis tensor (called also
the curvature of N--connection) defined as %
\begin{equation*}
\ ^{N}\mathbf{J}[\mathbf{X,Y}]:=-[\mathbf{X,Y}]+[\mathbf{JX,JY}]-\mathbf{J}[%
\mathbf{JX,Y}]-\mathbf{J}[\mathbf{X,JY}],\ \forall \mathbf{X,Y\in }T\mathbf{V%
}.
\end{equation*}

Linear connections on $(\mathbf{V},\mathbf{N})$ can be defined in N--adapted
form as distinguished connections, \textit{d--connections}, in order to
preserve under parallel transport the distribution (\ref{whitney}). Such a
covariant differential operator splits as $\mathbf{D}=(hD,vD).$ We can
associate to $\mathbf{D}$ a 1--form $\mathbf{\Gamma }_{\ \alpha }^{\gamma }=%
\mathbf{\Gamma }_{\ \alpha \beta }^{\gamma }\mathbf{e}^{\beta }$ and
elaborate a N--adapted differential form calculus. The torsion and curvature
are defined, respectively, in terms of standard formulae:%
\begin{equation}
\mathbf{T}(\mathbf{X,Y)}:=\mathbf{D}_{\mathbf{X}}\mathbf{Y-D}_{\mathbf{Y}}%
\mathbf{X}+[\mathbf{X,Y}]\mbox{ \ and \ }\mathbf{R}(\mathbf{X,Y)}:=\mathbf{D}%
_{\mathbf{X}}\mathbf{D}_{\mathbf{Y}}-\mathbf{D}_{\mathbf{Y}}\mathbf{D}_{%
\mathbf{X}}\mathbf{-D}_{[\mathbf{X,Y}]}.  \label{dtordcurv}
\end{equation}%
Also, in the usual way, the Ricci d--tensor $\mathbf{R}ic$ is constructed by the contraction of indices
in  the curvature tensor $\mathbf{R}=\{\mathbf{R}_{\ \beta \gamma \mu
}^{\alpha }\},$ $\mathbf{R}ic:=\{\mathbf{R}_{\beta \gamma }=\mathbf{R}_{\
\beta \gamma \alpha }^{\alpha }\},$ \cite{vnrf,vncrfvnrf,vjmp14}. Readers
may study such papers, and references therein, on deformation quantization
of gravity based on almost complex structures characterizing generic
off-diagonal solutions.

Let $\mathbf{g}$ be a metric of arbitrary signature on a nonholonomic
manifold/ bundle $\left( \mathbf{V,N}\right) $ which in N--adapted form (\ref%
{whitney}) is represented as a symmetric d--tensor,
\begin{equation*}
\mathbf{g}=hg\oplus vh=\mathbf{g}_{\alpha \beta }(u)\mathbf{e}^{\alpha
}\otimes \mathbf{e}^{\beta }=g_{ij}(x,y)dx^{i}\otimes dx^{j}+g_{ab}(x,y)%
\mathbf{e}^{a}\otimes \mathbf{e}^{b}.
\end{equation*}%
For any metric structure $\mathbf{g}$ on a nonholonomic manifold $\left(
\mathbf{V,N}\right) ,$ there are two "preferred" linear connections,
completely and uniquely, defined by
\begin{equation}
\mathbf{g\rightarrow }\left\{
\begin{array}{ccccc}
\nabla : &  & \nabla \mathbf{g=0;\ }^{\nabla }\mathbf{T}=0, &  &
\mbox{ the
Levi--Civita connection;} \\
\widehat{\mathbf{D}}: &  & \widehat{\mathbf{D}}\mathbf{g=0;\ }h\widehat{%
\mathbf{T}}=0,v\widehat{\mathbf{T}}=0, &  &
\mbox{ the canonical
d--connection.}%
\end{array}%
\right.  \label{doublecon}
\end{equation}%
It should be noted that $\nabla $ is not a d--connection because it's
parallel transport does not preserve the horizontal $h$- and the vertical $v$-splitting (\ref%
{whitney}). Nevertheless, there is a unique N--adapted distortion relation
\begin{equation}
\widehat{\mathbf{D}}=\nabla +\widehat{\mathbf{Z}}  \label{distrel}
\end{equation}%
when both linear connections $\widehat{\mathbf{D}}$ and $\nabla $ and the
distorting d--tensor $\widehat{\mathbf{Z}}$ are completely determined by the
metric structure $\mathbf{g}$ for a prescribed N--connection structure $%
\mathbf{N.}$ The Ricci and Riemannian tensors are different for $\widehat{%
\mathbf{D}}$ and $\nabla $ because, in general, $\widehat{\mathbf{T}}\neq 0$
but $\mathbf{\ }^{\nabla }\mathbf{T}=0.$ All geometric constructions with $(%
\mathbf{g},\nabla ;\mathbf{V})$ can be transformed equivalently into similar
ones with $(\mathbf{g},\mathbf{N,D};\mathbf{V}),$ and conversely, if
distortion relations (\ref{distrel}) are utilized.

There are two canonical scalars determined by a d--metric $\mathbf{g}$ via $%
\widehat{\mathbf{D}},$ $\ ^{s}\widehat{\mathbf{R}}:=\mathbf{g}^{\beta \gamma
}\widehat{\mathbf{R}}_{\beta \gamma }$ and the standard (pseudo) Riemannian
scalar determined by $\nabla ,R:=\mathbf{g}^{\beta \gamma }R_{\beta \gamma
}. $ Both values are related by a distortion relation which can be found by
contracting with $\mathbf{g}^{\beta \gamma }$ nonholonomic deformations of
the Ricci tensor, $\widehat{\mathbf{R}}ic=Ric+\widehat{\mathbf{Z}}ic,$ which
are computed by substituting (\ref{distrel}) in formulae (\ref{dtordcurv}).

\subsection{Nonholonomic jets and N--adapted manifolds and maps}

Nonholonomic jet structures can be introduced even if the prime and target
manifolds are considered only with holonomic distributions. In a more
general context, all maps and manifolds can be nonholonomic.

\subsubsection{Nonholonomic maps of holonomic manifolds}

Let us consider two holonomic manifolds $M$ and $Q$ and introduce the set of
nonholonomic 1--jets $\mathbf{J}^{1}(M,Q):=J^{1}(M,Q)$ for $r=1.$\footnote{%
In \cite{kures01}, this is written $\tilde{J}^{1}(M,Q)$ instead of boldface $%
\mathbf{J}^{1}(M,Q).$ As we mentioned above, we use boldface letters in
order to emphasize horizontal $h$- and vertical $v$--splittings via a N--connection structure
of a class of geometric objects/ maps / spaces. We can consider such
decompositions from the maps defining a jet structure (and write $\mathbf{J}%
^{1}$) even when the respective prime and target manifold are holonomic ones,
when $M$ and $Q$ are not boldface.} By induction, we can consider the source
projection $\alpha :\mathbf{J}^{r-1}(M,Q)\rightarrow M$ and the target
projection $\beta :\mathbf{J}^{r-1}(M,Q)\rightarrow Q$ as the target projection
of $(r-1)$--th nonholonomic jets.

\begin{definition}
An $\mathcal{X}\in $ $\mathbf{J}^{r}(M,Q)$ is said to be a nonholonomic $r$%
--jet with the source $x\in M$ and the target $y\in Q$ if there is a local
section $\sigma :M\rightarrow \mathbf{J}^{r}(M,Q)$ such that $\mathcal{X}=%
\mathbf{j}_{x}^{1}\sigma $ and $\beta (\sigma (x))=y$
\end{definition}

We write $\mathcal{X}=\mathbf{j}_{x}^{1}\sigma $ (with calligraphic $%
\mathcal{X}$) instead of $X=\mathbf{j}_{x}^{1}\sigma $ from Definition \ref%
{drjet} in order to emphasize that the jet map is defined, in general, in
nonholonomic form. There is a natural embedding $J^{r}(M,Q)\subset \mathbf{J}%
^{r}(M,Q).$ In general, any $\mathcal{X}$ induces a nonholonomic map $\mu
\mathcal{X}:(\underbrace{TT\ldots T}_{r-\mbox{times}}\ M)_{x}\rightarrow (%
\underbrace{TT\ldots T}_{r-\mbox{times}}\ Q)_{y},$ \cite{kures01}.

\subsubsection{Nonholonomic maps of nonholonomic manifolds}

We can generalize the constructions with nonholonomic jets by considering that
the geometric objects and transforms are defined by equivalence classes of
smooth maps between two nonholonomic manifolds $\mathbf{V},\dim \mathbf{V}%
=n+n,$ and $\mathbf{Q},\dim \mathbf{Q}=m+m,$ and such maps are represented
by Taylor polynomials in certain N--adapted local frames. Other types of
nonholonomic geometric models can  also be elaborated on in a similar manner for
N--adapted maps of type $\mathbf{V\rightarrow V}^{\prime },$ \ where $\dim
\mathbf{V}=n+m$ and $\dim \mathbf{V}^{\prime }=n^{\prime }+m^{\prime }$ and
there are defined N--connection decompositions $T\mathbf{V}=h\mathbf{V\oplus
}v\mathbf{V}$ and $T\mathbf{V}^{\prime }=h^{\prime }\mathbf{V}^{\prime }%
\mathbf{\oplus }v^{\prime }\mathbf{V}^{\prime }$ of type (\ref{whitney})
with corresponding mappings $\mathbf{N\rightarrow N}^{\prime }.$

\begin{definition}
An $\mathcal{X}\in $ $\mathbf{J}^{r}(\mathbf{V},\mathbf{V}^{\prime })$ is
said to be a complete nonholonomic $r$--jet with the source $\mathbf{u}\in
\mathbf{V}$ and the target $\mathbf{u}^{\prime }\in \mathbf{V}^{\prime }$ if
there is a local section $\sigma :\mathbf{V}\rightarrow \mathbf{J}^{r}(%
\mathbf{V},\mathbf{V}^{\prime })$ such that $\mathcal{X}=\mathbf{j}%
_{u}^{1}\sigma $ and $\beta (\sigma (\mathbf{u}))=\mathbf{u}^{\prime }%
\mathbf{.}$
\end{definition}

For simplicity, we use the same nonholonomic jet symbol $\mathcal{X}=\mathbf{%
j}_{\mathbf{u}}^{1}\sigma $ with boldface point $\mathbf{u}\in \mathbf{V.}$
There are also defined natural embeddings $J^{r}(V,V^{\prime })\subset
\mathbf{J}^{r}(V,V^{\prime })\subset \mathbf{J}^{r}(\mathbf{V},\mathbf{V}%
^{\prime }),$ that can be parameterized by local coordinate and/or
N--adapted frame systems and integrable or non-integrable maps. In general,
any $\mathcal{X}$ induces a nonholonomic map $\mu \mathcal{X}:(\underbrace{%
TT\ldots T}_{r-\mbox{times}}\ \mathbf{V})_{\mathbf{u}}\rightarrow (%
\underbrace{T^{\prime }T^{\prime }\ldots T^{\prime }}_{r-\mbox{times}}\
\mathbf{V}^{\prime })_{\mathbf{u}^{\prime }}$ that splits into horizontal and
vertical components with $h,v,...\rightarrow h^{\prime },v^{\prime },...$

We can generalize the concept of jet\ prolongation of fibered manifold, see
Definition \ref{dprolj}, to cases with nonholonomic maps and to prime and
target nonholonomic manifolds. Let $\mathbf{p}:\mathbf{Q}\rightarrow \mathbf{%
V}$ be a fibered manifold when, in general, both $\mathbf{Q}$ and $\mathbf{V}
$ are with nontrivial N--connection structures.

\begin{definition}
A nonholonomic map $\mathbf{j}^{r}\mathbf{f}:\mathbf{V}\rightarrow \mathbf{J}%
^{r}(\mathbf{V},\mathbf{Q})$ is called a $r$-th jet prolongation of $\mathbf{%
f}:\mathbf{V}\rightarrow \mathbf{Q}.$ The set $\mathbf{J}^{r}\mathbf{Q}$ of
all $r$--jets of the local sections of $\mathbf{Q}$ is called the $r$--th
jet prolongation of $\mathbf{Q}$ and $\mathbf{J}^{r}\mathbf{Q}\subset
\mathbf{J}^{r}(\mathbf{V},\mathbf{Q})$ is a closed submanifold.
\end{definition}

We note that if $\mathbf{Q}\rightarrow \mathbf{V}$ is a distinguished vector
bundle with nonholonomic base and nonholonomic total spaces, then $\mathbf{J}%
^{r}\mathbf{Q}$ is a also a distinguished vector bundle.

\subsubsection{Local expressions and $h$- $v$-coordinates}

In order to construct exact solutions in explicit forms by using geometric
methods, it is important to use certain local coordinate and N--adapted
constructions even if geometric models are (locally and/or globally)
intrinsically formulated. Let us establish the necessary conventions:\ We use $%
x=\{x^{i}\}$ as local coordinates on a prime manifold $M,$ when $%
i,j,...=1,2,...n.$ We use $y=\{y^{a}\}$ as local coordinates on a target
manifold $Q,$ when $a=n+1,...n+m.$ On $J^{r}(M,Q),$ our local
coordinates are $x^{i},y^{a}$ and the induced coordinates $v_{i_{1}...i_{p}}^{a}$
are symmetric on the low indices $i_{1,}i_{2},...i_{p}=1,...,n,$ for $p=1,...,r.$
Working with nonholonomic jet spaces $\mathbf{J}^{r}(M,Q)$ for the same
prime and target manifolds we use boldface induced coordinates $\mathbf{v}%
_{i_{1}...i_{p}}^{a}$ which are not symmetric on $i_{1,}i_{2},...i_{p}.$ We
can consider corresponding coordinate systems with the same coordinate
description for any $J^{r}Y,\mathbf{J}^{r}Y$ or $\mathbf{J}^{r}\mathbf{Y}.$

Let us introduce parameterizations for indices and coordinates of N--adapted
maps $\mathbf{V\rightarrow V}^{\prime },$ when $\mathbf{u=(x,y)}=\{u^{\alpha
}=(x^{i},y^{a})\}$ are local coordinates on $\mathbf{V}$ and $\mathbf{u}%
^{\prime }\mathbf{=(x}^{\prime }\mathbf{,y}^{\prime }\mathbf{)}=\{u^{\alpha
^{\prime }}=(x^{i^{\prime }},y^{a^{\prime }})\}$ are local coordinates on $%
\mathbf{V}^{\prime }\mathbf{.}$ We write $\partial _{\check{\alpha}}f:=\frac{%
\partial ^{|\check{\alpha}|}f}{(\partial u^{1})^{\alpha _{1}}...(\partial
u^{n+m})^{\alpha _{n+m}}}$ for the partial derivative of a function $f:%
\mathbf{U}\subset \mathbb{R}^{n+m}\rightarrow \mathbb{R},$ with a \textbf{%
multi-index} $\check{\alpha}$ of range $n+m,$ which is a $(n+m)$--tuple $%
\check{\alpha}=(\alpha _{1},...,\alpha _{n+m})$ of non-negative integers.
For such nonholonomic spaces, we write $|\check{\alpha}|=\alpha
_{1}+...+\alpha _{n+m},$ with $\check{\alpha}!=\check{\alpha}_{1}!\check{%
\alpha}_{2}!...\check{\alpha}_{n+m}!,$ $0!=1,$ and $u^{\check{\alpha}%
}=(u^{1})^{_{1}}...(u^{n+m})^{\alpha _{n+m}}$ for $%
u=(x^{1},x^{2},...,x^{n};y^{n+1},...,y^{n+m})\in \mathbb{R}^{n+m}.$

The local coordinate system is conventionally split on both nonholonomic
manifolds. Two N--adapted maps $^{1}\mathbf{f:V\rightarrow V}^{\prime }$ and
$^{2}\mathbf{f:V\rightarrow V}^{\prime }$ satisfy $\mathbf{j}_{\mathbf{u}%
}^{r}\mathbf{\ }^{1}\mathbf{f}=j_{\mathbf{u}}^{r}$\textbf{\ }$^{2}\mathbf{f}$
if and only if,  for the curves $u^{\alpha }=\zeta ^{\alpha }t$ with arbitrary
$\zeta ^{\alpha }$,
\begin{equation*}
\sum\limits_{|\check{\alpha}|=k}(\partial _{\check{\alpha}}\mathbf{\ }^{1}%
\mathbf{f}^{\alpha ^{\prime }}(u))\zeta ^{\check{\alpha}}=\sum\limits_{|%
\check{\alpha}|=k}(\partial _{\check{\alpha}}\mathbf{\ }^{2}\mathbf{f}%
^{\alpha ^{\prime }}(u))\zeta ^{\check{\alpha}},\mbox{ for }k=0,1,...,r.
\end{equation*}

We  define jet distinguished groups, d--groups,  with elements $%
L_{n+m,n^{\prime }+m^{\prime }}^{r}$ identified with the $r$-th order Taylor
expansions of the generating maps. These are $n+m$-tuples of polynomials of
degree $r$ in $n^{\prime }+m^{\prime }$ variables without absolute term,
with a polynomial representative of a $r$--jet which can written in the
form $\sum\limits_{1\leq |\check{\alpha}|\leq r}\zeta _{\check{\alpha}%
}^{\alpha ^{\prime }}u^{\check{\alpha}}.$ $L_{n+m,n^{\prime }+m^{\prime
}}^{r}$ in a numerical space of the variables $\zeta _{\check{\alpha}%
}^{\alpha ^{\prime }}.$ A standard combinatoric calculus gives
\begin{equation*}
\dim L_{n+m,n^{\prime }+m^{\prime }}^{r}=(n^{\prime }+m^{\prime })\left[
\left(
\begin{array}{c}
n+m+r \\
n+m%
\end{array}%
\right) -1\right] .
\end{equation*}%
In explicit form, the coordinates on $L_{n+m,n^{\prime }+m^{\prime }}^{r}$
are denoted by $\zeta _{\check{\alpha}}^{\alpha ^{\prime }},\zeta _{\check{%
\alpha}\check{\beta}}^{\alpha ^{\prime }},...,\zeta _{\check{\alpha}_{1}...%
\check{\alpha}_{r}}^{\alpha ^{\prime }}$ which are symmetric in all
subscripts if such values are taken in natural coordinate frames. The
projection $\pi _{s}^{r}:$ $L_{m,n}^{r}\rightarrow L_{m,n}^{s}$ consists in
suppressing all terms of degree $>s.$

The set of all invertible elements of $L_{n+m,n^{\prime }+m^{\prime }}^{r}$
with the jet composition is a Lie d--group $G_{n+m}^{r}$ called the $r$-th
differential d--group of the $r$-th jet d--group in dimension $n+m.$ For $%
r=1,$ the group $G_{n+m}^{1}$ is identified with a nonholonomic group
decomposition $GL(n+m,\mathbb{R})\rightarrow GL(n,\mathbb{R})\oplus GL(m,%
\mathbb{R})$ corresponding to a horizontal (h) and vertical (v) splitting (\ref{whitney}).

In this work, we study nonholonomic jet prolongations of the geometric
objects from section \ref{ssnm} in $\mathbf{J}^{r}(\mathbf{V},\mathbf{V}%
^{\prime })$--framework with local coordinates
\begin{equation}
u^{\alpha _{s}}=(x^{i},y^{a},\zeta _{\check{\alpha}_{1}...\check{\alpha}%
_{r}}^{\alpha ^{\prime }})=(x^{i},y^{a},\zeta ^{a_{s}}).  \label{standjcoord}
\end{equation}%
We use the label $s$ in order to perform a conventional splitting of
dimensions, $\dim \ ^{s}V=4+2s=2+2+...+2\geq 4;s\geq 0$ for conventional
finite dimensional (pseudo) Riemannian space$\ ^{s}V.$ The jet coordinates $%
v_{\check{\alpha}_{1}...\check{\alpha}_{r}}^{\alpha ^{\prime }}$ are
re--grouped in oriented two shells \footnote{%
In a similar form, we can split odd dimensions, for instance, $\dim
V=3+2+...+2.$} which allows us to apply the AFDM and to construct exact
solutions for generalized Einstein equations and metrics $\ ^{s}\mathbf{g}$
with arbitrary signatures $(\pm 1,\pm 1,\pm 1,...\pm 1).$ Such shells are
determined by nonholonomic data which transforms into $\zeta _{\check{\alpha}%
_{1}...\check{\alpha}_{r}}^{\alpha ^{\prime }}$ with symmetric lower indices
if the constructions are performed in coordinate bases. Let us establish
conventions on (abstract) indices and coordinates $u^{\alpha
_{s}}=(x^{i_{s}},y^{a_{s}}),$ for $s=0,1,2,3,....$ labellings of the oriented
number of two dimensional, 2-d, "shells" added to a 4--d spacetime. For $s=0$
(in a conventional form), we write $u^{\alpha }=(x^{i},y^{a})$ and consider
the following local systems of coordinates: {\small
\begin{eqnarray}
s &=&1:u^{\alpha _{1}}=(x^{\alpha }=u^{\alpha
},v^{a_{1}})=(x^{i},y^{a},\zeta ^{a_{1}}),  \label{jcoord} \\
\ s &=&2:u^{\alpha _{2}}=(x^{\alpha _{1}}=u^{\alpha
_{1}},v^{a_{2}})=(x^{i},y^{a},\zeta ^{a_{1}},\zeta ^{a_{2}}),  \notag \\
\ s &=&3:u^{\alpha _{3}}=(x^{\alpha _{2}}=u^{\alpha
_{2}},v^{a_{3}})=(x^{i},y^{a},\zeta ^{a_{1}},\zeta ^{a_{2}},\zeta
^{a_{3}}),...  \notag
\end{eqnarray}%
} for $i,j,...=1,2;a,b,...=3,4;a_{1},b_{1}...=5,6;a_{2},b_{2}...=7,8;$ $%
a_{3},b_{3}...=9,10,...$ and $i_{1},j_{1},...=1,2,3,4;i_{2},$ $j_{2},...$ $%
=1,2,3,4,5,6;\ i_{3},j_{3},...=1,2,3,4,5,6,7,8;...$ In compact notation, we  write
$u=(x,y);$ $\ ^{1}u=(u,\ ^{1}\zeta )=(x,y,\ ^{1}\zeta ),\ ^{2}u=(\ ^{1}u,\
^{2}\zeta )=(x,y,\ ^{1}\zeta ,\ ^{2}\zeta ),...$

We underline the indices in order to emphasize that certain values are
with respect to local coordinate bases. The transformations between local
frames, $e_{\alpha _{s}},$ and coordinate frames, $\partial _{\underline{%
\alpha }_{s}}=\partial /\partial u^{\underline{\alpha }_{s}}$ on $\ ^{s}V$
are written as $\ e_{\alpha _{s}}=e_{\ \alpha _{s}}^{\underline{\alpha }_{s}}(\
^{s}u)\partial /\partial u^{\underline{\alpha }_{s}}.$ General
parameterizations of coefficients $e_{\ \alpha _{s}}^{\underline{\alpha }%
_{s}}$ give nonholonomy relations $e_{\alpha _{s}}e_{\beta
_{s}}-e_{\beta _{s}}e_{\alpha _{s}}=W_{\alpha _{s}\beta _{s}}^{\gamma
_{s}}e_{\gamma _{s}}.$ The nonholonomy coefficients $W_{\alpha _{s}\beta
_{s}}^{\gamma _{s}}=W_{\beta _{s}\alpha _{s}}^{\gamma _{s}}(u)$ vanish for
holonomic configurations. Using the condition $e^{\alpha _{s}}\rfloor
e_{\beta _{s}}=\delta _{\beta _{s}}^{\alpha _{s}},$ where the 'hook'
operator $\rfloor $ corresponds to the inner derivative and $\delta _{\beta
_{s}}^{\alpha _{s}}$ is the Kronecker symbol, we  construct dual frames, $%
e^{\alpha _{s}}=e_{\ \underline{\alpha }_{s}}^{\ \alpha _{s}}(\ ^{s}u)du^{%
\underline{\alpha }_{s}}.$

It is important to distinguish the partial derivatives on spacetime
coordinates (for instance, $\partial _{i}=\partial /\partial x^{i},\partial
_{a}=\partial /\partial y^{a}$ and $\partial _{\alpha }=\partial /\partial
u^{\alpha })$ and on $r$--jet variables, when $\eth _{a_{s}}=\partial /\zeta
^{a_{s}}$ is used for a $2+2+...$conventinal splitting of partial
derivatives $\partial /\partial \zeta _{\check{\alpha}_{1}...\check{\alpha}%
_{r}}^{\alpha ^{\prime }}.$ In some sense, $\zeta _{\check{\alpha}_{1}...%
\check{\alpha}_{r}}^{\alpha ^{\prime }}$ can be considered as extra
dimension coordinates but with certain additional Lie d--group properties of
$G_{n+m}^{r}$ considered above.

\subsection{Jet prolongation of Ricci soliton and Einstein equations}

We can define canonical N--connection, frame, metric and distinguished
metric sturcures on $\mathbf{J}^{r}(\mathbf{V},\mathbf{V}^{\prime })$
determined by prolongations of respective prime objects on $\mathbf{V,}$ see
Definition \ref{dprolj}.

\subsubsection{Shell parameterized N--connection and associated frame
structures}

A map $j^{r}f:M\rightarrow J^{r}(M,Q)$ is called a $r$-th jet prolongation
of $f:M\rightarrow Q.$ The set $J^{r}Q$ of all $r$--jets of the local
sections of $Y$ is called the $r$--th jet prolongation of $Q$ and $%
J^{r}Q\subset J^{r}(M,Q)$ is a closed submanifold.

\begin{theorem}
Any N--connection structure $\mathbf{N}$ on $\mathbf{V}$ determines a $r$%
--th jet prolongation of N--connection $\ ^{s}\mathbf{N}$ on $\mathbf{J}^{r}(%
\mathbf{V},\mathbf{V}^{\prime })$ as the  Whitney sum
\begin{equation}
\ ^{s}\mathbf{N}:T\ ^{s}\mathbf{V}=h\mathbf{V}\oplus v\mathbf{V}\oplus \
^{1}v\mathbf{V}\oplus \ ^{2}v\mathbf{V}\oplus ...\oplus \ ^{s}v\mathbf{V},
\label{ncshell}
\end{equation}%
for a conventional horizontal (h) and vertical (v) "shell by shell"
splitting.
\end{theorem}

\begin{proof}
It is a natural construction when the coefficients of N--connection are
defined by jet prolongations and parameterized as $\ ^{s}\mathbf{N}%
=N_{i_{s}}^{a_{s}}(\ ^{s}u)dx^{i_{s}}\otimes \partial /\partial \zeta
^{a_{s}}$ on every chart on $\mathbf{J}^{r}(\mathbf{V},\mathbf{V}^{\prime
}), $ i.e. for $\ ^{s}\mathbf{V}.$

$\square $ (end proof).
\end{proof}

\vskip5pt

Using the coefficients of N--connection, we prove the following.

\begin{corollary}
$r$--th jet prolongations induce on $\mathbf{J}^{r}(\mathbf{V},\mathbf{V}%
^{\prime })$ a system of N-elongated bases/ partial derivatives, $\mathbf{e}%
_{\nu _{s}}=(\mathbf{e}_{i_{s}},e_{a_{s}}),$ and cobases, N--adapted
differentials, $\mathbf{e}^{\mu _{s}}=(e^{i_{s}},\mathbf{e}^{a_{s}}).$
\end{corollary}

\begin{proof}
Taking (\ref{dderdif}) for $\mathbf{V,}$ we prolongate on $s\geq 1$ shells,
\begin{eqnarray}
\mathbf{e}_{i_{s}} &=&\frac{\partial }{\partial x^{i_{s}}}-\
N_{i_{s}}^{a_{s}}\eth _{a_{s}},\ e_{a_{s}}=\eth _{a_{s}}=\frac{\partial }{%
\partial \zeta ^{a_{s}}},  \label{naders} \\
e^{i_{s}} &=&dx^{i_{s}},\mathbf{e}^{a_{s}}=d\zeta ^{a_{s}}+\
N_{i_{s}}^{a_{s}}dx^{i_{s}}.  \label{nadifs}
\end{eqnarray}%
$\square $
\end{proof}

\vskip5pt

The N--adapted operators (\ref{naders}) satisfy  nonholonomy relations:
\begin{equation}
\lbrack \mathbf{e}_{\alpha _{s}},\mathbf{e}_{\beta _{s}}]=\mathbf{e}_{\alpha
_{s}}\mathbf{e}_{\beta _{s}}-\mathbf{e}_{\beta _{s}}\mathbf{e}_{\alpha
_{s}}=W_{\alpha _{s}\beta _{s}}^{\gamma _{s}}\mathbf{e}_{\gamma _{s}},
\label{anhrel1}
\end{equation}%
when $W_{i_{s}a_{s}}^{b_{s}}=\partial _{a_{s}}N_{i_{s}}^{b_{s}}$ and $%
W_{j_{s}i_{s}}^{a_{s}}=\ ^{J}N_{i_{s}j_{s}}^{a_{s}},$ where the Neijenhuis
tensor, i.e. the curvature of the $r$--th jet prolongation of N--connection,
is $\ ^{J}N_{i_{s}j_{s}}^{a_{s}}=\mathbf{e}_{j_{s}}\left(
N_{i_{s}}^{a_{s}}\right) -\mathbf{e}_{i_{s}}\left( N_{j_{s}}^{a_{s}}\right)$.

\subsubsection{N--adapted shell prolongation of d--connections}

On $\mathbf{J}^{r}(\mathbf{V},\mathbf{V}^{\prime })$ with prolongation of
geometric objects from $\mathbf{V,}$ we define linear connection structures
in N--adapted form in the following.

\begin{theorem}
\textbf{--Definition: } There are distinguished connection, d--connection,
structures, $\ ^{s}\mathbf{D=\{D}_{\alpha _{s}}\},$ with $\mathbf{D}%
=(hD;vD),\ ^{1}\mathbf{D=}(\ ^{1}hD;\ ^{1}vD),...,\ ^{s-1}\mathbf{D=}(\
^{s-2}hD;\ ^{s-1}vD),\ ^{s}\mathbf{D=}(\ ^{s-1}hD;\ ^{s}vD)$, preserving
under parallelism the N--connection splitting (\ref{ncshell}) and
coefficients computed with respect to N--adapted bases (\ref{naders}) and (%
\ref{nadifs}).
\end{theorem}

\begin{proof}
We shall use, for instance, the term \textbf{Theorem--Definition} when a
geometric object is defined by an explicit construction which consists also of
a proof of a Theorem. In this case, we prove the existence of a
d--connection by considering N--adapted covariant derivatives {\small
\begin{eqnarray*}
\mathbf{D}_{\alpha } &=&(D_{i};D_{a}),\mathbf{D}_{\alpha _{1}}=(\
^{1}D_{\alpha };D_{a_{1}}),\ \mathbf{D}_{\alpha _{2}}=(\ ^{2}D_{\alpha
_{1}};D_{a_{2}}),...,\mathbf{D}_{\alpha _{s}}=(\ ^{s}D_{\alpha
_{s-1}};D_{a_{s}}), \mbox{ for } \\
hD &=&(L_{jk}^{i},L_{bk}^{a}),vD=(C_{jc}^{i},C_{bc}^{a}), \ ^{1}hD =
(L_{\beta \gamma }^{\alpha },L_{b_{1}\gamma }^{a_{1}}),\ ^{1}vD=(C_{\beta
c_{1}}^{\alpha },C_{b_{1}c_{1}}^{a_{1}}),\ \ ^{2}hD=(L_{\beta _{1}\gamma
_{1}}^{\alpha _{1}},L_{b_{2}\gamma _{1}}^{a_{2}}), \\
&& \ ^{2}vD=(C_{\beta _{1}c_{2}}^{\alpha _{1}},C_{b_{2}c_{2}}^{a_{2}}), ...,
\ ^{s}hD = (L_{\beta _{s-1}\gamma _{s-1}}^{\alpha _{s-1}},L_{b_{s}\gamma
_{s-1}}^{a_{s}}),\ ^{s}vD=(C_{\beta _{s-1}c_{s}}^{\alpha
_{s-1}},C_{b_{s}c_{s}}^{a_{s}}),
\end{eqnarray*}
} when the coefficients{
\begin{eqnarray}
\mathbf{\Gamma }_{\ \beta \gamma }^{\alpha }
&=&(L_{jk}^{i},L_{bk}^{a};C_{jc}^{i},C_{bc}^{a}), \mathbf{\Gamma }_{\ \beta
_{1}\gamma _{1}}^{\alpha _{1}} =(L_{\beta \gamma }^{\alpha },L_{b_{1}\gamma
}^{a_{1}};C_{\beta c_{1}}^{\alpha },C_{b_{1}c_{1}}^{a_{1}}),\   \label{coefd}
\\
\mathbf{\Gamma }_{\ \beta _{2}\gamma _{2}}^{\alpha _{2}} &=& (L_{\beta
_{1}\gamma _{1}}^{\alpha _{1}},L_{b_{2}\gamma _{1}}^{a_{2}};C_{\beta
_{1}c_{2}}^{\alpha _{1}},C_{b_{2}c_{2}}^{a_{2}}),..., \mathbf{\Gamma }_{\
\beta _{s}\gamma _{s}}^{\alpha _{s}} = (L_{\beta _{s-1}\gamma
_{s-1}}^{\alpha _{s-1}},L_{b_{s}\gamma _{s-1}}^{a_{s}};C_{\beta
_{s-1}c_{s}}^{\alpha _{s-1}},C_{b_{s}c_{s}}^{a_{s}})  \notag
\end{eqnarray}%
} of such a d--connection $\ ^{s}\mathbf{D=\{D}_{\alpha _{s}}\mathbf{\}}$
are computed in N--adapted form with respect to the frames (\ref{naders})--(%
\ref{nadifs}) following equations $\mathbf{D}_{\alpha _{s}}\mathbf{e}_{\beta
_{s}}=\mathbf{\Gamma }_{\ \beta _{s}\gamma _{s}}^{\alpha _{s}}\mathbf{e}%
_{\gamma _{s}}.$

$\square $
\end{proof}

\vskip5pt

It is possible always to consider such frame transforms when all shell
frames are N-adapted and $\ ^{1}D_{\alpha }=\mathbf{D}_{\alpha },\
^{2}D_{\alpha _{1}}=\mathbf{D}_{\alpha _{1}},...\ ^{s}D_{\alpha _{s-1}}=%
\mathbf{D}_{\alpha _{s-1}}$.

\begin{corollary}
\textbf{--Definition:} There are natural $r$--th jet prolongations of the
torsion and curvature d--tensors (\ref{dtordcurv}) defined on a prime $%
\mathbf{V}$ and elongated in N--adapted form on $\mathbf{J}^{r}(\mathbf{V},%
\mathbf{V}^{\prime })$ with prescribed shell splitting on $\ ^{s}\mathbf{V,}$
\begin{eqnarray}
\ ^{s}\mathbf{T}(\mathbf{X,Y)} &:=&\ ^{s}\mathbf{D}_{\mathbf{X}}\mathbf{Y-}\
^{s}\mathbf{D}_{\mathbf{Y}}\mathbf{X}+[\mathbf{X,Y}]\mbox{ \ and \ }
\label{torsshell} \\
\ ^{s}\mathbf{R}(\mathbf{X,Y)} &:=&\ ^{s}\mathbf{D}_{\mathbf{X}}\ ^{s}%
\mathbf{D}_{\mathbf{Y}}-\ ^{s}\mathbf{D}_{\mathbf{Y}}\ ^{s}\mathbf{D}_{%
\mathbf{X}}\mathbf{-}\ ^{s}\mathbf{D}_{[\mathbf{X,Y}]},  \label{curvshell}
\end{eqnarray}%
for any d--vectors $\mathbf{X,Y\subset }T\ ^{s}\mathbf{V.}$
\end{corollary}

\begin{proof}
To perform computations in N--adapted--shell form we  consider a
differential connection 1--form $\mathbf{\Gamma }_{\ \beta _{s}}^{\alpha
_{s}}=\mathbf{\Gamma }_{\ \beta _{s}\gamma _{s}}^{\alpha _{s}}\mathbf{e}%
^{\gamma _{s}}$ and elaborate a differential form calculus with respect to
skew symmetric tensor products of N--adapted frames (\ref{naders})--(\ref%
{nadifs}). Respectively, the torsion $\mathcal{T}^{\alpha _{s}}=\{\mathbf{T}%
_{\ \beta _{s}\gamma _{s}}^{\alpha _{s}}\}$ and curvature $\mathcal{R}%
_{~\beta _{s}}^{\alpha _{s}}=\{\mathbf{\mathbf{R}}_{\ \ \beta _{s}\gamma
_{s}\delta _{s}}^{\alpha _{s}}\}$ d--tensors of $\ \ ^{s}\mathbf{D}$ are
computed
\begin{eqnarray}
&&\mathcal{T}^{\alpha _{s}}:=\ ^{s}\mathbf{De}^{\alpha _{s}}=d\mathbf{e}%
^{\alpha _{s}}+\mathbf{\Gamma }_{\ \beta _{s}}^{\alpha _{s}}\wedge \mathbf{e}%
^{\beta _{s}},\   \label{dt} \\
&&\mathcal{R}_{~\beta _{s}}^{\alpha _{s}}:=\ ^{s}\mathbf{D\Gamma }_{\ \beta
_{s}}^{\alpha _{s}}=d\mathbf{\Gamma }_{\ \beta _{s}}^{\alpha _{s}}-\mathbf{%
\Gamma }_{\ \beta _{s}}^{\gamma _{s}}\wedge \mathbf{\Gamma }_{\ \gamma
_{s}}^{\alpha _{s}}=\mathbf{R}_{\ \beta _{s}\gamma _{s}\delta _{s}}^{\alpha
_{s}}\mathbf{e}^{\gamma _{s}}\wedge \mathbf{e}^{\delta _{s}},  \label{dc}
\end{eqnarray}%
see Refs. \cite{vex3,vexsol2} for explicit calculations of coefficients $\mathbf{R%
}_{\ \beta _{s}\gamma _{s}\delta _{s}}^{\alpha _{s}}$ in higher dimensions.
The formulae in the jet shell adapted coordinates (\ref{jcoord}) are very
similar to those in N--adapted bases for extra dimensional (pseudo)
Riemannian spaces. In standard $r$--jet coordinates (\ref{standjcoord}) for $%
\mathbf{J}^{r}(\mathbf{V},\mathbf{V}^{\prime }),$ $u^{\underline{\alpha }%
_{s}}=(x^{i},y^{a},\zeta _{\check{\alpha}_{1}...\check{\alpha}_{r}}^{\alpha
^{\prime }}),$ additional contraction of up-down indices and symmetrization
lead to very cumbersome coefficient formulae.

$\square $
\end{proof}

\vskip5pt

In Appendix \ref{sscoefcurv}, we present two Theorems on computing
N--adapted coefficients formulas for $\mathbf{T}_{\ \beta _{s}\gamma
_{s}}^{\alpha _{s}}$ and $\mathbf{\mathbf{R}}_{\ \ \beta _{s}\gamma
_{s}\delta _{s}}^{\alpha _{s}}.$

\subsubsection{Jet prolongation of d--metrics}

On $\ ^{s}\mathbf{V,}$ a metric tensor can be written in the form
\begin{equation}
\ ^{s}\mathbf{g} = g_{\alpha _{s}\beta _{s}}e^{\alpha _{s}}\otimes e^{\beta
_{s}}=g_{\underline{\alpha }_{s}\underline{\beta }_{s}}du^{\underline{\alpha
}_{s}}\otimes du^{\underline{\beta }_{s}} = g_{\underline{\alpha }\underline{%
\beta }}du^{\underline{\alpha }}\otimes du^{\underline{\beta }}+g_{%
\underline{\alpha }_{s+1}\underline{\beta }_{s+1}}d\zeta ^{\underline{\alpha
}_{s+1}}\otimes d\zeta ^{\underline{\beta }_{s+1}}, \ s=0,1,2,...,
\label{metr}
\end{equation}%
where $du^{\underline{\alpha }}\in T^{\ast }\mathbf{V}$ and the indices are
underlined in order to emphasize  coordinate dual bases. The
coefficients of such a metric are subject to frame transform rules, $%
g_{\alpha _{s}\beta _{s}}=e_{\ \alpha _{s}}^{\underline{\alpha }_{s}}e_{\
\beta _{s}}^{\underline{\beta }_{s}}g_{\underline{\alpha }_{s}\underline{%
\beta }_{s}}$, which can be respectively generalized for any tensor object.
We can not preserve a $2+2+2+...$ splitting of the dimensions under general
frame/coordinate transforms.

\begin{lemma}
Any metric structure $\ ^{s}\mathbf{g=\{g}_{\alpha _{s}\beta _{s}}\mathbf{\}}
$ on $\ ^{s}\mathbf{V}$ can be written as a distinguished metric (d--metric)
\begin{eqnarray}
\ \ ^{s}\mathbf{g} &=&\ g_{i_{s}j_{s}}(\ ^{s}u)\ e^{i_{s}}\otimes
e^{j_{s}}+\ g_{a_{s}b_{s}}(\ ^{s}u)\mathbf{e}^{a_{s}}\otimes \mathbf{e}%
^{b_{s}}  \label{dm} \\
&=&g_{ij}(x)\ e^{i}\otimes e^{j}+g_{ab}(u)\ \mathbf{e}^{a}\otimes \mathbf{e}%
^{b}+g_{a_{1}b_{1}}(\ ^{1}u)\ \mathbf{e}^{a_{1}}\otimes \mathbf{e}%
^{b_{1}}+....+\ g_{a_{s}b_{s}}(\ ^{s}u)\mathbf{e}^{a_{s}}\otimes \mathbf{e}%
^{b_{s}}.  \notag
\end{eqnarray}
\end{lemma}

\begin{proof}
Using frame/coordinate transforms, we can always parameterize any metric (\ref{metr}%
) in such form:%
\begin{equation*}
\ \ \underline{g}_{\alpha \beta }\left( \ u\right) =\left[
\begin{array}{cc}
\ g_{ij}+\ h_{ab}N_{i}^{a}N_{j}^{b} & h_{ae}N_{j}^{e} \\
\ h_{be}N_{i}^{e} & \ h_{ab}%
\end{array}%
\right] ,
\end{equation*}%
on the prime manifold and, for $r$--prolongations,
\begin{eqnarray*}
\underline{g}_{\alpha _{1}\beta _{1}}\left( \ ^{1}\zeta \right) &=&\left[
\begin{array}{cc}
\ \underline{g}_{\alpha \beta } & h_{a_{1}e_{1}}N_{\beta _{1}}^{e_{1}} \\
\ h_{b_{1}e_{1}}N_{\alpha _{1}}^{e_{1}} & \ h_{a_{1}b_{1}}%
\end{array}%
\right] ,\ \ \ \underline{g}_{\alpha _{2}\beta _{2}}\left( \ ^{2}\zeta
\right) =\left[
\begin{array}{cc}
\ \underline{g}_{\alpha _{1}\beta _{1}} & h_{a_{2}e_{2}}N_{\beta
_{1}}^{e_{2}} \\
\ h_{b_{2}e_{2}}N_{\alpha _{1}}^{e_{2}} & \ h_{a_{2}b_{2}}%
\end{array}%
\right] ,\ ...\  \\
\underline{g}_{\alpha _{s}\beta _{s}}\left( \ ^{s}\zeta \right) &=&\left[
\begin{array}{cc}
\ g_{i_{s}j_{s}}+\ h_{a_{s}b_{s}}N_{i_{s}}^{a_{s}}N_{j_{s}}^{b_{s}} &
h_{a_{s}e_{s}}N_{j_{s}}^{e_{s}} \\
\ h_{b_{s}e_{s}}N_{i_{s}}^{e_{s}} & \ h_{a_{s}b_{s}}%
\end{array}%
\right] .
\end{eqnarray*}%
By re--grouping terms shell by shell with respect to the bases (\ref{nadifs}), we
obtain (\ref{dm}).

$\square $
\end{proof}

\vskip5pt

In diverse dimensions, such parameterizations are similar to those introduced
in  Kaluza--Klein type theories when $\zeta ^{a_{s}},s\geq 1,$ are considered as
extra dimension coordinates with cylindrical compactification and $N_{\alpha
}^{e_{s}}(\ ^{s}u)\sim A_{a_{s}\alpha }^{e_{s}}(u)y^{\alpha }$ represent
certain (non) Abelian gauge fields $A_{a_{s}\alpha }^{e_{s}}(u).$ Jet
generalized gauge theories possess different symmetries than those with
potentials taking values in the Lie group algebras, see Ref. \cite{aldaya} on
unification of gravity with internal gauge interactions.

\subsubsection{Canonical jet distortions and linear connections on jet
bundles}

For any $r$--jet prolongation (pseudo) Riemannian metric$\ ^{s}\mathbf{g,}$
we can construct in standard form the Levi--Civita connection
(LC--connection), $\ ^{s}\nabla =\{\ _{\shortmid }\Gamma _{\ \beta
_{s}\gamma _{s}}^{\alpha _{s}}\}.$ By definition such a connection is metric
compatible, $\ ^{s}\nabla (\ ^{s}\mathbf{g)}=0,$ and with zero torsion, $\ \
_{\shortmid }T^{\alpha _{s}}=0$ (we  use formulae (\ref{dt}) for $\ ^{s}%
\mathbf{D\rightarrow }\ ^{s}\nabla ).$ It should be emphasized that such a
linear connection is not a d--connection because it does not preserve under
general coordinate transforms a N--connection splitting (\ref{ncshell}).

\begin{theorem}
\label{tcandist}There is a canonical distortion relation
\begin{equation}
\ ^{s}\widehat{\mathbf{D}}=\ ^{s}\nabla +\ ^{s}\widehat{\mathbf{Z}},
\label{distorsrel}
\end{equation}%
for a canonical d--connection $\ ^{s}\widehat{\mathbf{D}}$ which is
completely and uniquely defined by a (pseudo) Riemannian metric $\ ^{s}%
\mathbf{g}$ (\ref{dm}) for a chosen nonholonomic distribution $\ ^{s}\mathbf{%
N=\{}N_{i_{s}}^{a_{s}}\}$ when $\ \ ^{s}\widehat{\mathbf{D}}(\ ^{s}\mathbf{g)%
}=0$ and the horizontal and vertical torsions are zero, i.e. $h\widehat{%
\mathbf{T}}=\{\widehat{\mathbf{T}}_{\ jk}^{i}\}=0,$ $v\widehat{\mathbf{T}}=\{%
\widehat{\mathbf{T}}_{\ bc}^{a}\}=0,\ ^{1}v\widehat{\mathbf{T}}=\{\widehat{%
\mathbf{T}}_{\ b_{1}c_{1}}^{a_{1}}\}=0,...,\ ^{s}v\widehat{\mathbf{T}}=\{%
\widehat{\mathbf{T}}_{\ b_{s}c_{s}}^{a_{s}}\}=0;$ the distorting tensor $\
^{s}\widehat{\mathbf{Z}}=\{\widehat{\mathbf{\ Z}}_{\ \beta _{s}\gamma
_{s}}^{\alpha _{s}}\}$ is uniquely defined by the same data  $(^{s}\mathbf{g,}%
\ ^{s}\mathbf{N).}$
\end{theorem}

\begin{proof}
We sketch a proof in appendix \ref{prooftcands}.

$\square $
\end{proof}

\vskip5pt

The N--adapted coefficients of the distortion d--tensor $\widehat{\mathbf{\ Z%
}}_{\ \beta _{s}\gamma _{s}}^{\alpha _{s}}$ are algebraic combinations of $%
\widehat{T}_{\ \beta _{s}\gamma _{s}}^{\alpha _{s}}$ and vanish for zero
torsion. The nonholonomic variables $(\ ^{s}\mathbf{g}$ (\ref{dm})$\mathbf{,}%
\ ^{s}\mathbf{N,}\ ^{s}\widehat{\mathbf{D}})$ are equivalent to the standard
(pseudo) Riemannian ones $(\ ^{s}\mathbf{g}$ (\ref{metr}) $\ ^{s}\nabla ).$
For instance,  GR  in 4-d can be formulated equivalently using the
connection $\nabla $ and/or $\widehat{\mathbf{D}}$ if the distortion
relation (\ref{distrel}) is used, see details in \cite{vexsol1,vex3,vexsol2}%
. \ The $r$--jet prolongations give distortions (\ref{distorsrel}). We
consider nonholonomic jet deformations of a 4--d (pseudo) Riemannian space
to a $\mathbf{J}^{r}(\mathbf{V},\mathbf{V}^{\prime })$ with a canonical
nonzero d--torsion. In such cases, we are able to decouple modified Einstein
equations and construct integral varieties with jet variables. At the end,
we can impose additional nonholonomic constraints and fix the jet
coordinates in order to generate exact solutions of Ricci soliton/ Einstein
equations in 4--d, or higher dimensions, with \ $r$--jet symmetries.

Here we note that $\ ^{s}\nabla $ and $\ ^{s}\widehat{\mathbf{D}}$ are not
tensor objects. $\ ^{s}\widehat{\mathbf{D}}$ is a d--connection and such
linear connections are subject to different rules with respect to coordinate
transformations. It is possible to consider frame transformations with certain $\ ^{s}%
\mathbf{N=\{}N_{i_{s}}^{a_{s}}\}$ when the conditions $\ _{\shortmid }\Gamma
_{\ \alpha _{s}\beta _{s}}^{\gamma _{s}}=\widehat{\mathbf{\Gamma }}_{\
\alpha _{s}\beta _{s}}^{\gamma _{s}}$ are satisfied with respect to some
N--adapted frames (\ref{naders})-- (\ref{nadifs}). In general, $\ ^{s}\nabla
\neq \ ^{s}\widehat{\mathbf{D}}$ and the corresponding curvature tensors $\
_{\shortmid }R_{\ \beta _{s}\gamma _{s}\delta _{s}}^{\alpha _{s}}\neq
\widehat{\mathbf{R}}_{\ \beta _{s}\gamma _{s}\delta _{s}}^{\alpha _{s}}$ are
different, but the Ricci tensor components may coincide for certain classes
of nonholonomic constraints.

\subsubsection{Prolongation of Ricci soliton and Einstein equations on
nonholonomic jet configurations}

\label{ssgoals}In this section, we introduce important geometric and
physical equations in nonholonomic variables on $\mathbf{V}$ and consider
generalizations on $\mathbf{J}^{r}(\mathbf{V},\mathbf{V}^{\prime }).$

\begin{definition}
The geometric data $\left( \mathbf{g},\mathbf{N,D};\mathbf{V}\right) $
defines a gradient nonholonomic Ricci soliton if there exists a smooth
potential function $\kappa (x,y)$ such that%
\begin{equation}
\widehat{\mathbf{R}}_{\beta \gamma }+\widehat{\mathbf{D}}_{\beta }\widehat{%
\mathbf{D}}_{\gamma }\kappa =\lambda \mathbf{g}_{\beta \gamma }.
\label{nriccisol}
\end{equation}%
There are three types of such Ricci solitons determined by a constant $%
\lambda :$ steady ones for $\lambda =0;$ shrinking ones for $\lambda >0;$ and
expanding ones for $\lambda <0.$
\end{definition}

The above classification is determined by the Levi--Civita, LC, limits when
shrinking solutions help us to understand the asymptotic behaviour of the
ancient (old) solutions of the Ricci flow theory \cite{ham1,ham2,perelm}.
By generalizing and adapting the constructions to N--connection structures, one
can describe geometric flows with nonholonomic constraints \cite%
{vnrf,vncrfvnrf}. Here, we omit a study of geometric analysis issues and
generalized Ricci flow models  and restrict our research to
nonholonomic $r$--jet prolongations of equations and important classes of
solutions.

The N--adapted coefficients of the Ricci d--tensor $Ric=\{\mathbf{R}_{\alpha
_{s}\beta _{s}}:=\mathbf{R}_{\ \alpha _{s}\beta _{s}\tau _{s}}^{\tau _{s}}\}$
of a d--connection $\ ^{s}\mathbf{D}$ in $\mathbf{J}^{r}(\mathbf{V},\mathbf{V%
}^{\prime })$ are computed from the curvature
tensor (\ref{dc}),
\begin{equation}
\mathbf{R}_{\alpha _{s}\beta _{s}}=\{R_{i_{s}j_{s}}:=R_{\
i_{s}j_{s}k_{s}}^{k_{s}},\ \ R_{i_{1}a_{1}}:=-R_{\
i_{1}k_{1}a_{1}}^{k_{1}},...,\ R_{a_{s}i_{s}}:=R_{\
a_{s}i_{s}b_{s}}^{b_{s}}\}.  \label{dricci}
\end{equation}%
Using the inverse matrix of $\ ^{s}\mathbf{g}$ (\ref{dm}), we  compute
the scalar curvature of $\ ^{s}\mathbf{D,}$
\begin{equation}
\ ^{s}R:= \mathbf{g}^{\alpha _{s}\beta _{s}}\mathbf{R}_{\alpha _{s}\beta
_{s}}=g^{i_{s}j_{s}}R_{i_{s}j_{s}}+h^{a_{s}b_{s}}R_{a_{s}b_{s}} =R+S+\
^{1}S+...+\ ^{s}S,  \label{rdsc}
\end{equation}%
with respective to the horizontal (h) and vertical (v) components of the scalar curvature, $R=g^{ij}R_{ij},$
$S=h^{ab}R_{ab},$ $\ ^{1}S=h^{a_{1}b_{1}}R_{a_{1}b_{1}},...,\
^{s}S=h^{a_{s}b_{s}}R_{a_{s}b_{s}}.$

The Einstein d--tensor $\ ^{s}\mathbf{E}nst=\{\ ^{s}\mathbf{E}_{\alpha
_{s}\beta _{s}}\}$ for any nonholonomic $r$--jet data $(\ ^{s}\mathbf{g}$ $%
\mathbf{,}\ ^{s}\mathbf{N,}\ ^{s}\mathbf{D})$ is determined in standard
forms as,%
\begin{equation}
\ ^{s}\mathbf{E}_{\alpha _{s}\beta _{s}}:=\ ^{s}\mathbf{R}_{\alpha _{s}\beta
_{s}}-\frac{1}{2}\mathbf{g}_{\alpha _{s}\beta _{s}}\ ^{s}R.  \label{einstdt}
\end{equation}%
Such nonholonomic jet prolongations of a prime Einstein tensor are not
symmetric, and the d--tensor $\ ^{s}\mathbf{R}_{\alpha _{s}\beta _{s}}$ is
not symmetric for a general $\ ^{s}\mathbf{D}$ and $\ ^{s}\mathbf{D(}\ ^{s}%
\mathbf{\mathbf{E}}nst)\neq 0.$ For a canonical $\ ^{s}\widehat{\mathbf{D}}%
\mathbf{,}$ we can always compute $\ ^{s}\widehat{\mathbf{D}}\mathbf{(}\ ^{s}%
\widehat{\mathbf{\mathbf{E}}}nst\mathbf{)}$ as a unique distortion relation
determined by (\ref{distorsrel}).

\begin{proposition}
For a N--connection splitting (\ref{ncshell}) on $\mathbf{J}^{r}(\mathbf{V},%
\mathbf{V}^{\prime }),$ the  Levi--Civita and canonical
d--connection with prolongations are defined by the conditions
\begin{equation*}
\ ^{s}\mathbf{g\rightarrow }\left\{
\begin{array}{ccccc}
\ ^{s}\nabla : &  & \ ^{s}\nabla \ ^{s}\mathbf{g=0;\ }^{\ ^{s}\nabla }%
\mathbf{T}=0, &  & \mbox{ the
Levi--Civita connection;} \\
\ ^{s}\widehat{\mathbf{D}}: &  & \ ^{s}\widehat{\mathbf{D}}\ ^{s}\mathbf{%
g=0;\ }h\widehat{\mathbf{T}}=0,\ ^{1}v\widehat{\mathbf{T}}=0,...,\ ^{s}v%
\widehat{\mathbf{T}}=0. &  & \mbox{ the canonical
d--connection.}%
\end{array}%
\right.
\end{equation*}
\end{proposition}

\begin{proof}
It is a trivial N--adapted construction with nonholonomic $r$--jet
prolongations of (\ref{doublecon}).

$\square $
\end{proof}

\vskip5pt

Einstein equations for a metric $\ ^{s}\mathbf{g}$ are natural
prolongations that can be formulated in standard form using the
LC--connection $\ ^{s}\nabla .$ By computing the corresponding Ricci tensor, $\
_{\shortmid }R_{\alpha _{s}\beta _{s}},$ curvature scalar, $\ _{\shortmid
}^{s}R,$ and Einstein tensor, $\ _{\shortmid }E_{\alpha _{s}\beta _{s}},$ we arrive at %
\begin{equation}
\ _{\shortmid }E_{\alpha _{s}\beta _{s}}:=\ _{\shortmid }R_{\alpha _{s}\beta
_{s}}-\frac{1}{2}g_{\alpha _{s}\beta _{s}}\ _{\shortmid }^{s}R=\varkappa \
_{\shortmid }T_{\alpha _{s}\beta _{s}},  \label{einsteq}
\end{equation}%
where $\varkappa $ is the gravitational constant and $\ _{\shortmid
}T_{\alpha _{s}\beta _{s}}$ is the stress--energy tensor for matter fields.
In 4-d, there are well-defined geometric/variational and physically
motivated procedures for constructing $\ _{\shortmid }T_{\alpha _{s}\beta
_{s}}.$ Such values can be similarly (at least geometrically) re--defined
with respect to N--adapted frames using distortion relations (\ref%
{distorsrel}) and introducing extra dimensions.

The gravitational field equations (\ref{einsteq}) can be rewritten
equivalently in N--adapted form for the canonical d--connection $\ ^{s}%
\widehat{\mathbf{D}},$ as%
\begin{eqnarray}
&&\ ^{s}\widehat{\mathbf{R}}_{\ \beta _{s}\delta _{s}}-\frac{1}{2}\mathbf{g}%
_{\beta _{s}\delta _{s}}\ ^{s}R=\mathbf{\Upsilon }_{\beta _{s}\delta _{s}},
\label{cdeinst} \\
&&\widehat{L}_{a_{s}j_{s}}^{c_{s}}=e_{a_{s}}(N_{j_{s}}^{c_{s}}),\ \widehat{C}%
_{j_{s}b_{s}}^{i_{s}}=0,\ ^{N}\widehat{J}_{\ j_{s}i_{s}}^{a_{s}}=0.
\label{lcconstr}
\end{eqnarray}%
The sources $\mathbf{\Upsilon }_{\beta _{s}\delta _{s}}$ are constructed as
in GR but with nonholonomic jet deformations for the formal extra dimensions,
when $\mathbf{\Upsilon }_{\beta _{s}\delta _{s}}\rightarrow \varkappa
T_{\beta _{s}\delta _{s}}$ \ for $\ ^{s}\widehat{\mathbf{D}}\rightarrow \
^{s}\nabla .$ The solutions of \ (\ref{cdeinst}) contain nonholonomically
induced torsion (\ref{dt}).

If the conditions (\ref{lcconstr}) are satisfied, the d-torsion coefficients
(\ref{dtors}) are zero and we get the LC--connection, i.e. it is possible to
"extract" solutions like as for the standard Einstein equations. The decoupling
property can be proved in explicit form by working with $\ ^{s}\widehat{\mathbf{%
D}}$ and nonholonomic torsion configurations. Having constructed certain
classes of solutions in explicit form, with nonholonomically induced
torsions and depending on various sets of integration and generating
functions and parameters, we can "extract" the solutions for $\ ^{s}\nabla $
by imposing at the end additional constraints that give zero torsion.

Using natural prolongations from $\mathbf{V}$ on $\mathbf{J}^{r}(\mathbf{V},%
\mathbf{V}^{\prime })$, we prove

\begin{theorem}
\label{th2.4}In nonholonomic N--adapted $r$--jet variables, the gradient
canonical Ricci jet--solitons are defined by equations%
\begin{equation}
\widehat{\mathbf{R}}_{\ \beta _{s}\gamma _{s}}+\widehat{\mathbf{D}}_{\beta
_{s}}\widehat{\mathbf{D}}_{\gamma _{s}}\kappa =\lambda \mathbf{g}_{\beta
_{s}\gamma _{s}}.  \label{nriccisol1}
\end{equation}
\end{theorem}

Our first goal is to elaborate a geometric method for decoupling the
equations (\ref{nriccisol}) and (\ref{nriccisol1}). This is possible for
certain classes of nonholonomic constraints when the systems of nonlinear
PDE (\ref{cdeinst}) is supplemented with  additional zero torsion conditions (\ref%
{lcconstr}). We are able to find general classes of solutions when $%
\kappa $ is parameterized so as to satisfy%
\begin{eqnarray}
\widehat{\mathbf{R}}_{ij} &=&\ ^{h}\Upsilon (x^{k})\mathbf{g}_{ij},
\label{geq1} \\
\widehat{\mathbf{R}}_{ab} &=&\ ^{v}\Upsilon (x^{k},y^{a})\mathbf{g}_{ab},
\label{geq2} \\
\widehat{\mathbf{R}}_{\beta \gamma } &=&0,\mbox{ for }\beta \neq \gamma ,
\label{geq3}
\end{eqnarray}%
\begin{eqnarray}
\widehat{\mathbf{R}}_{a_{1}b_{1}} &=&\ ^{^{1}v}\Upsilon (x^{k},y^{a},\zeta
^{a_{1}})\mathbf{g}_{a_{1}b_{1}},  \label{geq4} \\
\widehat{\mathbf{R}}_{\beta _{1}\gamma _{1}} &=&0,\mbox{ for }\beta _{1}\neq
\gamma _{1},  \label{geq5}
\end{eqnarray}%
\begin{equation*}
...
\end{equation*}%
\begin{eqnarray}
\widehat{\mathbf{R}}_{a_{s}b_{s}} &=&\ ^{^{s}v}\Upsilon (x^{k},y^{a},\zeta
^{a_{1}},...,\zeta ^{a_{s}})\mathbf{g}_{a_{s}b_{s}},  \label{geqs} \\
\widehat{\mathbf{R}}_{\beta _{s}\gamma _{s}} &=&0,\mbox{ for }\beta _{s}\neq
\gamma _{s},  \label{geqs1}
\end{eqnarray}%
with respect to N--adapted frames (\ref{naders}) and (\ref{nadifs}). The
effective source (anisotropically polarized cosmological constant)
\begin{equation}
\Upsilon _{\beta }^{\alpha }=diag[\Upsilon _{1}^{1}=\Upsilon _{2}^{2}=\
^{h}\Upsilon ,\Upsilon _{3}^{3}=\Upsilon _{4}^{4}=\ ^{v}\Upsilon ,\Upsilon
_{5}^{5}=\Upsilon _{6}^{6}=\ ^{^{1}v}\Upsilon ,...,\Upsilon
_{2s-1}^{2s-1}=\Upsilon _{2s}^{2s}=\ ^{^{s}v}\Upsilon ],  \label{sourse}
\end{equation}%
is parameterized accordingly that allows the integration of the differential
equations in explicit forms. We shall prove that the general classes of
solutions $\mathbf{g}_{\alpha \beta }(u)$ depend generically on all
spacetime coordinates via the corresponding generating and integration functions,
and various integrations constants. The solutions will represent an explicit application of
the geometry of nonholonomic distributions and generalized connections in
mathematical relativity, modified gravity theories and the theory of
physically important nonlinear systems of PDEs.

The second goal is to find explicit solutions for the Levi-Civita
(LC-configurations) (\ref{lcconstr}) , i.e. for the additional constraints when
\begin{equation}
\ ^{s}\widehat{\mathbf{T}}=0,  \label{lccond}
\end{equation}%
(this formula follows from (\ref{torsshell})), when for some classes
of solutions $(\ ^{s}\mathbf{g,}\ ^{s}\mathbf{N,}\ ^{s}\widehat{\mathbf{D}})$
of (\ref{geq1})--(\ref{geqs1}) we can extract certain subvarieties of
solutions $(\ ^{s}\mathbf{\check{g},}\ ^{s}\mathbf{\check{N},}\ ^{s}\nabla
=\ ^{s}\mathbf{\check{D}=}\ ^{s}\mathbf{D}_{\ \shortmid \ ^{s}\mathbf{\check{%
T}\rightarrow 0}})$ for zero torsion, after re--scaling the generating
functions and sources, $\ ^{s}\mathbf{g\rightarrow }\ ^{s}\mathbf{\check{g}}%
,\ ^{s}\mathbf{N\rightarrow }\ ^{s}\mathbf{\check{N}}$ and $\ ^{h}\Upsilon
(x^{k}),..,\ ^{^{s}v}\Upsilon (x^{k})\rightarrow \lambda =const. $ With this procedure,
we formulate a geometric method of constructing exact solutions of the
Einstein equations (\ref{einsteq}) with re--defined source and N--adapted
frame structures when
\begin{equation}
R_{\alpha _{s}\beta _{s}}[\ ^{s}\nabla ]=\lambda \mathbf{\check{g}}_{\alpha
_{s}\beta _{s}}  \label{enstm}
\end{equation}%
on nonholonomic manifolds/bundles and $r$--jet prolongations.  The
metrics contain generic off--diagonal elements (i.e. can not be diagonalized via
coordinate transforms), which may depend on all spacetime coordinates and can be
prescribed (via generating/integrating functions and constants) to satisfy
various necessary types of  symmetry, boundary, Cauchy, and topological conditions,
with possible singularities and horizons. We note that the system (\ref%
{geq1})--(\ref{geqs1}) together with LC--conditions (\ref{lccond}) is
equivalent to (\ref{enstm}). Both such systems of PDEs are nonlinear and
parameter dependent. So, it is important at what stage  certain nonholonomic constraints
and ansatz conditions for frames, metrics and connections are imposed i.e., at
the end, when some solutions for $\ ^{s}\widehat{\mathbf{D}}$ have been
found, or at the beginning, when $\ ^{s}\widehat{\mathbf{D}}\rightarrow \
^{s}\nabla .$ We can not decouple such systems of equations in a general
form if we work from the very beginning with the Levi--Civita connection $\
^{s}\nabla .$

The third goal is to lay down certain geometric conditions on when
a general (pseudo) Riemannian manifold $(\mathbf{g,V)}$ (the metric $\mathbf{%
g}$ may  or not be a solution of any (modified)\ Einstein or Ricci soliton
equations) can be nonholonomically deformed via the corresponding nonhololonomic
jet maps with generalized connection structures into certain geometrically/
physically important classes of solutions of systems of the type (\ref{geq1})--(%
\ref{lccond}), or (\ref{enstm}). In such cases $(\mathbf{g,V)}%
\begin{array}{c}
\mbox{ nonholonomic }Jet \\
----------\longrightarrow \\
{\ }%
\end{array}%
(\mathbf{\check{g},\check{N},\check{D},\check{V}})$, when the target space $%
\mathbf{\check{V}}$ and the fundamental geometric structures $(\mathbf{\check{g},%
\check{N},\check{D})}$ are constructed with nonholonomic jet transformations and ($Jet$)
are solutions of certain (modified/generalized) Einstein or Ricci soliton
equations that depend on generalized jet parameters and corresponding jet
symmetries. We note, that for such geometric and physical models the jet
variables are prescribed certain constant values (we suppress the left label
$s$).

\section{Decoupling and Integration of Jet Prolongation of Einstein Equations%
}

\label{s3}In this section, we prove that the system of nonlinear PDEs (\ref%
{geq1})--(\ref{geqs1}) with possible constraints (\ref{lccond}) giving (\ref%
{enstm}), can be decoupled in very general forms with respect to N--adapted
frames with two dimensional shell parameterizations of jet variables. We
show how such decoupled systems can be integrated in general forms for
vacuum and non--vacuum solutions in (modified) gravity and Ricci soliton
theories.

\subsection{Off--diagonal metrics for $r$--jet configurations with one
Killing symmetry}

We study nonholonomic jet deformations of the "primary" geometric/physical data
into "target" data,
\begin{equation*}
\mbox{[ primary ]}(\ _{\circ }^{s}\mathbf{g,}\ _{\circ }^{s}\mathbf{N,}\
_{\circ }^{s}\widehat{\mathbf{D}})\ \rightarrow \mbox{[ target ]}(\ _{\eta
}^{s}\mathbf{g}=\ ^{s}\mathbf{\mathbf{g},}\ _{\eta }^{s}\mathbf{N}=\ ^{s}%
\mathbf{\mathbf{N},}\ _{\eta }^{s}\widehat{\mathbf{D}}=\ ^{s}\widehat{%
\mathbf{D}}).
\end{equation*}%
The values labeled by "$\circ "$  define exact solutions in a Ricci
soliton or gravity theory. The metrics with left label "$\eta $"  define
a solution of modified gravitational field equations (\ref{geq1})--(\ref%
{geqs1}). The prime ansatz is written
\begin{eqnarray}
\ \ _{\circ }^{s}\mathbf{g} &=&\ \mathring{g}_{i}(x^{k})dx^{i}\otimes dx^{i}+%
\mathring{h}_{a}(x^{k},y^{4})\mathbf{\mathring{e}}^{a}\otimes \mathbf{%
\mathring{e}}^{b}+\epsilon _{a_{1}}\ dy^{a_{1}}\otimes \ dy^{a_{1}}+....+\
\epsilon _{a_{s}}dy^{a_{s}}\otimes \ dy^{a_{s}},  \notag \\
\mathbf{\mathring{e}}^{a} &=&dy^{a}+\mathring{N}_{i}^{a}(x^{k},y^{4})dx^{i},%
\mbox{ with }\mathring{N}_{i}^{3}=\mathring{n}_{i},\mathring{N}_{i}^{4}=%
\mathring{w}_{i},  \label{ansprime}
\end{eqnarray}%
for $\epsilon _{a_{s}}=\pm 1$ depending on the signature of extra dimensions and
$(\mathring{g}_{i},\mathring{h}_{a};\mathring{N}_{i}^{a})$ defining, for
instance, the Kerr black hole solution trivially imbedded into a $4+2s$ jet
prolongation of spacetime. We express the N--adapted coefficients of a
target ansatz (\ref{dm}) as
\begin{equation}
g_{\alpha _{s}} =\eta _{\alpha _{s}}(u^{\beta _{s}})\mathring{g}_{\alpha
_{s}};N_{i_{s}}^{a_{s}}=\ _{\eta }N_{i_{s}}^{a_{s}}(u^{\beta
_{s-1}},y^{4+2s});\ n_{i} = \eta _{i}^{3}\mathring{n}_{i},w_{i}=\eta _{i}^{4}%
\mathring{w}_{i},\mbox{ not summation over i};  \label{etad}
\end{equation}%
with the so--called gravitational "polarization" functions and extra dimensional
N-coefficients, $\eta _{\alpha _{s}},\eta _{i}^{a}$ and$\ _{\eta
}N_{i_{s}}^{a_{s}}.$ To be able to study certain limits $(\ _{\eta }^{s}%
\mathbf{g,}\ _{\eta }^{s}\mathbf{N,}\ _{\eta }^{s}\widehat{\mathbf{D}}%
)\rightarrow (\ _{\circ }^{s}\mathbf{g,}\ _{\circ }^{s}\mathbf{N,}\ _{\circ
}^{s}\widehat{\mathbf{D}}),$ for $\varepsilon \rightarrow 0,$ depending on a
small parameter $\varepsilon ,0\leq \varepsilon \ll 1,$ we  introduce
"small" polarizations of the type $\eta =1+\varepsilon \chi (u...)$ and $_{\eta
}N_{i_{s}}^{a_{s}}=\varepsilon n_{i_{s}}^{a_{s}}(u...).$

The decoupling property of modified Einstein equations can be proven in the
simplest form for certain ansatz with at least one Killing symmetry on a
spacetime coordinate and certain parameterizations of nonholonomic $r$--jet
prolongations. \ We consider target metrics of type (\ref{dm}) parameterized
in the form
\begin{eqnarray}
\ \ _{\omega }^{s}\mathbf{g} &=&\ g_{i}(x^{k})dx^{i}\otimes
dx^{i}+h_{a}(x^{k},y^{4})\mathbf{e}^{a}\otimes \mathbf{e}^{b}+  \label{ansk}
\\
&&h_{a_{1}}(u^{\alpha },\zeta ^{6})\ \mathbf{e}^{a_{1}}\otimes \mathbf{e}%
^{a_{1}}+h_{a_{2}}(u^{\alpha _{1}},\zeta ^{8})\ \mathbf{e}^{a_{2}}\otimes
\mathbf{e}^{b_{2}}+....+\ h_{a_{s}}(\ u^{\alpha _{s-1}},\zeta ^{a_{s}})%
\mathbf{e}^{a_{s}}\otimes \mathbf{e}^{a_{s}},  \notag \\
\mbox{ where } \mathbf{e}^{a} &=&dy^{a}+N_{i}^{a}dx^{i},\mbox{\ for \ }%
N_{i}^{3}=n_{i}(x^{k},y^{4}),N_{i}^{4}=w_{i}(x^{k},y^{4});  \notag \\
\mathbf{e}^{a_{1}} &=&d\zeta ^{a_{1}}+N_{\alpha }^{a_{1}}du^{\alpha },%
\mbox{\ for \ }N_{\alpha }^{5}=\ ^{1}n_{\alpha }(u^{\beta },\zeta
^{6}),N_{\alpha }^{6}=\ ^{1}w_{\alpha }(u^{\beta },\zeta ^{6});  \notag \\
\mathbf{e}^{a_{2}} &=&d\zeta ^{a_{2}}+N_{\alpha _{1}}^{a_{2}}du^{\alpha
_{1}},\mbox{\ for \ }N_{\alpha _{1}}^{7}=\ ^{2}n_{\alpha _{1}}(u^{\beta
_{1}},\zeta ^{8}),N_{\alpha _{1}}^{8}=\ ^{2}w_{\alpha }(u^{\beta _{1}},\zeta
^{8});  \notag \\
&&....  \notag \\
\mathbf{e}^{a_{s}} &=&d\zeta ^{a_{s}}+N_{\alpha _{s-1}}^{a_{s}}du^{\alpha
_{s-1}},\mbox{\ for \ }N_{\alpha _{s-1}}^{4+2s-1}=\ ^{s}n_{\alpha
_{1}}(u^{\beta _{s-1}},\zeta ^{4+2s}),N_{\alpha _{1}}^{4+2s}=\ ^{s}w_{\alpha
}(u^{\beta _{s-1}},\zeta ^{4+2s}).  \notag
\end{eqnarray}%
Such ansatz also contains  a jet Killing vector $\partial /\partial \zeta
^{s-1}$ because the jet coordinate $\zeta ^{s-1}$ is not contained in the
coefficients of such metrics.

\subsection{Decoupling in nonholonomic $r$--jet shell variables}

Let us consider an ansatz (\ref{ansk}) with $\ g_{i}(x^{k})=\epsilon
_{i}e^{\psi (x^{k})},$ where $\epsilon _{i}=\pm 1,$ and the $\gamma ,\alpha
,\beta $--coefficients are defined by respective generating functions $\phi
,\ ^{s}\phi $ following the formulae
\begin{eqnarray}
\gamma &:= &\partial _{4}(\ln |h_{3}|^{3/2}/|h_{4}|),\ \ \alpha
_{i}=(\partial _{i}\phi )(\partial _{4}h_{3})/2h_{3},\ \beta =(\partial
_{4}\phi )(\partial _{4}h_{3})/2h_{3},  \label{ca1} \\
&&\mbox{ for generating function }\phi =\ln \left\vert \partial _{4}h_{3}/%
\sqrt{|h_{3}h_{4}|}\right\vert ,  \label{c1} \\
\ ^{1}\gamma &:=&\eth _{6}(\ln |h_{5}|^{3/2}/|h_{6}|),\ ^{1}\alpha _{\tau
}=(\eth _{\tau }\ ^{1}\phi )(\eth _{6}h_{5})/2h_{5},\ ^{1}\beta =(\eth
_{\tau }\ ^{1}\phi )(\eth _{6}h_{5})/2h_{5},  \label{ca2} \\
&&\mbox{ for r--jet generating function }\ ^{1}\phi =\ln \left\vert (\eth
_{6}h_{5})/\sqrt{|h_{5}h_{6}|}\right\vert ,  \label{c2} \\
\ ^{2}\gamma &:= &\eth _{8}(\ln |h_{7}|^{3/2}/|h_{8}|),\ \ ^{2}\alpha _{\tau
_{1}}=(\eth _{\tau _{1}}\ ^{2}\phi )(\eth _{8}h_{7})/2h_{7},\ \ ^{2}\beta
=(\eth _{\tau _{1}}\ ^{2}\phi )(\eth _{8}h_{7})/2h_{7},  \notag \\
&&\mbox{ for r--jet generating function }\ ^{s}\phi =\ln \left\vert \eth
_{2s}h_{2s-1}/\sqrt{|h_{2s-1}h_{2s}|}\right\vert ,  \notag \\
&&....,  \notag
\end{eqnarray}%
with nonzero $\partial _{4}\phi ,\partial _{4}h_{a},$ $\eth _{6}\ ^{1}\phi
,\eth _{6}h_{a_{1}},\eth _{2s}\ ^{2}\phi ,\eth _{2s}h_{a_{2}}.$

We assume that via the N--adapted frame transformations the sources $\mathbf{\Upsilon
}_{\beta _{s}\delta _{s}}$ (\ref{sourse}) in the equations (\ref{geq1}), (\ref%
{geq2}), (\ref{geq4}) and (\ref{geqs}) can be parameterized in the form
\begin{eqnarray}
\mathbf{\Upsilon }_{1}^{1} &=&\mathbf{\Upsilon }_{2}^{2}=\ ^{v}\Lambda
(x^{k},y^{4})+\ _{1}^{v}\Lambda (u^{\beta },\zeta ^{6})+\ _{2}^{v}\Lambda
(u^{\beta _{1}},\zeta ^{8}),\mathbf{\Upsilon }_{3}^{3}=\mathbf{\Upsilon }%
_{4}^{4}=\Lambda (x^{k})+\ _{1}^{v}\Lambda (u^{\beta },\zeta ^{6})+\
_{2}^{v}\Lambda (u^{\beta _{1}},\zeta ^{8}),  \notag \\
\mathbf{\Upsilon }_{5}^{5} &=&\mathbf{\Upsilon }_{6}^{6}=\Lambda (x^{k})+\
^{v}\Lambda (x^{k},y^{4})+\ _{2}^{v}\Lambda (u^{\beta _{1}},\zeta ^{8}),%
\mathbf{\Upsilon }_{7}^{7}=\mathbf{\Upsilon }_{8}^{8}=\Lambda (x^{k})+\
^{v}\Lambda (x^{k},y^{4})+\ _{1}^{v}\Lambda (u^{\beta },\zeta ^{6}).
\label{dsource}
\end{eqnarray}%
Such parameterizations are very general for (effective) $\mathbf{\Upsilon }%
_{\beta _{s}\delta _{s}}$ with arbitrary contributions from Ricci soliton or
(modified) gravity and matter fields and further $r$-jet generalizations
when the N--adapted coefficients are modeled in certain systems of
references by "polarized" cosmological constants $\Lambda (x^{k}),\
^{v}\Lambda (x^{k},y^{4}),\ _{1}^{v}\Lambda (u^{\beta },\zeta ^{6}),$ $\
_{2}^{v}\Lambda (u^{\beta _{1}},\zeta ^{8})$ etc. For certain models of
 gravity in extra dimensions , we  consider, for simplicity, $\Lambda =\
^{v}\Lambda =$ $\ _{1}^{v}\Lambda =\ _{2}^{v}\Lambda =\ ^{\circ }\Lambda
=const.$ Such effective sources can always be introduced by re--defining the
generating functions (see below) for non--vacuum configurations.

\begin{theorem}
\label{tdecoupling}For a general off--diagonal ansatz (\ref{ansk}) with
Killing symmetry $\partial _{3}$ and N--adapted parameterizations for
generating functions (\ref{ca1})--(\ref{c2}) and sources (\ref{dsource}),
the system of modified Einstein equations (see N--adapted equations (\ref%
{equ1})--(\ref{equ4d2s})) decouple in the following form:\newline
For a nonholonomic 2+2 spacetime splitting,
\begin{equation}
\epsilon _{1}\partial _{11}\psi +\epsilon _{2}\partial _{22}\psi =2\Lambda
(x^{k}),  \label{e1}
\end{equation}%
\begin{eqnarray}
(\partial _{4}\phi )(\partial _{4}h_{3}) &=&2h_{3}h_{4}\ ^{v}\Lambda
(x^{k},y^{4}),\   \label{e2} \\
\partial _{44}n_{i}+\gamma \partial _{4}n_{i} &=&0,  \label{e3} \\
\beta w_{i}-\alpha _{i} &=&0;\   \label{e4}
\end{eqnarray}%
and, on nonholonomic $r$--jet variables,%
\begin{eqnarray}
(\eth _{6}\ ^{1}\phi )(\eth _{6}h_{5}) &=&2h_{5}h_{6}\ _{1}^{v}\Lambda
(u^{\beta },\zeta ^{6}),  \label{e2aa} \\
\eth _{66}^{2}\ ^{1}n_{\tau }+\ ^{1}\gamma \eth _{6}\ ^{1}n_{\tau } &=&0,
\label{e3aa} \\
\ ^{1}\beta \ ^{1}w_{\tau }-\ ^{1}\alpha _{\tau } &=&0,\   \label{e4aa} \\
&&...  \notag
\end{eqnarray}%
\begin{eqnarray}
(\eth _{2s}\ ^{s}\phi )(\eth _{2s}h_{2s-1}) &=&2h_{2s-1}h_{2s}\
_{s}^{v}\Lambda (u^{\beta _{s-1}},\zeta ^{2s}),  \notag \\
\eth _{2s\ 2s}^{2}\ ^{2}n_{\tau _{1}}+\ ^{2}\gamma \eth _{2s}\ ^{2}n_{\tau
_{1}} &=&0,  \notag \\
\ ^{2}\beta \ ^{2}w_{\tau _{1}}-\ ^{2}\alpha _{\tau _{1}} &=&0,\
\label{e4dd}
\end{eqnarray}%
\begin{equation*}
....
\end{equation*}
\end{theorem}

\begin{proof}
It is a tedious technical proof following an explicit calculation of
nontrivial components of the Ricci d--tensor for the mentioned ansatz and
parameterizations of (effective) sources and parametrization functions (see
Appendix \ref{assdecoup}). Such equations are straightforward consequences
of the symmetries of the Einstein and the Ricci d--tensors (\ref{sourc1}) for the
canonical d--connection $\ ^{s}\widehat{\mathbf{D}}.$

$\square $
\end{proof}

\vskip5pt

Let us explain in brief the decoupling property for the nonholonomic 4--d and $r$%
--jet equations (\ref{e1})-- (\ref{e4dd}):

\begin{enumerate}
\item The equation (\ref{e1}) is just a 2-d Laplace, or d'Alambert equation
(depending on the prescribed signature for $\varepsilon _{i}=\pm 1)$ with
solutions determined by any source $\Lambda (x^{k}).$

\item The equation (\ref{e2}) contains only the partial derivative $\partial
_{4}$ and constitutes, together with the algebraic formula for the coefficient (%
\ref{c1}), a system of two equations for four functions: $h_{3}(x^{i},y^{4}),$ $%
h_{4}(x^{i},y^{4})$ and $\phi (x^{i},y^{4})$ and source $\ ^{v}\Lambda
(x^{k},y^{4}).$ Prescribing two any such functions, we can define
(integrating w.r.t. to the coordinate $y^{4},$ or differentiating w.r.t.
this coordinate) two other functions of the same type. Such functions can be re--defined
in order to transform $\ ^{v}\Lambda (x^{k},y^{4})$ into an effective
constant. The function $\phi (x^{i},y^{4})$ and/ or any its functional can
be considered as a generating function which can be prescribed following
certain geometric or physical arguments on symmetries, boundary conditions,
any explicit singular or non--singular behaviour, smooth class conditions,
stochastic conditions, topological configurations etc. This allows us to
compute $h_{a}(x^{i},y^{4}) $ in explicit form. If necessary, we can
consider, for instance, the coefficient $h_{3}(x^{i},y^{4})$ (or $%
h_{4}(x^{i},y^{4})$) to be a generating function and compute $h_{4}$ (or $%
h_{3}$) and $\phi $ for a given source $\ ^{v}\Lambda (x^{k},y^{4}).$ We
note that if we consider vacuum solutions for $\ ^{s}\widehat{\mathbf{D}}$
with $\ ^{v}\Lambda =0$ in (\ref{e2}), we are constrained to study only
configurations with N--adapted coefficients $\partial _{4}h_{3}=0$ and/or $%
\partial _{4}\phi =0.$ Such solutions are also important to study geometric
and physical properties of vacuum off--diagonal configurations and their
possible diagonal limits. The decoupling property is explicit for such
vacuum and non--vacuum equations because they contain only two
coefficients of the metric, $h_{a};$ and the coupling with other diagonal and/or
off--diagonal coefficients (like $g_{i}$, $w_{i},$ $n_{i},$ and/or jet
prolongations) are not involved.

\item Having computed the coefficient $\gamma $ (\ref{ca1}), the
N--connection coefficients $n_{i}$ can be defined after two integrations on $%
y^{4}$ in (\ref{e3}). This defines a part of N--connection coefficients. A
value $n_{i}$ does not depend on other coefficients of a d--metric except for
the coefficient $\gamma $ determined by $h_{a}.$

\item Using $h_{3}$ and $\phi $ from the  previous discussion, we compute the
coefficients $\alpha _{i}$ and $\beta ,$ see (\ref{ca1}), which allows us to
define $w_{i}$ from the algebraic equations (\ref{e4}). Such $w_{i}$ define
a different part of the N--connection coefficients, can determine off--diagonal
coefficients of the metric and  may also contribute to the diagonal
coefficients if such metrics are written in coordinate bases. Nevertheless,
the functions $w_{i}$ are independent from other coefficients of the
d--metric and N--connection with respect to N--adapted frames.

\item The procedure in items 2-4 can be repeated step by step for any shell with $r$%
--jet variables. The equations (\ref{e2aa})--(\ref{e4aa}) are completely
similar to (\ref{e2})--(\ref{e4}) but contain additional dependencies on jet
coordinates and derivatives $\eth $ that depend on respective jet coordinates. For
instance, the equation (\ref{e2aa}) and formula (\ref{c2}) with partial
derivative $\partial _{6}$ are for functions $h_{5}(x^{i},y^{a},\zeta ^{6}),$
$h_{6}(x^{i},y^{a},\zeta ^{6})$ and $\ ^{1}\phi (x^{i},y^{a},\zeta ^{6})$ $\
$and source $\ _{1}^{v}\Lambda (u^{\beta },\zeta ^{6}).$ We can compute any
two such functions integrating w.r.t. the coordinate $\zeta ^{6}$ if two
other ones are prescribed. In a similar form, we follow points 3 and 4 with $%
\ ^{1}\alpha _{\tau },\ ^{1}\beta ,\ ^{1}\gamma ,$ see (\ref{ca2}), and
compute the higher order N--connection coefficients $\ ^{1}n_{\tau }$ and $\
^{1}w_{\tau }.$

\item The splitting property holds on any 2-d shell as it is stated in (\ref%
{e4dd}). Here we note that it is not clear if any splitting of equations
could be proven in general form for 3--d shells. This is partly because  the
topological properties of 2-d and 3--d shells are very different. The
equations of type (\ref{e2}), (\ref{e2aa}) are degenerate for 1--d shells. That is partly
why our AFDM is based on 2+2+2+... splitting which allows  to decouple and
solve such nonlinear systems of PDEs in general form. We can consider in
similar form splitting of type 3+2+2+... for 3--d bases when the point 1
refers to the 3-d Laplace, or d'Alameber equations. For certain configurations,
we can generate extra--dimension and jet configurations by imbedding 1,2 and
3 dimensional metrics into some d--metric configurations with the splitting
2+2+2+... In arbitrary systems of reference, such effective vacuum and
non--vacuum nonholonomic dynamical systems depend on spacetime coordinates.
They may be of nonlinear evolution type, or Ricci soliton fixed
configurations, for respectively prescribed signatures. Nevertheless, the
assumption on 2--d shall is a condition imposing certain (2-d) topological
restrictions on the jet variables extensions.
\end{enumerate}

Finally, we emphasize that the splitting property of nonholonomic and
holonomic Einstein equations for higher dimensions was proven in \cite{vex3}%
. Geometric and physical models with jet variables and extra dimension
coordinates are similar in certain sense but with different jet symmetry
conditions.\footnote{In this work, the concept of jet symmetry of some classes of solutions for  (modified) Einstein equations is used in a sense that it generalizes certain Lie group, Killing type, or anholonomy relations by introducing jet variables. In particular, we obtain elements of the type $L^r_{n,m}$  introduced in section \ref{ssholj}, and studied in details in Ref. \cite{kolar,kures98,kures01}.} Formal 2+2+2+... splitting are possible only in adapted
nonholonomic systems of reference and the derived nonholonomic
dynamical/evolution equations encode a different type of interior gauge like
dynamics.

\subsection{Jet integral varieties for off--diagonal metrics and generalized
connections}

The system of nonlinear PDEs (\ref{e1})-- (\ref{e4dd}) in spacetime and jet
variables can be integrated in general forms for any 2--d shell $\dim \ ^{s}%
\mathbf{V}\geq 4.$ We note that the coefficients $g_{i}=\epsilon _{i}e^{\psi
(x^{k})}$ are defined by solutions of corresponding Laplace/ d'Alambert
equation (\ref{e1}) that do not contain jet coordinates in N--adapted
frames. General solutions will be considered for "vertical" spacetime and
jet variables.

\subsubsection{ 4--d non--vacuum spacetime nonholonomic configurations}

We can solve (\ref{e2}) and (\ref{c1}) for any $\partial _{4}\phi \neq
0,h_{a}\neq 0$ and $\ ^{v}\Lambda \neq 0.$ Let us re-write respectively
the relevant equations,
\begin{equation}
\ h_{3}h_{4}=(\partial _{4}\phi )(\partial _{4}h_{3})/2\ ^{v}\Lambda
\mbox{
and }|h_{3}h_{4}|=(\partial _{4}h_{3})^{2}e^{-2\phi }.  \label{eq4bb}
\end{equation}%
By considering a new generating function $\Phi :=e^{\phi }$ and introducing the
first equation into the second one, we get
\begin{equation}
|\partial _{4}h_{3}|=\frac{\partial _{4}(e^{2\phi })}{4|\ ^{v}\Lambda |}=%
\frac{\partial _{4}[\Phi ^{2}]}{4|\ ^{v}\Lambda |}.  \label{aux01}
\end{equation}%
Integrating w.r.t. the coordinate $y^{4},$ we find

\begin{equation}
h_{3}[\Phi ,\ ^{v}\Lambda ]=\ ^{0}h_{3}(x^{k})+\frac{\epsilon _{3}\epsilon
_{4}}{4}\int dy^{4}\frac{\partial _{4}(\Phi ^{2})}{\ ^{v}\Lambda },
\label{h3aux}
\end{equation}%
where $\ ^{0}h_{3}=\ ^{0}h_{3}(x^{k})$ is an integration function and $%
\epsilon _{3},\epsilon _{4}=\pm 1.$ To compute $h_{4}$ we can use the first
equation in (\ref{eq4bb}) when
\begin{equation}
h_{4}[\Phi ,\ ^{v}\Lambda ]=\frac{(\partial _{4}\phi )}{\ ^{v}\Lambda }%
\partial _{4}(\ln \sqrt{|h_{3}|})=\frac{1}{2\ ^{v}\Lambda }\frac{\partial
_{4}\Phi }{\Phi }\frac{\partial _{4}h_{3}}{h_{3}}.  \label{h4aux}
\end{equation}

The formulae (\ref{h3aux}) and (\ref{h4aux}) for $h_{a}[\Phi ,\ ^{v}\Lambda
] $ can be re--parameterized in a more convenient form with an effective
cosmological constant $\widetilde{\Lambda }_{0}=const\neq 0.$ Let us
re--define the generating function $\Phi \rightarrow \tilde{\Phi},$ when $%
\frac{\partial _{4}[\Phi ^{2}]}{\ ^{v}\Lambda }=\frac{\partial _{4}[\tilde{%
\Phi}^{2}]}{\ \tilde{\Lambda}_{0}},$ i.e.%
\begin{equation}
\Phi ^{2}=\widetilde{\Lambda }_{0}^{-1}\int dy^{4}(\ ^{v}\Lambda )\partial
_{4}(\tilde{\Phi}^{2})\mbox{
and }\tilde{\Phi}^{2}=\widetilde{\Lambda }_{0}\int dy^{4}(\ ^{v}\Lambda
)^{-1}\partial _{4}(\Phi ^{2}).  \label{rescgf}
\end{equation}%
By introducing the integration function$\ ^{0}h_{3}(x^{k})$ and $\epsilon _{3},$
 $\epsilon _{4},$ in $\Phi $ and respectively, in $\ ^{v}\Lambda ,$ we
express the solutions for $h_{a}$ as functionals on $[\tilde{\Phi},%
\widetilde{\Lambda }_{0},\Xi ],$
\begin{equation}
h_{3}[\tilde{\Phi},\widetilde{\Lambda }_{0}]=\frac{\tilde{\Phi}^{2}}{4%
\widetilde{\Lambda }_{0}}\mbox{ and
}h_{4}[\tilde{\Phi},\widetilde{\Lambda }_{0},\Xi ]=\frac{(\partial _{4}%
\tilde{\Phi})^{2}}{\Xi }.  \label{solha}
\end{equation}%
The functional $\Xi \lbrack \ ^{v}\Lambda ,\tilde{\Phi}]=\int dy^{4}(\
^{v}\Lambda )\partial _{4}(\tilde{\Phi}^{2})$ in the last formula can be
considered as a re--defined source for a prescribed generating function $%
\tilde{\Phi},$ $\ ^{v}\Lambda \rightarrow \Xi ,$ when $\ ^{v}\Lambda
=\partial _{4}\Xi /\partial _{4}(\tilde{\Phi}^{2}) $ (it contains
information on the Ricci soliton contribution, and/or effective  energy-momentum
tensor of matter in modified gravity theories). We can work with a couple
of generating data, $(\Phi ,\ ^{v}\Lambda )$ and $(\tilde{\Phi},\Xi ),$
related by formulae (\ref{rescgf}) for a prescribed effective cosmological
constant $\tilde{\Lambda}_{0}.$

Using the values $h_{a}$ (\ref{solha}), we compute the coefficients $\alpha
_{i},\beta $ and $\gamma $ from (\ref{ca1}). The resulting solutions for
N--coefficients, i.e of respective equations (\ref{e3}) and (\ref{e4}), can
be expressed as,
\begin{eqnarray}
n_{k} &=&\ _{1}n_{k}+\ _{2}n_{k}\int dy^{4}h_{4}/(\sqrt{|h_{3}|})^{3}=\
_{1}n_{k}+\ _{2}\widetilde{n}_{k}\int dy^{4}(\partial _{4}\tilde{\Phi})^{2}/%
\tilde{\Phi}^{3}\Xi ,\mbox{ and }  \notag \\
w_{i} &=&\partial _{i}\phi /\partial _{4}\phi =\partial _{i}\Phi /\partial
_{4}\Phi =\partial _{i}\Phi ^{2}/\partial _{4}\Phi ^{2}=\int dy^{4}\partial
_{i}[(\ ^{v}\Lambda )\partial _{4}(\tilde{\Phi}^{2})]/(\ ^{v}\Lambda
)\partial _{4}(\tilde{\Phi}^{2})=\partial _{i}\Xi /\partial _{4}\Xi ,
\label{solhn}
\end{eqnarray}%
where $\ _{1}n_{k}(x^{i})$ and $\ _{2}n_{k}(x^{i}),$ or $_{2}\widetilde{n}%
_{k}(x^{i})=8\ _{2}n_{k}(x^{i})|\widetilde{\Lambda }|^{3/2},$ are
integration functions. \ Putting together the formulae for the coefficients (\ref%
{solha})-(\ref{solhn}), we prove:

\begin{theorem}
\label{th3.2}The system of nonlinear PDEs (\ref{e1})-- (\ref{e4}) for
non--vacuum 4--d configurations with Killing symmetry $\partial _{3}$ is
integrated in general form by  quadratic line element of the form  $ds_{4[dK]}^{2}=g_{%
\alpha \beta }(x^{k},y^{4})du^{\alpha }du^{\beta }$, when {\small
\begin{equation}
ds_{4[dK]}^{2}=\epsilon _{i}e^{\psi (x^{k})}(dx^{i})^{2}+\frac{\tilde{\Phi}%
^{2}}{4\widetilde{\Lambda }_{0}}[dy^{3}+(\ _{1}n_{k}+_{2}\widetilde{n}%
_{k}\int dy^{4}\frac{(\partial _{4}\tilde{\Phi})^{2}}{\tilde{\Phi}^{3}\Xi }%
)dx^{k}]^{2}+\frac{(\partial _{4}\tilde{\Phi})^{2}}{\Xi }\ [dy^{4}+\frac{%
\partial _{i}\Xi }{\partial _{4}\Xi }dx^{i}]^{2}.  \label{qnk4d}
\end{equation}%
}
\end{theorem}

This line element defines a family of generic off--diagonal solutions with
Killing symmetry $\partial /\partial y^{3}$ of the 4--d nonholonomic
Einstein equations (\ref{geq1})--(\ref{geq3}) with source parametrization of the
type (\ref{dsource}) and for the canonical d--connection $\ \widehat{\mathbf{%
D}}$ (the label $4dK$ is for "nonholonomic 4-d Killing solutions). We can
verify by straightforward computations that the nonholonomy coefficients $%
W_{\alpha \beta }^{\gamma }$ in (\ref{anhrel1}) are not zero if arbitrary
generating function $\phi $ and integration funtions ($\ ^{0}h_{a},\
_{1}n_{k}$ and $\ _{2}n_{k})$ are considered. This means that such metrics
can not be diagonalized by coordinate transforms in a finite spacetime
region. The class of solutions (\ref{qnk4d}) carry  nontrivial canonical
d--torsion (\ref{dt}) which can be proven by using explicit N--adapted
coefficient formulae (\ref{dtors}) for the canonical d--connection (\ref%
{candcon}). In section \ref{sslc}, we shall state additional conditions when
such solutions define LC--configurations. Vacuum nonholonomic spacetime
quadratic elements are considered in section \ref{assvacuum}.

\subsubsection{Nonholonomic $r$--jet prolongations of non--vacuum solutions}

The solutions with jet variables can be constructed in certain forms which
are similar to the 4--d case but achieved by using new classes of generating and
integration functions with dependencies on $r$--jet shell coordinates. We
can generate solutions of the system (\ref{e2aa})--(\ref{e4aa}) with
coefficients (\ref{c2}) and (\ref{ca2}) following a formal analogy when the
generating functions and (effective) sources from the previous paragraph are
generalized in the form: $\partial _{4}\rightarrow \eth _{6},\phi
(x^{k},y^{4})\rightarrow \ ^{1}\phi (u^{\tau },\zeta ^{6}),\ ^{v}\Lambda
(x^{k},y^{4})\rightarrow \ _{1}^{v}\Lambda (u^{\tau },\zeta ^{6})...$ and with
re--defined values $\ \tilde{\Phi}(x^{k},y^{4})\rightarrow $ $\ ^{1}\tilde{%
\Phi}(u^{\tau },\zeta ^{6})$ and $\ \widetilde{\Lambda }_{0}\rightarrow \
^{1}\widetilde{\Lambda }_{0}=const.$

The first set of $r$--jet coefficients of the d--metric are computed to be $%
h_{5}[\ ^{1}\tilde{\Phi},\ ^{1}\widetilde{\Lambda }]=\frac{\ ^{1}\tilde{\Phi}%
^{2}}{4\ ^{1}\widetilde{\Lambda }}$ and $h_{6}[\ ^{1}\tilde{\Phi}]=\frac{%
(\eth _{6}\ ^{1}\tilde{\Phi})^{2}}{\ ^{1}\Xi }$, for $\ ^{1}\Xi =\int d\zeta
^{6}(\ _{1}^{v}\Lambda )\eth _{6}(\ ^{1}\tilde{\Phi}^{2})$ and, for
N--coefficients,
\begin{eqnarray*}
\ ^{1}n_{\tau } &=&\ _{1}^{1}n_{\tau }+\ _{2}^{1}n_{\tau }\int d\zeta
^{6}h_{6}/(\sqrt{|h_{5}|})^{3}=\ _{1}^{1}n_{k}+\ _{2}^{1}\widetilde{n}%
_{k}\int d\zeta ^{6}(\eth _{6}\ ^{1}\tilde{\Phi})^{2}/(\ ^{1}\tilde{\Phi}%
)^{3}\ ^{1}\Xi , \\
\ ^{1}w_{\tau } &=&\partial _{\tau }\ ^{1}\phi /\eth _{6}\ ^{1}\phi
=\partial _{\tau }\ ^{1}\Phi /\eth _{6}\ ^{1}\Phi =\partial _{\tau }\
^{1}\Xi /\eth _{6}\ ^{1}\Xi ,
\end{eqnarray*}%
where $\ ^{0}h_{a_{1}}=\ ^{0}h_{a_{1}}(u^{\tau }),$ $\ _{1}^{1}n_{k}(u^{\tau
})$ and $\ _{2}^{1}n_{k}(u^{\tau })$ are integration functions.

A general class of quadratic line elements with one shell jet variables
defining generic off--diagonal solutions of the nonholonomic canonical
deformations of the Einstein equations can be parameterized in the form
{\small
\begin{equation}
ds_{4+2[dK]}^{2}=ds_{4[dK]}^{2}+\frac{\ ^{1}\tilde{\Phi}^{2}}{4\ ^{1}%
\widetilde{\Lambda }}\left[ d\zeta ^{5}+\left( \ _{1}^{1}n_{k}+\ _{2}^{1}%
\widetilde{n}_{k}\int d\zeta ^{6}\frac{(\eth _{6}\ ^{1}\tilde{\Phi})^{2}}{(\
^{1}\tilde{\Phi})^{3}\ ^{1}\Xi }\right) du^{\tau }\right] ^{2}+\ \frac{(\eth
_{6}\ ^{1}\tilde{\Phi})^{2}}{\ ^{1}\Xi }\left[ d\zeta ^{6}+\frac{\partial
_{\tau }\ \ ^{1}\Xi }{\eth _{6}\ \ ^{1}\Xi }du^{\tau }\right] ^{2},
\label{qnk6d}
\end{equation}%
} where $ds_{4dK}^{2}$ is given by the formulae (\ref{qnk4d}) and $\tau =1,2,3,4.$
This quadratic line element carries the  Killing jet symmetry $\eth _{5}$ (in
N--adapted frames, the metric does not depend on $\zeta ^{5}).$

Extending the constructions to the jet shell $s=2$ with $\eth
_{6}\rightarrow \eth _{8},\ ^{1}\phi (u^{\tau },\zeta ^{6})\rightarrow \
^{2}\phi (u^{\tau _{1}},\zeta ^{8}),\ _{1}^{v}\Lambda (u^{\tau },\zeta
^{6})\rightarrow \ _{2}^{v}\Lambda (u^{\tau _{1}},\zeta ^{8})...,$ with$\
^{2}\tilde{\Phi}(u^{\tau _{1}},\zeta ^{8}),\ ^{2}\widetilde{\Lambda }_{0},$
where $\tau _{1}=1,2,...,$ $5,6,$ we generate off--diagonal solutions in
8--d jet modified gravity model, {\small
\begin{equation}
ds_{4+4[dK]}^{2}=ds_{4+2[dK]}^{2}+\frac{\ ^{2}\tilde{\Phi}^{2}}{4\ ^{2}%
\widetilde{\Lambda }}\left[ d\zeta ^{7}+\left( \ _{1}^{2}n_{\tau _{s}}+\
_{2}^{2}\widetilde{n}_{k}\int d\zeta ^{8}\frac{(\eth _{8}\ ^{2}\tilde{\Phi}%
)^{2}}{(\ ^{2}\tilde{\Phi})^{3}\ ^{2}\Xi }\right) du^{\tau _{1}}\right]
^{2}+\ \frac{(\eth _{8}\ ^{2}\tilde{\Phi})^{2}}{\ ^{2}\Xi }\left[ d\zeta
^{8}+\frac{\partial _{\tau _{1}}\ \ ^{2}\Xi }{\eth _{8}\ \ ^{2}\Xi }du^{\tau
_{1}}\right] ^{2},  \label{qnk8d}
\end{equation}%
} where $ds_{4+2[dK]}^{2}$ is given by (\ref{qnk6d}), $\ ^{2}\Xi =\int
d\zeta ^{8}(\ _{2}^{v}\Lambda )\eth _{8}(\ ^{2}\tilde{\Phi}^{2}),$ and the
corresponding integration/generating functions $\ ^{0}h_{a_{2}}(u^{\tau
_{1}});a_{2}=7,8;\ _{1}n_{\tau _{1}}(u^{\tau _{1}})$ and $\ _{2}n_{\tau
_{1}}(u^{\tau _{1}})$ are integration functions.

Using 2+2+... symmetries of the off--diagonal parameterizations (\ref{qnk6d})
and (\ref{qnk8d}), we can construct exact solutions for arbitrary finite
sets of $r$--jet shells on $\ ^{s}\mathbf{V}$ for a nonholonomic $\mathbf{J}%
^{r}(\mathbf{V},\mathbf{V}^{\prime }).$ The corresponding quadratic line
elements are
\begin{eqnarray}
ds_{4+2s[dK]}^{2} &=&ds_{2+2s[dK]}^{2}+\frac{\ ^{s}\tilde{\Phi}^{2}}{4\ ^{s}%
\widetilde{\Lambda }}\left[ d\zeta ^{3+2s}+\left( \ _{1}^{s}n_{\tau
_{s-1}}+\ _{2}^{s}\widetilde{n}_{\tau _{s-1}}\int d\zeta ^{4+2s}\frac{(\eth
_{4+2s}\ ^{s}\tilde{\Phi})^{2}}{(\ ^{s}\tilde{\Phi})^{3}\ ^{s}\Xi }\right)
du^{\tau _{s-1}}\right] ^{2}  \notag \\
&&+\ \frac{(\eth _{4+2s}\ ^{s}\tilde{\Phi})^{2}}{\ ^{s}\Xi }\left[ d\zeta
^{4+2s}+\frac{\partial _{\tau _{s-1}}\ \ ^{s}\Xi }{\eth _{4+2s}\ \ ^{s}\Xi }%
du^{\tau _{s-1}}\right] ^{2},  \label{qnksd}
\end{eqnarray}
where $\ ^{s}\Xi =\int d\zeta ^{4+2s}(\ _{s}^{v}\Lambda )\eth _{4+2s}(\ ^{s}%
\tilde{\Phi}^{2}),$ and the corresponding integration/generating functions;$\
_{1}^{s}n_{\tau _{s-1}}(u^{\tau _{s-1}})$ and $\ _{2}^{s}n_{\tau
_{s-1}}(u^{\tau _{s-1}})$ are also integration functions.

\subsubsection{ The Levi--Civita conditions}

\label{sslc}All solutions constructed in this section and (for vacuum
configurations) in Appendix define subclasses of generic off--diagonal
metrics (\ref{ansk}) for the canonical d--connections $\ ^{s}\widehat{%
\mathbf{D}}$ and nontrivial nonholonomically induced d--torsion coefficients
$\widehat{\mathbf{T}}_{\ \alpha _{s}\beta _{s}}^{\gamma _{s}}\ $ (\ref{dtors}%
). It is natural to have such torsion fields in theories with gauge like jet
symmetries of gravitational and matter fields. Nevertheless, we can perform $%
r $--jet prolongations in such forms that the nonholonomic induced torsion
vanishes for a subclass of nonholonomic distributions with necessary types
of parameterizations of the generating and integration functions and
sources. In explicit form, we construct LC--configurations by imposing
additional constraints,  including certain "shell by shell  jet variables",
on the d--metric and N--connection coefficients. By straightforward
computations (see details in Refs. \cite{vex3}, and Appendix \ref{zt}, for $%
r $--jet variables), we  verify that if in N--adapted frames
\begin{eqnarray}
\mbox{ for }\ \partial _{4}w_{i} &=&\mathbf{e}_{i}\ln \sqrt{|\ h_{4}|},%
\mathbf{e}_{i}\ln \sqrt{|\ h_{3}|}=0,\partial _{i}w_{j}=\partial _{j}w_{i}%
\mbox{ and }\partial _{4}n_{i}=0;  \notag \\
s &=&1:\ \eth _{6}\ ^{1}w_{\alpha }=\ ^{1}\mathbf{e}_{\alpha }\ln \sqrt{|\
h_{6}|},\ ^{1}\mathbf{e}_{\alpha }\ln \sqrt{|\ h_{5}|}=0,\partial _{\alpha
}\ ^{1}w_{\beta }=\partial _{\beta }\ ^{1}w_{\alpha }\mbox{ and }\eth _{6}\
^{1}n_{\gamma }=0;  \label{zerot} \\
s &=&2:\ \eth _{8}\ ^{2}w_{\alpha _{1}}=\ ^{2}\mathbf{e}_{\alpha _{1}}\ln
\sqrt{|\ h_{8}|},\ ^{2}\mathbf{e}_{\alpha _{1}}\ln \sqrt{|\ h_{7}|}=0,\eth
_{\alpha _{1}}\ ^{2}w_{\beta _{1}}=\eth _{\beta _{1}}\ ^{2}w_{\alpha _{1}}%
\mbox{ and }\eth _{8}\ ^{2}n_{\gamma _{1}}=0;  \notag \\
&&...  \notag
\end{eqnarray}%
the torsion coefficients vanish. For $n$--coefficients, such conditions
are satisfied if $\ _{2}n_{k}(x^{i})=0$ and $\partial _{i}\
_{1}n_{j}(x^{k})=\partial _{j}\ _{1}n_{i}(x^{k});\ _{2}^{1}n_{\alpha
}(u^{\beta })=0$ and $\eth _{\gamma }\ _{1}^{1}n_{\tau }(u^{\beta })=\eth
_{\tau }\ _{1}^{1}n_{\gamma }(u^{\beta });\ _{2}^{2}n_{\alpha _{1}}(u^{\beta
_{1}})=0$ and $\eth _{\gamma _{1}}\ _{1}^{2}n_{\tau _{1}}(u^{\beta
_{1}})=\eth _{\tau _{1}}\ _{1}^{2}n_{\gamma _{1}}(u^{\beta _{1}})$ etc. The
explicit form of solutions of constraints on $w_{k}$ derived from (\ref%
{zerot}) depend on the class of vacuum or non--vacuum metrics and their jet
prolongations.

Let us find explicit solutions for the LC--conditions (\ref{zerot}) using the
spacetime coordinates. Such nonholonomic constraints can not be solved in
explicit form for arbitrary data $(\Phi ,\ ^{v}\Lambda ),$ or $(\tilde{\Phi}%
,\Xi ,\ \tilde{\Lambda}_{0}),$ and for all types of nonzero integration
functions $\ _{1}n_{j}(x^{k})$ and $\ _{2}n_{k}(x^{i}).$ We are able to
write such solutions in explicit form if, using coordinate and frame
transforms, we fix $_{2}n_{k}(x^{i})=0$ and $\ _{1}n_{j}(x^{k})=\partial
_{j}n(x^{k})$ for a function $n(x^{k}).$ We use the property that $\mathbf{e}%
_{i}\Phi =(\partial _{i}-w_{i}\partial _{4})\Phi \equiv 0$ for any $\Phi $
if $w_{i}=\partial _{i}\Phi /\partial _{4}\Phi ,$ see (\ref{solhn}). The
equality $\mathbf{e}_{i}H=(\partial _{i}-w_{i}\partial _{4})H=\frac{\partial
H}{\partial \Phi }(\partial _{i}-w_{i}\partial _{4})\Phi \equiv 0$ holds for
any functional $H[\Phi ].$ We can restrict our construction to a subclass of
generating data $(\Phi ,\ ^{v}\Lambda )$ and $(\tilde{\Phi},\Xi ,\ \tilde{%
\Lambda}_{0})$ that are related via formulae (\ref{rescgf}) when $H=\tilde{\Phi}[\Phi
]$ is a functional which allows us to generate LC--configurations in
explicit form. Using $h_{3}[\tilde{\Phi}]=\tilde{\Phi}^{2}/4\widetilde{%
\Lambda }$ (\ref{solha}) for $H=$ $\tilde{\Phi}=\ln \sqrt{|\ h_{3}|}$, we
satisfy the second condition, $\mathbf{e}_{i}\ln \sqrt{|\ h_{3}|}=0,$ in (%
\ref{zerot}).

Next, we solve the first condition in (\ref{zerot}) for spacetime
coordinates. Taking the derivative $\partial _{4}$ of $\ w_{i}=\partial
_{i}\Phi /\partial _{4}\Phi $ (\ref{solhn}), we obtain%
\begin{equation}
\partial _{4}w_{i}=\frac{(\partial _{4}\partial _{i}\Phi )(\partial _{4}\Phi
)-(\partial _{i}\Phi )\partial _{4}\partial _{4}\Phi }{(\partial _{4}\Phi
)^{2}}=\frac{\partial _{4}\partial _{i}\Phi }{\partial _{4}\Phi }-\frac{%
\partial _{i}\Phi }{\partial _{4}\Phi }\frac{\partial _{4}\partial _{4}\Phi
}{\partial _{4}\Phi }.  \label{fder}
\end{equation}%
Choosing a generating function $\Phi =\check{\Phi}$ for which
\begin{equation}
\partial _{4}\partial _{i}\check{\Phi}=\partial _{i}\partial _{4}\check{\Phi}
\label{explcond}
\end{equation}%
and using (\ref{fder}), we compute $\partial _{4}w_{i}=\mathbf{e}_{i}\ln
|\partial _{4}\Phi |.$ Taking $h_{4}[\Phi ,\ ^{v}\Lambda ]$ (\ref{h4aux}),
we write $\mathbf{e}_{i}\ln \sqrt{|\ h_{4}|}=\mathbf{e}_{i}[\ln |\partial
_{4}\Phi |-\ln \sqrt{|\ ^{v}\Lambda |}]$, (see also the conditions (\ref%
{explcond}) and  $\mathbf{e}_{i}\check{\Phi}=0).$ Using the last two
formulae, we  obtain $\partial _{4}w_{i}=\mathbf{e}_{i}\ln \sqrt{|\ h_{4}|%
}$ if $\mathbf{e}_{i}\ln \sqrt{|\ ^{v}\Lambda |}=0$. This is possible for $\
^{v}\Lambda =const,$ or if $\ ^{v}\Lambda $ can be expressed as a functional
$\ ^{v}\Lambda (x^{i},y^{4})=\ ^{v}\Lambda \lbrack \check{\Phi}].$ Here, we
note that the third condition, $\partial _{i}w_{j}=\partial _{j}w_{i},$ in (%
\ref{zerot}), can be solved for any $\check{A}=\check{A}(x^{k},y^{4})$ for
which $w_{i}=\check{w}_{i}=\partial _{i}\check{\Phi}/\partial _{4}\check{\Phi%
}=\partial _{i}\check{A}.$

In shell jet variables, we can extend the above constructions for the  "shell"
generating functions:
\begin{eqnarray}
s=1: &&\ ^{1}\Phi =\ ^{1}\check{\Phi}(u^{\tau },\zeta ^{6}),\eth
_{6}\partial _{\tau }\ ^{1}\check{\Phi}=\partial _{\tau }\eth _{6}\ ^{1}%
\check{\Phi};\ \partial _{\alpha }\ ^{1}\check{\Phi}/\eth _{6}\ ^{1}\check{%
\Phi}=\partial _{\alpha }\ ^{1}\check{A};\ _{1}^{1}n_{\tau }=\partial _{\tau
}\ ^{1}n(u^{\beta });  \label{expconda} \\
s=2: &&\ ^{2}\Phi =\ ^{2}\check{\Phi}(u^{\tau _{1}},\zeta ^{8}),\eth
_{8}\eth _{\tau _{1}}\ ^{2}\check{\Phi}=\eth _{\tau _{1}}\eth _{8}\ ^{2}%
\check{\Phi};\ \eth _{\alpha _{1}}\ ^{2}\check{\Phi}/\eth _{8}\ ^{2}\check{%
\Phi}=\eth _{\alpha _{2}}\ ^{2}\check{A};\ _{1}^{2}n_{\tau _{1}}=\eth _{\tau
_{1}}\ ^{2}n(u^{\beta _{1}}); ...  \notag
\end{eqnarray}%
We can re--define the generating functions as functionals of the  "inverse hat"
values, when
\begin{eqnarray*}
\check{\Phi}^{2} &=&(\widetilde{\Lambda }_{0})^{-1}\int dy^{4}(\ ^{v}\Lambda
)\partial _{4}(\tilde{\Phi}^{2})\mbox{
and }\tilde{\Phi}^{2}=\widetilde{\Lambda }\int dy^{4}(\ ^{v}\Lambda
)^{-1}\partial _{4}(\check{\Phi}^{2}); \\
\ ^{1}\check{\Phi}^{2} &=&(\ ^{1}\widetilde{\Lambda }_{0})^{-1}\int d\zeta
^{6}(\ _{1}^{v}\Lambda )\eth _{6}(\ ^{1}\tilde{\Phi}^{2})\mbox{
and }\ ^{1}\tilde{\Phi}^{2}=\ ^{1}\widetilde{\Lambda }\int d\zeta ^{6}(\
_{1}^{v}\Lambda )^{-1}\eth _{6}(\ ^{1}\check{\Phi}^{2}); \\
\ ^{2}\check{\Phi}^{2} &=&(\ ^{2}\widetilde{\Lambda }_{0})^{-1}\int d\zeta
^{8}(\ _{2}^{v}\Lambda )\eth _{8}(\ ^{2}\tilde{\Phi}^{2})\mbox{
and }\ ^{2}\tilde{\Phi}^{2}=\ ^{2}\widetilde{\Lambda }\int d\zeta ^{8}(\
_{2}^{v}\Lambda )^{-1}\eth _{8}(\ ^{2}\check{\Phi}^{2}),
\end{eqnarray*}%
and compute the values $\Xi (\tilde{\Phi}[\check{\Phi}]),$ $\ ^{1}\Xi (\ ^{1}%
\tilde{\Phi}[\ ^{1}\check{\Phi}])$ and $\ ^{2}\Xi (\ ^{2}\tilde{\Phi}[\ ^{2}%
\check{\Phi}])$ as in (\ref{qnk8d}). This way, we  construct quadratic
line elements for LC--configurations as
\begin{eqnarray}
ds_{4+2s[dK]}^{2} &=&\epsilon _{i}e^{\psi (x^{k})}(dx^{i})^{2}+\frac{\ (%
\tilde{\Phi}[\check{\Phi}])^{2}}{4\widetilde{\Lambda }_{0}}\left[
dy^{3}+(\partial _{i}\ n)dx^{i}\right] ^{2}+\frac{(\partial _{4}\tilde{\Phi}[%
\check{\Phi}])^{2}}{\Xi (\tilde{\Phi}[\check{\Phi}])}\left[ dy^{4}+(\partial
_{i}\ \check{A})dx^{i}\right] ^{2}  \notag \\
&&+\frac{(\ ^{1}\tilde{\Phi}[\ ^{1}\check{\Phi}])^{2}}{4\ ^{1}\widetilde{%
\Lambda }_{0}}\left[ d\zeta ^{5}+(\partial _{\tau }\ ^{1}n)du^{\tau }\right]
^{2}+\ \frac{(\eth _{6}\ ^{1}\tilde{\Phi}[\ ^{1}\check{\Phi}])^{2}}{\
^{1}\Xi (\ ^{1}\tilde{\Phi}[\ ^{1}\check{\Phi}])}\ \left[ d\zeta
^{6}+(\partial _{\tau }\ ^{1}\check{A})du^{\tau }\right] ^{2}  \label{qellcs}
\\
&&+\frac{(\ ^{2}\tilde{\Phi}[\ ^{2}\check{\Phi}])^{2}}{4\ ^{2}\widetilde{%
\Lambda }_{0}}\left[ d\zeta ^{7}+(\eth _{\tau _{1}}\ ^{2}n)du^{\tau _{1}}%
\right] ^{2}+\ \frac{(\eth _{8}\ ^{2}\tilde{\Phi}[\ ^{2}\check{\Phi}])^{2}}{%
\ ^{2}\Xi (\ ^{2}\tilde{\Phi}[\ ^{2}\check{\Phi}])}\ \left[ d\zeta
^{8}+(\eth _{\tau _{1}}\ ^{2}\check{A})du^{\tau _{1}}\right] ^{2} +....
\notag \\
&&+\frac{(\ ^{s}\tilde{\Phi}[\ ^{s}\check{\Phi}])^{2}}{4\ ^{s}\widetilde{%
\Lambda }_{0}}\left[ d\zeta ^{3+2s}+(\eth _{\tau _{s-1}}\ ^{2}n)du^{\tau
_{s-1}}\right] ^{2}+\ \frac{(\eth _{2+2s}\ ^{s}\tilde{\Phi}[\ ^{s}\check{\Phi%
}])^{2}}{\ ^{s}\Xi (\ ^{s}\tilde{\Phi}[\ ^{s}\check{\Phi}])}\ \left[ d\zeta
^{4+2s}+(\eth _{\tau _{s-1}}\ ^{s}\check{A})du^{\tau _{s-1}}\right] ^{2}.
\notag
\end{eqnarray}

The torsions of such non--vacuum exact solutions (\ref{qellcs}) generated by
respective data $(\ ^{s}\mathbf{\check{g},}\ ^{s}\mathbf{\check{N},}\ ^{s}%
\mathbf{\check{\nabla}})$ are zero, which is different from the class of
exact solutions (\ref{qnksd}) with nontrivial canonical d--torsions (\ref%
{dtors}) completely determined by arbitrary data $(\ ^{s}\mathbf{g,}\ ^{s}%
\mathbf{N,}\ ^{s}\widehat{\mathbf{D}})$ with Killing symmetry  $\eth _{7}.$
For an arbitrary shell $s,$ we always have a Killing symmetry  $\eth _{s-1}.$

\subsection{ Violation of Killing symmetries and jet prolongations}

Considering prolongations of 4--d nonholonomic Ricci soliton and Einstein
equations on jet variables we can generate new classes of solutions with
non--Killing symmetries both on spacetime coordinates and on jet shells. On $%
\mathbf{J}^{r}(\mathbf{V},\mathbf{V}^{\prime }),$ there are two general
possibilities to generate "non--Killing" configurations mentioned in Refs.
\cite{vex3,gheorghiu} \ that in this work are generalized for nonholonomic
jet variables: 1) to perform a formal embedding into, for instance, higher
dimension jet prolongation of vacuum spacetimes and/or by 2) "vertical"
conformal nonholonomic deformations, in general, with jet variables.

\subsubsection{Imbedding into a jet prolongation of a vacuum solution}

Let us analyze an example when a subclass of off--diagonal metrics for 6--d
space with jet variables via nonholonomic constraints and
re--parameterizations transform into 4--d non--Killing vacuum solutions. We
consider the geometric case: $\Lambda =\ ^{v}\Lambda =\ _{1}^{v}\Lambda =0;$ $%
h_{3}=\epsilon _{3},h_{5}=\epsilon _{5},n_{k}=0$ and $\ ^{1}n_{\alpha }=0$
with a 2-d $h$--metric $\epsilon _{i}e^{\psi (x^{k},\Lambda
=0)}(dx^{i})^{2}. $ The coefficients of the Ricci d--tensor are zero, see
formulae (\ref{equ1})-(\ref{equ4}) and (\ref{equ5})-(\ref{equ7}). For such
conditions, we can not use  equations (\ref{e1})-(\ref{e4aa}) derived for
$\partial _{4}h_{3}\neq 0,$ $\eth _{6}h_{5}\neq 0$ etc. because such
conditions do not allow, for instance, values $h_{3}=\epsilon
_{3},h_{5}=\epsilon _{5},$ for any nontrivial data $%
h_{4}(x^{i},y^{4}),w_{k}(x^{i},y^{4});$ $h_{6}(x^{i},y^{4},\zeta ^{6}),\
^{1}w_{k}(x^{i},y^{4}),$ $\ ^{1}w_{4}(x^{i},y^{4},\zeta ^{6}).$ Such
functions, depending in general, on spacetime and jet variables, can be
considered as generating functions for vacuum quadratic line elements%
\begin{equation}
ds_{6\rightarrow 4}^{2} = \epsilon _{i}e^{\psi (x^{k},\Lambda
=0)}(dx^{i})^{2}+\epsilon _{3}(dy^{3})^{2}+h_{4}(dy^{4}+w_{k}dx^{k})^{2}
+\epsilon _{5}(d\zeta ^{5})^{2}+h_{6}(d\zeta ^{6}+\ ^{1}w_{k}dx^{k}+\
^{1}w_{4}dy^{4})^{2}  \label{6to4}
\end{equation}%
on the first 2--d jet shell on $\mathbf{J}^{r}(\mathbf{V},\mathbf{V}^{\prime }).$
This class of vacuum 6-d metrics with two jet variables are with nonzero
nonholonomically induced d--torsion (\ref{dtors}). Such solutions can not be
considered as a subclass of vacuum solutions (\ref{qe6dvacuum}) when $%
h_{3}\rightarrow \epsilon _{3}$ and $h_{5}\rightarrow \epsilon _{5}$ because
the conditions $\partial _{4}h_{3}\neq 0$ and $\eth _{6}h_{5}\neq 0$ impose
additional constraints on the class of possible generating functions $h_{4}$
and $h_{6}.$ By fixing from the very beginning certain configurations with $%
\partial _{4}h_{3}=0$ and $\eth _{6}h_{5}=0,$ we can consider the values $%
h_{4},h_{6}$ and $w_{k},\ ^{1}w_{k},\ ^{1}w_{4}$ as independent generating
functions.

We generate LC--configurations if the coefficients of the d--metric (\ref%
{6to4}) are subject to additional  constraints (\ref{zerot}) up to $%
s=1.$ We can follow a formal procedure which is similar to that outlined in
section \ref{sslc}. For any constant $h_{3}=\epsilon _{3}$ and $%
h_{5}=\epsilon _{5},$ the conditions $\mathbf{e}_{i}\ln \sqrt{|\ h_{3}|}=0$
and $\ ^{1}\mathbf{e}_{\alpha }\ln \sqrt{|\ h_{5}|}=0$ are satisfied. The
class of generating functions can be restricted to solve the equations
\begin{eqnarray}
\partial _{4}w_{i}(x^{i},y^{4}) &=&\mathbf{e}_{i}\ln \sqrt{|\
h_{4}(x^{i},y^{4})|},\partial _{i}w_{j}=\partial _{j}w_{i},\mbox{ and }\
\label{zerota} \\
\eth _{6}\ ^{1}w_{\alpha }(x^{i},y^{4},\zeta ^{6}) &=&\ ^{1}\mathbf{e}%
_{\alpha }\ln \sqrt{|\ h_{6}(x^{i},y^{4},\zeta ^{6})|},\partial _{\alpha }\
^{1}w_{\beta }=\partial _{\beta }\ ^{1}w_{\alpha },  \notag
\end{eqnarray}%
Such equations do not depend on the spacetime coordinate $y^{3}$ and on jet
variable $\zeta ^{5}.$ By prescribing any values of $\ h_{4}$ and $\ h_{6}$ we
can find LC--admissible $w$--coefficients solving the system of first order
partial derivative equations in (\ref{zerota}). In general, such solutions
are defined for certain nonholonomic constraints, i.e. in "non--explicit"
form. If the respective d--metric and N--connection coefficients $h_{4}[%
\check{\Phi}],h_{6}[\ ^{1}\check{\Phi}]$ and $w_{k}[\check{\Phi}],\
^{1}w_{k}[\ ^{1}\check{\Phi}],\ ^{1}w_{4}[\ ^{1}\check{\Phi}]$ are
determined by $\check{\Phi}(x^{i},y^{4})$ and $\ ^{1}\check{\Phi}%
(x^{i},y^{4},\zeta ^{6})$ and satisfy conditions (\ref{explcond}), (\ref%
{expconda}) (for such configurations, $h_{3}$ and $h_{5}$ may be not
functionals of type (\ref{solha})), then we can solve  equations (\ref{zerota}%
) in explicit form.

By choosing any generating function $\check{\Phi}$ or $\ ^{1}\check{\Phi}$
and functionals $h_{4}[\check{\Phi}],h_{6}[\ ^{1}\check{\Phi}]$ we compute%
\begin{equation}
w_{i} =\check{w}_{i}=\partial _{i}\check{\Phi}/\partial _{4}\check{\Phi}%
=\partial _{i}\check{A}\mbox{ and } \ ^{1}w_{i} = \ ^{1}\check{w}%
_{i}=\partial _{i}\ ^{1}\check{\Phi}/\eth _{6}\ ^{1}\check{\Phi}=\partial
_{i}\ ^{1}\check{A},\ ^{1}w_{4}=\ ^{1}\check{w}_{4}=\partial _{4}\ ^{1}%
\check{\Phi}/\eth _{6}\ ^{1}\check{\Phi}=\partial _{4}\ ^{1}\check{A},
\label{data4c}
\end{equation}%
for some $\check{A}(x^{i},y^{4})$ and $\ ^{1}\check{A}(x^{i},y^{4},\zeta
^{6})$ which are necessary to satisfy the equalities $\partial
_{i}w_{j}=\partial _{j}w_{i}$ and $\partial _{\alpha }\ ^{1}w_{\beta
}=\partial _{\beta }\ ^{1}w_{\alpha }.$ Applying the functional derivatives
of type (\ref{fder}) and N--coefficients of type (\ref{data4c}) when $H[%
\check{\Phi}]=\ln \sqrt{|\ h_{4}|}$ and $\ ^{1}H[\ ^{1}\check{\Phi}]=\ln
\sqrt{|\ h_{6}|},$ we can satisfy the LC--conditions (\ref{zerota}).

The constructions from the last two paragraphs allow to define a subclass of
metrics of (\ref{6to4}) determined by generic off--diagonal metrics as
solutions of 6--d vacuum Einstein equations with two jet variables from the
first shell,
\begin{equation}
ds_{6\rightarrow 4}^{2} = \epsilon _{i}e^{\psi (x^{k},\Lambda
=0)}(dx^{i})^{2}+\epsilon _{3}(dy^{3})^{2}+h_{4}[\check{\Phi}%
](dy^{4}+\partial _{k}\check{A}dx^{k})^{2} +\epsilon _{5}(d\zeta
^{5})^{2}+h_{6}[\ ^{1}\check{\Phi}](d\zeta ^{6}+\partial _{k}\ ^{1}\check{A}%
\ dx^{k}+\partial _{4}\ ^{1}\check{A}\ dy^{4})^{2}.  \label{6to4lc}
\end{equation}%
The terms $\epsilon _{3}(dy^{3})^{2}$ and $\epsilon _{5}(dy^{5})^{2}$ are
for trivial extensions from 4-d to 6--d configurations but imbedded in a
nontrivial form in a jet extra dimensional vacuum background. Re--defining
the coordinate $\zeta ^{6}\rightarrow y^{3},$ we generate vacuum solutions
in 4--d gravity with metrics (\ref{6to4lc}) depending on all four
coordinates $x^{i},y^{3}$ and $y^{4}.$ This way we mimic certain 4-d
gravitational interactions on a jet prolongation of a 3--d spacetime
manifold. Finally, we note that the nonholonomy coefficients (\ref{anhrel1})
are not zero and that such metrics can not be diagonalized by coordinate or
jet coordinates transformations. This class of 4--d vacuum spacetimes do not
possess, in general, any Killing symmetries.

\subsubsection{"Vertical" nonholonomic conformal and jet deformations}

We briefly touch upon  another possibility to generate off--diagonal solutions
depending on all spacetime coordinates and, in general, with nontrivial
sources of type (\ref{sourc1}) \cite{vex3}. To work with jet type variable the
 formal re--definition of extra dimension coordinates into
nonholonomic shell jet coordinates is necessary. By straightforward but tedious
computations, we can prove

\begin{corollary}
Any metric {\small
\begin{equation}
\mathbf{g} = g_{i}(x^{k})dx^{i}\otimes dx^{i}+\omega ^{2}(u^{\alpha
})h_{a}(x^{k},y^{4})\mathbf{e}^{a}\otimes \mathbf{e}^{a}+ \ ^{1}\omega
^{2}(u^{\alpha _{1}})h_{a_{1}}(u^{\alpha },\zeta ^{6})\mathbf{e}%
^{a_{1}}\otimes \mathbf{e}^{a_{1}}+...+\ ^{s}\omega ^{2}(u^{\alpha
_{s}})h_{a_{s}}(u^{\alpha _{s-1}},\zeta ^{4+2s})\mathbf{e}^{a_{s}}\otimes
\mathbf{e}^{a_{s}},  \label{ans1}
\end{equation}%
} with the conformal $v$--factors subject to the conditions
\begin{eqnarray}
\mathbf{e}_{k}\omega &=&\partial _{k}\omega +n_{k}\partial _{3}\omega
+w_{k}\partial _{4}\omega =0,  \label{vconfc} \\
\ ^{1}\mathbf{e}_{\beta }\ ^{1}\omega &=&\partial _{\beta }\ ^{1}\omega +\
^{1}n_{\beta }\eth _{5}\ ^{1}\omega +\ ^{1}w_{\beta }\eth _{6}\ ^{1}\omega
=0, \ ^{2}\mathbf{e}_{\beta _{1}}\ ^{2}\omega = \partial _{\beta _{1}}\
^{2}\omega +\ ^{2}n_{\beta _{1}}\eth _{7}\ ^{2}\omega +\ ^{2}w_{\beta
_{1}}\eth _{8}\ ^{2}\omega =0, ...  \notag
\end{eqnarray}%
does not change the Ricci d--tensor (\ref{equ1})--(\ref{equ4d2s}).
\end{corollary}

As a result of this Corollary, any class of solutions considered in this
section can be generalized to non--Killing configurations using "vertical"
nonholonomic conformal and jet transformations and deformations.

\section{Nonholonomic Jet Prolongations of the Kerr Metric and Ricci Solitons%
}

\label{s4} In this section, we study nonholonomic off--diagonal and/or jet
deformations of the Kerr black hole solution. The approach develops the results
from section 4 of Ref. \cite{gheorghiu} for jet variables and Ricci soliton
configurations when the constructions for massive gravity are re--considered
for jet modified gravity theories. A series of new class of exact
solutions when the metrics are nonholonomically deformed into general or
ellipsoidal stationary configurations in four dimensional gravity with Ricci
soliton correction and/or extra dimensions treated as jet variables. We cite
the monographs \cite{heusler,kramer,misner} for the standard methods and
bibliography on stationary black holes.

\subsection{N--adapted parameterizations of the Kerr vacuum solution}

A 4-d ansatz%
\begin{equation*}
ds_{[0]}^{2}=Y^{-1}e^{2h}(d\rho ^{2}+dz^{2})-\rho
^{2}Y^{-1}dt^{2}+Y(d\varphi +Adt)^{2}
\end{equation*}%
parameterized in terms of three functions $(h,Y,A)$ on coordinates $(\rho
,z) $ defines the Kerr solution of the vacuum Einstein equations (for
rotating black holes) if we choose {\small
\begin{equation*}
Y =\frac{1-(p\widehat{x}_{1})^{2}-(q\widehat{x}_{2})^{2}}{(1+p\widehat{x}%
_{1})^{2}+(q\widehat{x}_{2})^{2}},\ A=2M\frac{q}{p}\frac{(1-\widehat{x}%
_{2})(1+p\widehat{x}_{1})}{1-(p\widehat{x}_{1})-(q\widehat{x}_{2})},\ e^{2h}
= \frac{1-(p\widehat{x}_{1})^{2}-(q\widehat{x}_{2})^{2}}{p^{2}[(\widehat{x}%
_{1})^{2}+(\widehat{x}_{2})^{2}]},\ \rho ^{2}=M^{2}(\widehat{x}_{1}^{2}-1)(1-%
\widehat{x}_{2}^{2}),\ z=M\widehat{x}_{1}\widehat{x}_{2},
\end{equation*}%
} where $M=const$ and $\rho =0$ states the horizon $\widehat{x}_{1}=0$ with
the "north / south" segment of the rotation axis, $\widehat{x}_{2}=+1/-1.$
For our purposes, such a metric is written in the form
\begin{equation}
ds_{[0]}^{2}=(dx^{1})^{2}+(dx^{2})^{2}-\rho ^{2}Y^{-1}(\mathbf{e}^{3})^{2}+Y(%
\mathbf{e}^{4})^{2},  \label{kerr1}
\end{equation}%
with some coordinates $x^{1}(\widehat{x}_{1},\widehat{x}_{2})$ and $x^{2}(%
\widehat{x}_{1},\widehat{x}_{2})$, when $%
(dx^{1})^{2}+(dx^{2})^{2}=M^{2}e^{2h}(\widehat{x}_{1}^{2}-\widehat{x}%
_{2}^{2})Y^{-1}\left( \frac{d\widehat{x}_{1}^{2}}{\widehat{x}_{1}^{2}-1}+%
\frac{d\widehat{x}_{2}^{2}}{1-\widehat{x}_{2}^{2}}\right)$ and $y^{3}=t+%
\widehat{y}^{3}(x^{1},x^{2}),y^{4}=\varphi +\widehat{y}^{4}(x^{1},x^{2},t).$
We write $\mathbf{e}^{3}=dt+(\partial _{i}\widehat{y}^{3})dx^{i},\mathbf{e}%
^{4}=dy^{4}+(\partial _{i}\widehat{y}^{4})dx^{i}$, for some functions $%
\widehat{y}^{a},$ $a=3,4,$ with $\partial _{t}\widehat{y}^{4}=-A(x^{k}).$

The Boyer--Linquist coordinates for the Kerr metric are introduced as $%
(r,\vartheta ,\varphi ,t),$ where $r=m_{0}(1+p\widehat{x}_{1}),\widehat{x}%
_{2}=\cos \vartheta .$ The parameters $p,q$ are related to the total black
hole mass, $m_{0}$ and the total angular momentum, $am_{0},$ for the
asymptotically flat, stationary and axisymmetric Kerr spacetime. The
formulae $m_{0}=Mp^{-1}$ and $a=Mqp^{-1}$ when $p^{2}+q^{2}=1$ imply $%
m_{0}^{2}-a^{2}=M^{2}.$ In terms of these variables, the metric (\ref{kerr1}) is
written%
\begin{eqnarray}
ds_{[0]}^{2} &=&(dx^{1^{\prime }})^{2}+(dx^{2^{\prime }})^{2}+\overline{A}(%
\mathbf{e}^{3^{\prime }})^{2}+(\overline{C}-\overline{B}^{2}/\overline{A})(%
\mathbf{e}^{4^{\prime }})^{2},  \label{kerrbl} \\
\mathbf{e}^{3^{\prime }} &=&dt+d\varphi \overline{B}/\overline{A}%
=dy^{3^{\prime }}-\partial _{i^{\prime }}(\widehat{y}^{3^{\prime }}+\varphi
\overline{B}/\overline{A})dx^{i^{\prime }},\mathbf{e}^{4^{\prime
}}=dy^{4^{\prime }}=d\varphi .  \notag
\end{eqnarray}%
In these quadratic elements, we consider coordinate functions $x^{1^{\prime
}}(r,\vartheta ),\ x^{2^{\prime }}(r,\vartheta ),\ y^{3^{\prime }}=t+%
\widehat{y}^{3^{\prime }}(r,\vartheta ,\varphi )+\varphi \overline{B}/%
\overline{A},y^{4^{\prime }}=\varphi ,\ \partial _{\varphi }\widehat{y}%
^{3^{\prime }}=-\overline{B}/\overline{A}$, for which $(dx^{1^{\prime
}})^{2}+(dx^{2^{\prime }})^{2}=\Xi \left( \Delta ^{-1}dr^{2}+d\vartheta
^{2}\right) ,$ when the coefficients are
\begin{eqnarray}
\overline{A} &=&-\Xi ^{-1}(\Delta -a^{2}\sin ^{2}\vartheta ),\overline{B}%
=\Xi ^{-1}a\sin ^{2}\vartheta \left[ \Delta -(r^{2}+a^{2})\right] ,  \notag
\\
\overline{C} &=&\Xi ^{-1}\sin ^{2}\vartheta \left[ (r^{2}+a^{2})^{2}-\Delta
a^{2}\sin ^{2}\vartheta \right] ,\mbox{ and } \Delta = r^{2}-2m_{0}+a^{2},\
\Xi =r^{2}+a^{2}\cos ^{2}\vartheta .  \label{kerrcoef}
\end{eqnarray}

We consider the  prime data
\begin{eqnarray}
&& \mathring{g}_{1}=1,\mathring{g}_{2}=1,\mathring{h}_{3}=-\rho ^{2}Y^{-1},%
\mathring{h}_{4}=Y,\mathring{N}_{i}^{a}=\partial _{i}\widehat{y}^{a},
\label{dkerr} \\
\mbox{ i.e. }\mathring{g}_{1^{\prime }} &=&1,\mathring{g}_{2^{\prime }}=1,%
\mathring{h}_{3^{\prime }}=\overline{A},\mathring{h}_{4^{\prime }}=\overline{%
C}-\overline{B}^{2}/\overline{A},  \notag \\
\mathring{N}_{i^{\prime }}^{3} &=&\mathring{n}_{i^{\prime }}=-\partial
_{i^{\prime }}(\widehat{y}^{3^{\prime }}+\varphi \overline{B}/\overline{A}),%
\mathring{N}_{i^{\prime }}^{4}=\mathring{w}_{i^{\prime }}=0  \notag
\end{eqnarray}%
for the quadratic linear elements (\ref{kerr1}), or (\ref{kerrbl}), which
define exact solutions with rotational spherical symmetry of the vacuum
Einstein equations parameterized in the form (\ref{cdeinst}) and (\ref%
{lcconstr}) with zero sources. The Kerr vacuum solution in 4-d GR consists of a
"degenerate" case of 4--d off--diagonal vacuum solutions determined by
primary metrics with data (\ref{dkerr}) when the diagonal coefficients
depend only on two "horizontal" N--adapted coordinates and the off--diagonal
terms are induced by rotating frames.

\subsection{Deformations of Kerr metrics by an effective Ricci soliton source%
}

Let us consider the coefficients (\ref{dkerr}) for the Kerr metric as the
data for a prime metric $\mathbf{\mathring{g}.}$ Our goal is to study
nonholonomic off--diagonal deformations of the Kerr solution into a Ricci
soliton configuration, i.e. when the vacuum Einstein equations are modified
by a Ricci soliton, with
\begin{equation*}
(\mathbf{\mathring{g}},\mathbf{\mathring{N},\ }^{v}\mathring{\Upsilon}=0,%
\mathring{\Upsilon}=0)\rightarrow (\widetilde{\mathbf{g}},\widetilde{\mathbf{%
N}}\mathbf{,\ }^{v}\widetilde{\Upsilon }=\Lambda (x^{k}),\ ^{h}\widetilde{%
\Upsilon }=\ ^{v}\Lambda (x^{k},y^{4})),\widetilde{\Lambda }_{0}=const\neq 0,
\end{equation*}%
where the target source (\ref{dsource}) is parameterized as $\widetilde{%
\Upsilon }_{1}^{1}=\widetilde{\Upsilon }_{2}^{2}=\ ^{h}\widetilde{\Upsilon }%
=\ ^{v}\Lambda (x^{k},y^{4})$ and $\Upsilon _{3}^{3}=\Upsilon _{4}^{4}=\
^{v}\Upsilon =\Lambda (x^{k})$ and encode contributions of gradient function
$\kappa $ and the constant $\lambda $ from the Ricci soliton equations (\ref%
{nriccisol}) into solutions of equations (\ref{geq1})--(\ref{geq3}). The
target metric $\widetilde{\mathbf{g}}$ is constrained to define a generic
off--diagonal solution of the field equations with effective horizontal ($h$)- and vertical ($v$)%
--polarized gravitational constants. In some sense, the Ricci soliton
contributions may induce a mass term of the type $\widetilde{\Lambda }_{0}=\mu
_{g}^{2}\ \widetilde{\lambda },$ like the one considered in \cite{gheorghiu}, for
respective parameterizations. The N--adapted deformations of coefficients of
metrics and frames are written as
\begin{equation}
\lbrack \mathring{g}_{i},\mathring{h}_{a},\mathring{w}_{i},\mathring{n}%
_{i}]\rightarrow \lbrack \widetilde{g}_{i}=\widetilde{\eta }_{i}\mathring{g}%
_{i},\widetilde{h}_{3}=\widetilde{\eta }_{3}\mathring{h}_{3},\widetilde{h}%
_{4}=\widetilde{\eta }_{4}\mathring{h}_{4},\widetilde{w}_{i}=\mathring{w}%
_{i}+\ ^{\eta }w_{i},n_{i}=\mathring{n}_{i}+\ ^{\eta }n_{i}],  \notag
\end{equation}%
where the values \ $\widetilde{\eta }_{a},\widetilde{w}_{i},\tilde{n}_{i}$
and $\varpi $ are functions of three coordinates $(x^{k^{\prime
}},y^{4}=\varphi )$ and $\widetilde{\eta }_{i}(x^{k})$ and depend only on
h--coordinates $x^{k}.$ The prime data $\mathring{g}_{i},\mathring{h}_{a},%
\mathring{w}_{i},\mathring{n}_{i}$ for a Kerr metric are given by
coefficients depending only on $(x^{k}).$ The quadratic line elements,  determined
by target solutions of type (\ref{qnk4d}), are paremeterized in the form
\begin{equation}
ds_{4[dK]}^{2}=e^{\psi (x^{k^{\prime }})}[(dx^{1^{\prime
}})^{2}+(dx^{2^{\prime }})^{2}]-\frac{\tilde{\Phi}^{2}}{4\widetilde{\Lambda }%
_{0}}\left[ dy^{3}+\left( \ _{1}n_{k}+_{2}\widetilde{n}_{k}\int d\varphi
\frac{(\partial _{\varphi }\tilde{\Phi})^{2}}{\tilde{\Phi}^{3}\Xi }\right)
dx^{k}\right] ^{2}+\frac{(\partial _{4}\tilde{\Phi})^{2}}{\Xi }\ \left[
d\varphi +\frac{\partial _{i}\Xi }{\partial _{\varphi }\Xi }dx^{i}\right]
^{2},  \notag
\end{equation}%
where $\Xi \lbrack \ ^{v}\Lambda ,\tilde{\Phi}]=\int d\varphi (\ ^{v}\Lambda
)\partial _{\varphi }(\tilde{\Phi}^{2}).$

In terms of $\eta $--functions (\ref{etad}) giving $h_{a}^{\ast }\neq 0,$ $%
g_{i}=c_{i}e^{\psi {(x^{k^{\prime }})}}$ and LC--configurations, the
solutions of type (\ref{qnk4d}) with an effective cosmological constant $%
\widetilde{\Lambda }_{0}$ induced by off--diagonal Ricci soliton
configurations and $\ _{2}n_{k^{\prime }}=0$ can be re--written in the form%
\begin{eqnarray}
ds^{2} &=&e^{\psi (x^{k^{\prime }})}[(dx^{1^{\prime }})^{2}+(dx^{2^{\prime
}})^{2}]-  \label{nvlcmgs} \\
&&\ \widetilde{\eta }_{3^{\prime }}\overline{A}[dy^{3^{\prime }}+\left(
\partial _{k^{\prime }}\ ^{\eta }n(x^{i^{\prime }})-\partial _{k^{\prime }}(%
\widehat{y}^{3^{\prime }}+\varphi \overline{B}/\overline{A})\right)
dx^{k^{\prime }}]^{2}+\ \widetilde{\eta }_{4^{\prime }}(\overline{C}-%
\overline{B}^{2}/\overline{A})[d\varphi +(\partial _{i^{\prime }}\ ^{\eta }%
\widetilde{A})dx^{i^{\prime }}]^{2},  \notag
\end{eqnarray}%
where  use is made of  "primed" coordinates and prime Kerr data (\ref{kerrbl})
and (\ref{dkerr}). The gravitational polarizations $(\eta _{i},\eta _{a})$
and N--coefficients $(n_{i},w_{i})$ are computed\
\begin{eqnarray}
e^{\psi (x^{k})} &=&\widetilde{\eta }_{1^{\prime }}=\widetilde{\eta }%
_{2^{\prime }},\ \widetilde{\eta }_{3^{\prime }}=\tilde{\Phi}^{2}/4%
\widetilde{\Lambda }_{0}\overline{A},\ \widetilde{\eta }_{4^{\prime
}}=(\partial _{\varphi }\tilde{\Phi})^{2}/\Xi (\overline{C}-\overline{B}^{2}/%
\overline{A}),  \label{polarkerr} \\
w_{i^{\prime }} &=&\mathring{w}_{i^{\prime }}+\ ^{\eta }w_{i^{\prime
}}=\partial _{i^{\prime }}(\ ^{\eta }\widetilde{A}[\tilde{\Phi}]),\
n_{k^{\prime }}=\mathring{n}_{k^{\prime }}+\ ^{\eta }n_{k^{\prime
}}=\partial _{k^{\prime }}(-\widehat{y}^{3^{\prime }}+\varphi \overline{B}/%
\overline{A}+\ ^{\eta }n),  \notag
\end{eqnarray}%
where $\ ^{\eta }\widetilde{A}(x^{k},\varphi )$ is introduced via formulae
and assumptions similar to (\ref{expconda}), for $s=1,$ and $\psi (x^{k})$
is a solution of 2--d Poisson equation,
\begin{equation*}
\partial _{11}^{2}\psi +\partial _{22}^{2}\psi =2\ \Lambda (x^{k^{\prime }}).
\end{equation*}%
To extract LC--configurations,  the parameterizations (\ref%
{solhn}) are made use of when $\mathring{h}_{3^{\prime }}\mathring{h}_{4^{\prime }}=%
\overline{A}\overline{C}-\overline{B}^{2}$ and the N--coefficients are
computed as
\begin{equation*}
w_{i^{\prime }}=\mathring{w}_{i^{\prime }}+\ ^{\eta }w_{i^{\prime
}}=\partial _{i^{\prime }}(\ \tilde{\Phi}\sqrt{|\overline{A}\overline{C}-%
\overline{B}^{2}|})/\ \partial _{\varphi }\tilde{\Phi}\sqrt{|\overline{A}%
\overline{C}-\overline{B}^{2}|}=\partial _{i^{\prime }}\ ^{\eta }\widetilde{A%
}
\end{equation*}%
for $\ _{1}n_{i^{\prime }}=\partial _{i^{\prime }}$ $^{\eta }n(x^{k})$
computed for an arbitrary function $^{\eta }n(x^{k}).$

\begin{theorem}
\label{th4.1}Quadratic elements (\ref{nvlcmgs}) define nonholonomic
deformations of a prime Kerr solution $[\mathring{g}_{i},\mathring{h}_{a},%
\mathring{w}_{i},\mathring{n}_{i}]$ (\ref{dkerr}) into target Ricci soliton
LC--configurations with Killing symmetry $\partial /\partial \widehat{y}%
^{3^{\prime }}$ determined by polarization functions (\ref{polarkerr})
generated by data $\left[ \psi (x^{k^{\prime }}),\ \widetilde{\eta }%
_{3^{\prime }}(x^{k^{\prime }},\varphi ),\ ^{\eta }\widetilde{A}(\
\widetilde{\eta }_{3^{\prime }}),^{\eta }n(x^{k}),\ ^{v}\Lambda
(x^{k^{\prime }},\varphi ),\widetilde{\Lambda }_{0}\right] .$
\end{theorem}

\begin{proof}
Let us show that $\ \widetilde{\eta }_{4^{\prime }}$ can be defined by $%
\widetilde{\eta }_{3^{\prime }}$ which can be considered as a generating
function instead of $\tilde{\Phi}.$ Considering the second formula in (\ref%
{polarkerr}), we express%
\begin{equation*}
\tilde{\Phi}^{2}=4\widetilde{\Lambda }_{0}\overline{A}\ \widetilde{\eta }%
_{3^{\prime }}
\end{equation*}
and compute $\Xi =4\widetilde{\Lambda }_{0}\overline{A}\int d\varphi (\
^{v}\Lambda )\partial _{\varphi }(\widetilde{\eta }_{3^{\prime }}).$ We
introduce these formulae into the third formula in (\ref{polarkerr}) and derive%
\begin{equation*}
\ \widetilde{\eta }_{4^{\prime }}=\overline{A}\left( \partial _{\varphi }%
\sqrt{|\widetilde{\eta }_{3^{\prime }}|}\right) ^{2}/(\overline{A}\overline{C%
}-\overline{B}^{2})\int d\varphi (\ ^{v}\Lambda )\partial _{\varphi }(%
\widetilde{\eta }_{3^{\prime }}).
\end{equation*}%
It follows that, by prescribing any polarization function $\widetilde{\eta }_{3^{\prime
}}(x^{k^{\prime }},\varphi )$ and $v$--source $\ ^{v}\Lambda (x^{k^{\prime
}},\varphi ),$ we can compute $\ \widetilde{\eta }_{4^{\prime }}.$ The
polarizations $\widetilde{\eta }_{1^{\prime }}=\widetilde{\eta }_{2^{\prime
}}$ are determined by function $\psi (x^{k^{\prime }}),$ i.e. by source $%
\Lambda (x^{k^{\prime }}).$ Finally, by prescribing any functional $\ ^{\eta }%
\widetilde{A}(\ \widetilde{\eta }_{3^{\prime }})$ and function $^{\eta
}n(x^{k})$ we can compute the N--connection coefficients for any fixed
effective cosmological constant $\widetilde{\Lambda }_{0}.$

$\square $
\end{proof}

\vskip5pt

The solutions (\ref{nvlcmgs}) are for stationary LC--configurations,
generated canonically as off--diagonal Ricci solitons from Kerr black holes
when the new class of spacetimes carry Killing symmetry $\partial
/\partial y^{3^{\prime }}$ and generic dependence on three (from maximally
four) coordinates, $(x^{i^{\prime }}(r,\vartheta ),\varphi ).$ Off--diagonal
modifications are possible even for very small values of the effective
cosmological constant which can mimic gravitational effects determined by a
gravitational mass parameter $\ \mu _{g}.$

\subsubsection{Nonholonomically induced torsion and Ricci soliton modified
gravity}

If we do not impose the LC--conditions (\ref{lcconstr}), a nontrivial source
$\ ^{v}\Lambda (x^{k^{\prime }},\varphi )$ induces stationary configuration
with nontrivial d--torsion (\ref{dtors}). For simplicity, we can study
nonholonomic torsion effects for a $v$--source not depending on the coordinate $%
\varphi ,\ $i.e. for $^{v}\Lambda (x^{k^{\prime }})$ The torsion
coefficients are determined by metrics of the type (\ref{qnk4d}) with nontrivial
$\widetilde{\Lambda }_{0}$ and certain parameterizations of coefficients of
an associated N--connection, canonical d--torson and coordinates
distinguishing the prime data for a Kerr metric (\ref{dkerr}). The
corresponding quadratic elements can be written in the form {\small
\begin{eqnarray}
ds^{2} &=&e^{\psi (x^{k^{\prime }})}[(dx^{1^{\prime }})^{2}+(dx^{2^{\prime
}})^{2}]-\frac{\Phi ^{2}}{4|\widetilde{\Lambda }_{0}|}\overline{A}%
[dy^{3^{\prime }}+\left( \ _{1}n_{k^{\prime }}(x^{i^{\prime }})+\
_{2}n_{k^{\prime }}(x^{i^{\prime }})\frac{(\partial _{\varphi }\Phi )^{2}}{%
\Phi ^{5}}-\partial _{k^{\prime }}(\widehat{y}^{3^{\prime }}+\varphi
\overline{B}/\overline{A})\right) dx^{k^{\prime }}]^{2}  \notag \\
&&+\frac{(\partial _{\varphi }\Phi )^{2}}{\ ^{v}\Lambda (x^{k^{\prime
}})\Phi ^{2}}(\overline{C}-\overline{B}^{2}/\overline{A})[d\varphi +\frac{%
\partial _{i^{\prime }}\Phi }{\partial _{\varphi }\Phi }dx^{i^{\prime
}}]^{2},  \label{ofindtmg}
\end{eqnarray}%
} where nonzero values of $\ _{2}n_{k}(x^{i^{\prime }})$ are considered. We
can see that Ricci soliton effects may give nontrivial stationary
off--diagonal torsion effects if the integration function $\ _{2}n_{k}\neq
0. $ Considering two different classes of off--diagonal solutions (\ref%
{ofindtmg}) and (\ref{nvlcmgs}), we can study the issue if a Ricci modified
gravity theory carries induced torsion or is characterized by additional
nonholonomic constraints as in GR (giving zero torsion).

It should be noted that configurations of type (\ref{ofindtmg}) can be
constructed in various theories with noncommutative, brane,
extra--dimension, warped and trapped brane type variables in sting, or
Finsler like and/or Ho\v{r}ava--Lifshits theories \cite%
{vexsol1,vex3,vexsol2,gheorghiu} when nonholonomically induced torsion
effects are significant/non-vanishing.

\subsubsection{Small Ricci soliton modifications of Kerr metrics and
modelling modified and massive gravity}

We can construct off--diagonal solutions for superposition of Ricci soliton
effects and $f$--modified and massive gravity interactions, see original
contributions and reviews of results in Refs. \cite%
{capoz,odints1,odints2,drg1,drg3,hr1,hr2,kour,stavr,mavr,gheorghiu}. Small
nonlinear effects and modifications can be distinguished in explicit form if
we take into account additional $f$--deformations, for instance, a "prime"
solution for massive gravity/ effective modeled in GR with source $\ ^{\mu
}\Lambda =\mu _{g}^{2}\ \lambda (x^{k^{\prime }}),$ or re--defined to $\
^{\mu }\tilde{\Lambda}=\mu _{g}^{2}\ \tilde{\lambda}=const.$ By adding a
"small" value $\ \widetilde{\Lambda }$ as determined by $f$--modifications, we
work in N--adapted frames with an effective source $\Upsilon =\widetilde{%
\Lambda }+\widetilde{\lambda }.$ We construct a class of off--diagonal
solutions in modified $f$--gravity generated from the Kerr black hole
solution as a result of two nonholonomic deformations
\begin{equation*}
(\mathbf{\mathring{g}},\mathbf{\mathring{N},\ }^{v}\mathring{\Upsilon}=0,%
\mathring{\Upsilon}=0)\rightarrow (\widetilde{\mathbf{g}},\widetilde{\mathbf{%
N}},\ ^{v}\widetilde{\Upsilon }=\widetilde{\lambda },\widetilde{\Upsilon }=%
\widetilde{\lambda })\rightarrow (\ ^{\varepsilon }\mathbf{g},\
^{\varepsilon }\mathbf{N,\Upsilon =\varepsilon \ }\widetilde{\Lambda }+\
^{\mu }\tilde{\Lambda},\ ^{v}\mathbf{\Upsilon =\varepsilon \ }\widetilde{%
\Lambda }+\ ^{\mu }\tilde{\Lambda}),
\end{equation*}%
when the target data $\mathbf{g=}\ ^{\varepsilon }\mathbf{g}$ and$\ \mathbf{%
N=}\ ^{\varepsilon }\mathbf{N}$ depend on a small parameter $\varepsilon ,$ $%
0<\varepsilon \ll 1.$ For simplicity, we construct generic off--diagonal
solutions with $|\mathbf{\varepsilon \ }\widetilde{\Lambda }|\ll |\ ^{\mu }%
\tilde{\Lambda}|,$ when $f$--modifications in N--adapted frames are much
smaller than massive gravity effects. A similar analysis for nonlinear
interactions with $|\mathbf{\varepsilon \ }\widetilde{\Lambda }|\gg |\ ^{\mu
}\tilde{\Lambda}|)$ is omitted. The corresponding N--adapted transforms are
parameterized as
\begin{eqnarray}
&&[\mathring{g}_{i},\mathring{h}_{a},\mathring{w}_{i},\mathring{n}%
_{i}]\rightarrow   \label{def2} \\
&&[g_{i}=(1+\varepsilon \chi _{i})\widetilde{\eta }_{i}\mathring{g}%
_{i},h_{3}=(1+\varepsilon \chi _{3})\widetilde{\eta }_{3}\mathring{h}%
_{3},h_{4}=(1+\varepsilon \chi _{4})\widetilde{\eta }_{4}\mathring{h}_{4},\
^{\varepsilon }w_{i}=\mathring{w}_{i}+\widetilde{w}_{i}+\varepsilon
\overline{w}_{i},\ ^{\varepsilon }n_{i}=\mathring{n}_{i}+\tilde{n}%
_{i}+\varepsilon \overline{n}_{i}];  \notag \\
&&\mathbf{\Upsilon =}\ ^{\mu }\tilde{\Lambda}(1+\varepsilon \ \widetilde{%
\Lambda }/\ ^{\mu }\tilde{\Lambda});\ \ \ ^{\varepsilon }\tilde{\Phi}=\tilde{%
\Phi}(x^{k},\varphi )[1+\varepsilon \ \ ^{1}\tilde{\Phi}(x^{k},\varphi )/%
\tilde{\Phi}(x^{k},\varphi )]=\exp [\ \ ^{\varepsilon }\varpi (x^{k},\varphi
)],  \notag
\end{eqnarray}%
\begin{equation*}
ds_{4\varepsilon dK}^{2}=\epsilon _{i}(1+\varepsilon \chi _{i})e^{\psi
(x^{k})}(dx^{i})^{2}+\frac{\ ^{\varepsilon }\tilde{\Phi}^{2}}{4\ \mathbf{%
\Upsilon }}\left[ dy^{3}+(\partial _{i}\ n)dx^{i}\right] ^{2}+\ \frac{%
(\partial _{\varphi }\ \ ^{\varepsilon }\tilde{\Phi})^{2}}{\ \mathbf{%
\Upsilon }\ \ ^{\varepsilon }\tilde{\Phi}^{2}}\left[ dy^{4}+(\partial _{i}\
\ ^{\varepsilon }\check{A})dx^{i}\right] ^{2},
\end{equation*}%
which for LC--configurations, $\partial _{i}\ \ ^{\varepsilon }\check{A}%
=\partial _{i}\ \ ^{\varepsilon }\check{A}+\varepsilon \partial _{i}\ \ ^{1}%
\check{A}$ is determined by $\ ^{\varepsilon }\tilde{\Phi}=\tilde{\Phi}%
+\varepsilon \ ^{1}\tilde{\Phi}$ following conditions (\ref{data4c}). The
values labeled by "$\circ $" and "$\widetilde{}$" are taken as in previous
sections \ but, for simplicity, we omit priming of indices and consider $%
\varepsilon \overline{n}_{i}=0.$ The $\chi $- and $w$--values are computed
for $\varepsilon $--deformed LC--configurations, see formulae (\ref{zerot})
for spacetime components, as solutions of the system (\ref{sourc1}) in the
form (\ref{e1})--(\ref{e4}) for a source $\mathbf{\Upsilon =\ ^{\mu }\tilde{%
\Lambda}+\ }\varepsilon \widetilde{\Lambda }.$

The nonholonomic deformations (\ref{def2}) of the off--diagonal metrics (\ref%
{nvlcmgs}) give a new class of $\varepsilon $--deformed solutions with%
\begin{eqnarray}
\chi _{1} &=&\chi _{2}=\chi ,\mbox{ for }\partial _{11}^{2}\chi +\epsilon
_{2}\partial _{22}^{2}\chi =2\widetilde{\Lambda };  \label{edefcel} \\
\chi _{3} &=&2\ ^{1}\tilde{\Phi}/\tilde{\Phi}-\mathbf{\ }\widetilde{\Lambda }%
/\ ^{\mu }\tilde{\Lambda},\ \chi _{4}=2\partial _{4}\ ^{1}\tilde{\Phi}/%
\tilde{\Phi}-2\ ^{1}\tilde{\Phi}/\tilde{\Phi}-\widetilde{\Lambda }/\ ^{\mu }%
\tilde{\Lambda}, \ \overline{w}_{i} = (\frac{\partial _{i}\ ^{1}\tilde{\Phi}%
}{\partial _{i}\tilde{\Phi}}-\frac{\partial _{4}\ ^{1}\tilde{\Phi}}{\partial
_{4}\tilde{\Phi}})\frac{\partial _{i}\tilde{\Phi}}{\partial _{4}\tilde{\Phi}}%
=\partial _{i}\ \ ^{1}\check{A},\overline{n}_{i}=0.  \notag
\end{eqnarray}%
There is no summation on index "$i"$ in the last formula and $\mathring{h}%
_{3^{\prime }}\mathring{h}_{4^{\prime }}=\overline{A}\overline{C}-\overline{B%
}^{2}$. The deformations are determined respectively by two generating
functions $\tilde{\Phi}$ and $\ ^{1}\tilde{\Phi}$ and two sources $\ ^{\mu }%
\tilde{\Lambda}$ and $\widetilde{\Lambda }.$

Summarizing the results, we construct an off--diagonal generalization of the
Kerr metric by Ricci solitons, "main" mass gravity terms and additional $%
\varepsilon $--parametric $f$--modifications,
\begin{eqnarray}
ds^{2} &=&e^{\psi (x^{k^{\prime }})}(1+\varepsilon \chi (x^{k^{\prime
}}))[(dx^{1^{\prime }})^{2}+(dx^{2^{\prime }})^{2}]-  \notag \\
&&\frac{\tilde{\Phi}^{2}}{4|\ ^{\mu }\tilde{\Lambda}|}\overline{A}%
[1+\varepsilon (2\ ^{1}\tilde{\Phi}/\tilde{\Phi}-\mathbf{\ }\widetilde{%
\Lambda }/\ ^{\mu }\tilde{\Lambda})][dy^{3^{\prime }}+\left( \partial
_{k^{\prime }}\ ^{\eta }n(x^{i^{\prime }})-\partial _{k^{\prime }}(\widehat{y%
}^{3^{\prime }}+\varphi \overline{B}/\overline{A})\right) dx^{k^{\prime
}}]^{2}+  \label{nvlcmgse} \\
&&\frac{(\partial _{\varphi }\tilde{\Phi})^{2}}{\ \ ^{\mu }\tilde{\Lambda}%
\tilde{\Phi}^{2}}(\overline{C}-\overline{B}^{2}/\overline{A})[1+\varepsilon
(2\partial _{4}\ ^{1}\tilde{\Phi}/\tilde{\Phi}-2\ ^{1}\tilde{\Phi}/\tilde{%
\Phi}-\widetilde{\Lambda }/\ ^{\mu }\tilde{\Lambda})][d\varphi +(\partial
_{i^{\prime }}\ \widetilde{A}+\varepsilon \partial _{i^{\prime }}\ \ ^{1}%
\check{A})dx^{i^{\prime }}]^{2}.  \notag
\end{eqnarray}

We can consider $\varepsilon $--deformations of type (\ref{def2}) for (\ref%
{ofindtmg}) and generate new classes of off--diagonal solutions with
nonholonomically induced torsion determined both by Ricci soliton, massive
and $f$--modifications of GR. Such geometric and physical models are new and can not be
identified with   effective ones with anisotropic polarizations in GR which
also give different $r$--jet symmetries and prolongations.

\subsection{Nonholonomic $r$--jet off--diagonal Ricci soliton prolongations
of the Kerr solution}

In reference \cite{gheorghiu}, we studied generic off--diagonal deformations
of the Kerr metric into solutions on higher dimensional spacetimes. The goal
of this section is to show how prolongations on $r$--jet variables can
performed following similar methods but generalized to include nonholonomic
jet variables.

\subsubsection{Jet one shell deformations with nontrivial cosmological
constant}

Jet symmetries impose certain constraints on possible off--diagonal
deformations of a Kerr metric generalized for a corresponding class of
solutions with any nontrivial cosmological constant in 6--d. (In a similar
form we can generalize the constructions for any finite number of shells).
The corresponding class of Kerr -- de Sitter jet prolongation configurations
are generated by nonholonomic deformations $(\mathbf{\mathring{g}},\mathbf{%
\mathring{N},\ }^{v}\mathring{\Upsilon}=0,\mathring{\Upsilon}=0)\rightarrow (%
\widetilde{\mathbf{g}},\widetilde{\mathbf{N}}\mathbf{,\ }^{v}\widetilde{%
\Upsilon }=\Lambda ,\widetilde{\Upsilon }=\Lambda ,\mathbf{\ }^{v_{1}}%
\widetilde{\Upsilon }=\Lambda )$ when solutions are characterized by a jet
Killing symmetry $\eth /\partial \zeta ^{5}$ and parameterized as%
\begin{eqnarray}
ds^{2} &=&e^{\psi (x^{k^{\prime }})}[(dx^{1^{\prime }})^{2}+(dx^{2^{\prime
}})^{2}]-\frac{\tilde{\Phi}^{2}}{4\Lambda }\overline{A}[dy^{3^{\prime
}}+\left( \partial _{k^{\prime }}\ ^{\eta }n(x^{i^{\prime }})-\partial
_{k^{\prime }}(\widehat{y}^{3^{\prime }}+\varphi \overline{B}/\overline{A}%
)\right) dx^{k^{\prime }}]^{2}+  \label{6dks} \\
&&\frac{(\partial _{\varphi }\tilde{\Phi})^{2}}{\ \Lambda \overline{A}\tilde{%
\Phi}^{2}}(\overline{A}\overline{C}-\overline{B}^{2})[d\varphi +(\partial
_{i^{\prime }}\ ^{\eta }\widetilde{A})dx^{i^{\prime }}]^{2}+\frac{\ ^{1}%
\tilde{\Phi}^{2}}{4\ \Lambda }\left[ d\zeta ^{5}+(\partial _{\tau }\
^{1}n)du^{\tau }\right] ^{2}+\ \frac{(\eth _{6}\ ^{1}\tilde{\Phi})^{2}}{\
\Lambda \ ^{1}\tilde{\Phi}^{2}}\left[ d\zeta ^{6}+(\partial _{\tau }\ ^{1}%
\check{A})du^{\tau }\right] ^{2}.  \notag
\end{eqnarray}%
The generating functions for such d--metrics are parameterized as {\small
\begin{equation*}
\tilde{\Phi} =\tilde{\Phi}(x^{k^{\prime }},\varphi ),\ ^{1}\tilde{\Phi}%
(u^{\beta },\zeta ^{6})=\ ^{1}\tilde{\Phi}(x^{k^{\prime }},t,\varphi ,\zeta
^{6});\ ^{\eta }n=\ ^{\eta }n(x^{i^{\prime }}), \ ^{1}n = \ ^{1}n(u^{\beta
},\zeta ^{6});\ ^{\eta }\widetilde{A}=\ ^{\eta }\widetilde{A}(x^{k^{\prime
}},\varphi ),\ ^{1}\check{A}=\ ^{1}\check{A}(u^{\beta },\zeta ^{6}),
\end{equation*}%
} and subject to LC--conditions and conditions of integrability and the
"primary" data $\overline{A},\overline{B},\overline{C}$ are taken for the
Kerr  solution in the form (\ref{kerrcoef}).

By imposing additional symmetries and constraints on the spacetime generating
functions, we can "extract" ellipsoid configurations for a subclass of
metrics with $\varepsilon $--deformations,
\begin{eqnarray*}
ds^{2} &=&e^{\psi (x^{k^{\prime }})}[(dx^{1^{\prime }})^{2}+(dx^{2^{\prime
}})^{2}]-\frac{\tilde{\Phi}^{2}}{4\Lambda }\overline{A}[1+2\varepsilon
\underline{\zeta }\sin (\omega _{0}\varphi +\varphi _{0})][dy^{3^{\prime
}}+\left( \partial _{k^{\prime }}\ ^{\eta }n(x^{i^{\prime }})-\partial
_{k^{\prime }}(\widehat{y}^{3^{\prime }}+\varphi \frac{\overline{B}}{%
\overline{A}})\right) dx^{k^{\prime }}]^{2} \\
&&+\frac{(\partial _{\varphi }\tilde{\Phi})^{2}}{\Lambda \tilde{\Phi}^{2}}(%
\overline{C}-\overline{B}^{2}/\overline{A})[1+\varepsilon (2\frac{\partial
_{\varphi }\tilde{\Phi}}{\tilde{\Phi}}\underline{\zeta }\sin (\omega
_{0}\varphi +\varphi _{0})+2\omega _{0}\underline{\zeta }\cos (\omega
_{0}\varphi +\varphi _{0}))][d\varphi +(\partial _{i^{\prime }}\ ^{\eta }%
\widetilde{A})dx^{i^{\prime }}]^{2} \\
&&+\frac{\ ^{1}\tilde{\Phi}^{2}}{4\ \Lambda }\left[ d\zeta ^{5}+(\partial
_{\tau }\ ^{1}n)du^{\tau }\right] ^{2}+\ \frac{(\eth _{6}\ ^{1}\tilde{\Phi}%
)^{2}}{\ \Lambda \ ^{1}\tilde{\Phi}^{2}}\left[ d\zeta ^{6}+(\partial _{\tau
}\ ^{1}\check{A})du^{\tau }\right] ^{2},
\end{eqnarray*}%
where $\underline{\zeta },\omega _{0}$ and $\varphi _{0}$ are certain
constants determining gravitational rotoid configurations with eccentricity $%
\varepsilon .$ For small values of $\varepsilon ,$ such metrics describe
"slightly" deformed Kerr black holes embedded self--consistently into a
generic off--diagonal jet prolongation as a 6--d spacetime.

\subsubsection{Two shell effective 8--d jet prolongations}

Applying the AFDM, we can construct two shell nonholonomic jet prolongations
of the Kerr metric which, in general, are with nontrivial induced torsion
for an effective 8-d spacetime with interior jet symmetries. The
nonholonomic deformations are defined by the data $(\mathbf{\mathring{g}},%
\mathbf{\mathring{N},\ }^{v}\mathring{\Upsilon}=0,\mathring{\Upsilon}%
=0)\rightarrow (\widetilde{\mathbf{g}},\widetilde{\mathbf{N}}\mathbf{,\ }^{v}%
\widetilde{\Upsilon }=\Lambda ,\widetilde{\Upsilon }=\Lambda ,\mathbf{\ }%
^{v_{1}}\widetilde{\Upsilon }=\Lambda ,\mathbf{\ }^{v_{2}}\widetilde{%
\Upsilon }=\Lambda )$ and extending on jet variables the 4--d quadratic element (%
\ref{ofindtmg}) but for a different source (we consider a cosmological
constant $\Lambda $ for all dimensions). The corresponding class of
solutions is determined by {\small
\begin{eqnarray}
ds^{2} &=&e^{\psi (x^{k^{\prime }})}[(dx^{1^{\prime }})^{2}+(dx^{2^{\prime
}})^{2}]-\frac{\Phi ^{2}}{4\Lambda }\overline{A}[dy^{3^{\prime }}+\left( \
_{1}n_{k^{\prime }}(x^{i^{\prime }})+\ _{2}n_{k^{\prime }}(x^{i^{\prime }})%
\frac{(\partial _{\varphi }\Phi )^{2}}{\Phi ^{5}}-\partial _{k^{\prime }}(%
\widehat{y}^{3^{\prime }}+\varphi \frac{\overline{B}}{\overline{A}})\right)
dx^{k^{\prime }}]^{2}  \notag \\
&&+\frac{(\partial _{\varphi }\Phi )^{2}}{\ \Lambda \Phi ^{2}\overline{A}}(%
\overline{A}\overline{C}-\overline{B}^{2})[d\varphi +\frac{\partial
_{i^{\prime }}\Phi }{\partial _{\varphi }\Phi }dx^{i^{\prime }}]^{2}+\frac{\
^{1}\tilde{\Phi}^{2}}{4\ \Lambda }\left[ d\zeta ^{5}+(\partial _{\tau }\
^{1}n)du^{\tau }\right] ^{2}+\ \frac{(\eth _{6}\ ^{1}\tilde{\Phi})^{2}}{\
\Lambda \ ^{1}\tilde{\Phi}^{2}}\left[ d\zeta ^{6}+(\partial _{\tau }\ ^{1}%
\check{A})du^{\tau }\right] ^{2}  \notag \\
&&+\frac{\ ^{2}\tilde{\Phi}^{2}}{4\ \Lambda }\left[ d\zeta ^{7}+(\partial
_{\tau _{1}}\ ^{2}n)du^{\tau _{1}}\right] ^{2}+\ \frac{(\eth _{8}\ ^{2}%
\tilde{\Phi})^{2}}{\ \Lambda \ ^{2}\tilde{\Phi}^{2}}\left[ d\zeta
^{8}+(\partial _{\tau _{1}}\ ^{2}\check{A})du^{\tau _{1}}\right] ^{2}.
\label{8dfd}
\end{eqnarray}%
} The generating functions depend on spacetime and jet variables,
\begin{eqnarray}
\Phi &=&\Phi (x^{k^{\prime }},\varphi ),\ ^{1}\tilde{\Phi}(u^{\beta },\zeta
^{6})=\ ^{1}\tilde{\Phi}(x^{k^{\prime }},t,\varphi ,\zeta ^{6}),\ \ ^{2}%
\tilde{\Phi}(u^{\beta _{1}},\zeta ^{8})=\ ^{2}\tilde{\Phi}(x^{k^{\prime
}},t,\varphi ,\zeta ^{5},\zeta ^{6},\zeta ^{8});  \label{genf8fd} \\
\ ^{1}n &=&\ ^{1}n(u^{\beta },\zeta ^{6}),\ ^{2}n=\ ^{2}n(u^{\beta
_{1}},\zeta ^{8}),\ ^{\eta }\widetilde{A}=\ ^{\eta }\widetilde{A}%
(x^{k^{\prime }},\varphi ),\ ^{1}\check{A}=\ ^{1}\check{A}(x^{k^{\prime
}},t,\varphi ,\zeta ^{6}),\ ^{2}\check{A}=\ ^{2}\check{A}(x^{k^{\prime
}},t,\varphi ,\zeta ^{5},\zeta ^{6},\zeta ^{8}).  \notag
\end{eqnarray}%
Such values are chosen in such forms when the nonholonomically induced
torsion (\ref{dtors}) is effectively modeled on a 4--d pseudo--Riemannian
spacetime but on jet shells $s=1$ and $s=2$ the torsion fields are zero.
This mean that there are jet coordinate transforms to certain classes of
holonomic variables. We can generate jet depending nontrivial torsion
N--adapted coefficients if nontrivial integration functions of type $\
_{2}n_{k^{\prime }}(x^{i^{\prime }})$ are extended to contain jet variables.

\subsubsection{Kerr Ricci soliton deformations and vacuum $r$--jet
prolongations}

Classes of solutions exist with jet variables describing vacuum
ellipsoid spacetime configurations with prolongations on two shell jet
variables when the source is of type $\mathbf{\Upsilon =\widetilde{\lambda }%
+\ }\varepsilon (\widetilde{\Lambda }+\Lambda )=0,$with effective massive
gravity term $\ ^{\mu }\tilde{\Lambda}\mathbf{=}\mu _{g}^{2}|\ \lambda |,$
and give ellipsoidal off--diagonal configurations in GR. For such metrics, $%
\varepsilon =-\ ^{\mu }\tilde{\Lambda}/(\widetilde{\Lambda }+\Lambda )\ll 1$
can be considered as an eccentricity parameter. The corresponding models of
off--diagonal jet interior gravitational interactions are with $f$%
--modifications when $\widetilde{\Lambda }$ compensates nonholonomic
contributions via effective constant $\widetilde{\Lambda }$ and relates the
constructions to massive gravity deformations of a Kerr solution. This
subclass of solutions for $\varepsilon $--deformations into vacuum solutions
is parameterized by target ansatz {\small
\begin{eqnarray}
ds^{2} &=&e^{\psi (x^{k^{\prime }})}(1+\varepsilon \chi (x^{k^{\prime
}}))[(dx^{1^{\prime }})^{2}+(dx^{2^{\prime }})^{2}]-\frac{\tilde{\Phi}^{2}}{%
4\ ^{\mu }\tilde{\Lambda}}\overline{A}[1+\varepsilon \chi _{3^{\prime
}}][dy^{3^{\prime }}+\left( \partial _{k^{\prime }}\ ^{\eta }n(x^{i^{\prime
}})-\partial _{k^{\prime }}(\widehat{y}^{3^{\prime }}+\varphi \overline{B}/%
\overline{A})\right) dx^{k^{\prime }}]^{2}+  \notag \\
&&\frac{(\partial _{\varphi }\tilde{\Phi})^{2}\eta _{4^{\prime }}}{\ ^{\mu }%
\tilde{\Lambda}\tilde{\Phi}^{2}}(\overline{C}-\frac{\overline{B}^{2}}{%
\overline{A}})[1+\varepsilon \chi _{4^{\prime }}][d\varphi +(\partial
_{i^{\prime }}\ \widetilde{A}+\varepsilon \partial _{i^{\prime }}\ \ ^{1}%
\check{A})dx^{i^{\prime }}]^{2}+\frac{\ ^{1}\tilde{\Phi}^{2}}{4(\ \widetilde{%
\Lambda }+\Lambda )}\left[ d\zeta ^{5}+(\partial _{\tau }\ ^{1}n)du^{\tau }%
\right] ^{2}+  \label{kmasedvac} \\
&&\frac{(\eth _{6}\ ^{1}\tilde{\Phi})^{2}}{(\ \widetilde{\Lambda }+\Lambda
)\ ^{1}\tilde{\Phi}^{2}}\left[ d\zeta ^{6}+(\partial _{\tau }\ ^{1}\check{A}%
)du^{\tau }\right] ^{2}+\frac{\ ^{2}\tilde{\Phi}^{2}}{4\ (\widetilde{\Lambda
}+\Lambda )}\left[ d\zeta ^{7}+(\eth _{\tau _{1}}\ ^{2}n)du^{\tau _{1}}%
\right] ^{2}+\ \frac{(\eth _{8}\ ^{2}\tilde{\Phi})^{2}}{\ (\ \widetilde{%
\Lambda }+\Lambda )\ ^{2}\tilde{\Phi}^{2}}\left[ d\zeta ^{8}+(\eth _{\tau
_{1}}\ ^{2}\check{A})du^{\tau _{1}}\right] ^{2}.  \notag
\end{eqnarray}%
} The jet components are generated by functions $\ ^{1}\tilde{\Phi},$ $^{2}%
\tilde{\Phi}$ and N--coefficients similar to solutions (\ref{8dfd}) but
with modified effective jet prolongation sources, $\Lambda \rightarrow \
\widetilde{\Lambda }+\Lambda .$ This result shows that interior jet
interactions can mimic $\varepsilon $--deformations in order to compensate
contributions from $f$--modifications and even the effective vacuum
configurations for the 4--d horizontal part. In general, vacuum metrics (\ref%
{kmasedvac}) encode jet modifications/ polarizations of physical constants
and coefficients of metrics under nonlinear polarizations of an effective
8-d vacuum distinguishing 4--d nonholonomic configurations and Ricci soliton
or massive gravity contributions. Jet variables and $f$--modified
contributions are described by terms proportional to eccentricity parameter $%
\varepsilon .$

\subsubsection{Jet ellipsoid like Kerr -- de Sitter configurations}

Using the solutions (\ref{8dfd}), we can construct a class of non--vacuum 8--d
jet prolonged solutions with rotoid configurations. For this, we choose for $%
\varepsilon $--deformations (see a similar formula (\ref{edefcel}) for 4-d)
a small polarization $\chi _{3}=2\ ^{1}\tilde{\Phi}/\tilde{\Phi}-\mathbf{\ }%
(\ \widetilde{\Lambda }+\Lambda )/\ ^{\mu }\tilde{\Lambda}=2\underline{\zeta
}\sin (\omega _{0}\varphi +\varphi _{0}).$ Re--expressing $\ ^{1}\tilde{\Phi}%
=\tilde{\Phi}[\mathbf{\ }(\ \widetilde{\Lambda }+\Lambda )/2\ ^{\mu }\tilde{%
\Lambda}+\underline{\zeta }\sin (\omega _{0}\varphi +\varphi _{0})]$ and (%
\ref{genf8fd}), one can generate a class of off--diagonal jet prolongations
of ellipsoid Kerr -- de Sitter configurations \\
{\small
\begin{eqnarray*}
&&ds^{2}=\\
&&e^{\psi (x^{k^{\prime }})}(1+\varepsilon \chi (x^{k^{\prime
}}))[(dx^{1^{\prime }})^{2}+(dx^{2^{\prime }})^{2}]- \frac{\tilde{\Phi}^{2}%
\overline{A}}{4|\ ^{\mu }\tilde{\Lambda}|}[1+2\varepsilon \underline{\zeta }%
\sin (\omega _{0}\varphi +\varphi _{0})][dy^{3^{\prime }}+(\partial
_{k^{\prime }}\ ^{\eta }n(x^{i^{\prime }})-\partial _{k^{\prime }}(\widehat{y%
}^{3^{\prime }}+\varphi \frac{\overline{B}}{\overline{A}})) dx^{k^{\prime
}}]^{2} \\
&&+\frac{(\partial _{\varphi }\tilde{\Phi})^{2}}{\ \ ^{\mu }\tilde{\Lambda}%
\tilde{\Phi}^{2}}(\overline{C}-\frac{\overline{B}^{2}}{\overline{A}}%
)[1+\varepsilon (\frac{\partial _{\varphi }\tilde{\Phi}}{\tilde{\Phi}}\frac{%
\ \widetilde{\Lambda }+\Lambda }{\ ^{\mu }\tilde{\Lambda}}+2\frac{\partial
_{\varphi }\tilde{\Phi}}{\tilde{\Phi}}\underline{\zeta }\sin (\omega
_{0}\varphi +\varphi _{0})+2\omega _{0}\mathbf{\ }\underline{\zeta }\cos
(\omega _{0}\varphi +\varphi _{0}))][d\varphi +(\partial _{i^{\prime }}\
\widetilde{A}+\varepsilon \partial _{i^{\prime }}\ \ ^{1}\check{A}%
)dx^{i^{\prime }}]^{2} \\
&&+\frac{\ ^{1}\tilde{\Phi}^{2}}{4\ (\ \widetilde{\Lambda }+\Lambda )}\left[
d\zeta ^{5}+(\partial _{\tau }\ ^{1}n)du^{\tau }\right] ^{2}+\ \frac{(\eth
_{6}\ ^{1}\tilde{\Phi})^{2}}{\ (\ \widetilde{\Lambda }+\Lambda )\ ^{1}\tilde{%
\Phi}^{2}}\left[ d\zeta ^{6}+(\partial _{\tau }\ ^{1}\check{A})du^{\tau }%
\right] ^{2} \\
&&+\frac{\ ^{2}\tilde{\Phi}^{2}}{4\ (\widetilde{\Lambda }+\Lambda )}\left[
d\zeta ^{7}+(\eth _{\tau _{1}}\ ^{2}n)du^{\tau _{1}}\right] ^{2}+\ \frac{%
(\eth _{8}\ ^{2}\tilde{\Phi})^{2}}{\ (\ \widetilde{\Lambda }+\Lambda )\ ^{2}%
\tilde{\Phi}^{2}}\left[ d\zeta ^{8}+(\eth _{\tau _{1}}\ ^{2}\check{A}%
)du^{\tau _{1}}\right] ^{2}.
\end{eqnarray*}%
}These metrics possess the Killing symmetry $\eth _{7}$ and define $%
\varepsilon $--deformations of Kerr -- de Sitter black holes into ellipsoid
configurations with effective cosmological constants determined,
respectively, by constants in Ricci soliton models, massive gravity, $f$%
--modifications and jet prolongation contributions.

\appendix

\setcounter{equation}{0} \renewcommand{\theequation}
{A.\arabic{equation}} \setcounter{subsection}{0}
\renewcommand{\thesubsection}
{A.\arabic{subsection}}

\section{ N--adapted Coefficients and Proofs}

We provide a set of necessary N-adapted coefficient formulae that are
important for proofs and applications. A series of results obtained in \cite%
{vexsol1,vex3,vexsol2} are reformulated and generalized for nonholonomic $r$%
--jet variables with conventional $2+2+....$ splitting.

\subsection{Torsions and Curvature of d--connections on $\mathbf{J}^{r}(%
\mathbf{V},\mathbf{V}^{\prime })$ with 2-d shells}

\label{sscoefcurv}For any d--connection strucutre $\ ^{s}\mathbf{D}$ and $r$%
--jet 2d shell prolongations with coefficients (\ref{coefd}), there are two
important theorems:

\begin{theorem}
The N--adapted coefficients of d--torsion $\ ^{s}\mathbf{T}=\{\mathbf{T}_{\
\beta _{s}\gamma _{s}}^{\alpha _{s}}\}$ from (\ref{dt}) are computed
recurrently "shall by shell" following formulae {\small
\begin{eqnarray}
T_{\ jk}^{i} &=&L_{jk}^{i}-L_{kj}^{i},T_{\ ja}^{i}=C_{jb}^{i},T_{\
ji}^{a}=-\ ^{N}J_{\ ji}^{a}, T_{aj}^{c} = L_{aj}^{c}-\partial
_{a}(N_{j}^{c}),T_{\ bc}^{a}=C_{bc}^{a}-C_{cb}^{a},%
\mbox{ spacetime
components }; ...  \label{dtorsj} \\
T_{\ j_{s}k_{s}}^{i_{s}}
&=&L_{j_{s}k_{s}}^{i_{s}}-L_{k_{s}j_{s}}^{i_{s}},T_{\
j_{s}a_{s}}^{i_{s}}=C_{j_{s}b_{s}}^{i_{s}},T_{\ j_{s}i_{s}}^{a_{s}}=-\
^{N}J_{\ j_{s}i_{s}}^{a_{s}}, T_{a_{s}j_{s}}^{c_{s}} =
L_{a_{s}j_{s}}^{c_{s}}-\eth _{a_{s}}(N_{j_{s}}^{c_{s}}),T_{\
b_{s}c_{s}}^{a_{s}}=C_{b_{s}c_{s}}^{a_{s}}-C_{cb}^{a_{s}},\mbox{ r--jet }.
\notag
\end{eqnarray}
}
\end{theorem}

\begin{proof}
The coefficients (\ref{dtors}) are computed for any $\widehat{\mathbf{D}}=\{%
\mathbf{\Gamma }_{\ \beta _{s}\gamma _{s}}^{\alpha _{s}}\}$ and N--adapted
frames (\ref{naders}) and (\ref{nadifs}) using standard differential form
calculus with (\ref{dt}) (or, in operator form, using the formula (\ref%
{torsshell})).

$\square $
\end{proof}

\vskip5pt

\begin{theorem}
The N--adapted coefficients of d--curvature $\ ^{s}\mathbf{R}=\{\mathbf{%
\mathbf{R}}_{\ \ \beta _{s}\gamma _{s}\delta _{s}}^{\alpha _{s}}\}$ from (%
\ref{dc}) are computed recurrently "shell by shell" following formulae
\begin{eqnarray}
R_{\ hjk}^{i} &=&\mathbf{\partial }_{k}L_{\ hj}^{i}-\partial _{j}L_{\
hk}^{i}+L_{\ hj}^{m}L_{\ mk}^{i}-L_{\ hk}^{m}L_{\ mj}^{i}-C_{\ ha}^{i}\
^{N}J_{\ kj}^{a},  \notag \\
R_{\ bjk}^{a} &=&\mathbf{\partial }_{k}L_{\ bj}^{a}-\partial _{j}L_{\
bk}^{a}+L_{\ bj}^{c}L_{\ ck}^{a}-L_{\ bk}^{c}L_{\ cj}^{a}-C_{\ bc}^{a}\
^{N}J_{\ kj}^{c},  \label{dcurv} \\
R_{\ jka}^{i} &=&\mathbf{\partial }_{a}L_{\ jk}^{i}-D_{k}C_{\ ja}^{i}+C_{\
jb}^{i}\widehat{T}_{\ ka}^{b}, \ R_{\ bka}^{c} = e_{a}L_{\ bk}^{c}-D_{k}C_{\
ba}^{c}+C_{\ bd}^{c}T_{\ ka}^{c},  \notag \\
R_{\ jbc}^{i} &=&\mathbf{\partial }_{c}C_{\ jb}^{i}-\partial _{b}C_{\
jc}^{i}+C_{\ jb}^{h}C_{\ hc}^{i}-C_{\ jc}^{h}C_{\ hb}^{i},  \notag \\
R_{\ bcd}^{a} &=&\partial _{d}C_{\ bc}^{a}-\partial _{c}C_{\ bd}^{a}+C_{\
bc}^{e}C_{\ ed}^{a}-C_{\ bd}^{e}C_{\ ec}^{a},\mbox{ spacetime components };
\notag
\end{eqnarray}
\begin{eqnarray*}
R_{\ h_{s}j_{s}k_{s}}^{i_{s}} &=&\eth _{k_{s}}L_{\ h_{s}j_{s}}^{i_{s}}-\eth
_{j_{s}}L_{\ h_{s}k_{s}}^{i_{s}}+L_{\ h_{s}j_{s}}^{m_{s}}L_{\
m_{s}k_{s}}^{i_{s}}-L_{\ h_{s}k_{s}}^{m_{s}}L_{\ m_{s}j_{s}}^{i_{s}}-C_{\
h_{s}a_{s}}^{i_{s}}\ ^{N}J_{\ k_{s}j_{s}}^{a_{s}}, \\
R_{\ b_{s}j_{s}k_{s}}^{a_{s}} &=&\eth _{k}L_{\ bj}^{a}-\eth _{j}L_{\
bk}^{a}+L_{\ bj}^{c}L_{\ ck}^{a}-L_{\ bk}^{c}L_{\ cj}^{a}-C_{\ bc}^{a}\
^{N}J_{\ kj}^{c}, \\
R_{\ j_{s}k_{s}a_{s}}^{i_{s}} &=&\eth _{a_{s}}L_{\
j_{s}k_{s}}^{i_{s}}-D_{k_{s}}C_{\ j_{s}a_{s}}^{i_{s}}+C_{\
j_{s}b_{s}}^{i_{s}}\widehat{T}_{\ k_{s}a_{s}}^{b_{s}}, \ R_{\
b_{s}k_{s}a_{s}}^{c_{s}} = \eth _{a_{s}}L_{\
b_{s}k_{s}}^{c_{s}}-D_{k_{s}}C_{\ b_{s}a_{s}}^{c_{s}}+C_{\
b_{s}d_{s}}^{c_{s}}T_{\ k_{s}a_{s}}^{c_{s}}, \\
R_{\ j_{s}b_{s}c_{s}}^{i_{s}} &=&\eth _{c_{s}}C_{\ j_{s}b_{s}}^{i_{s}}-\eth
_{b_{s}}C_{\ j_{s}c_{s}}^{i_{s}}+C_{\ j_{s}b_{s}}^{h_{s}}C_{\
h_{s}c_{s}}^{i_{s}}-C_{\ j_{s}c_{s}}^{h_{s}}C_{\ h_{s}b_{s}}^{i_{s}}, \\
R_{\ b_{s}c_{s}d_{s}}^{a_{s}} &=&\eth _{d_{s}}C_{\ b_{s}c_{s}}^{a_{s}}-\eth
_{c_{s}}C_{\ b_{s}d_{s}}^{a_{s}}+C_{\ b_{s}c_{s}}^{e_{s}}C_{\
e_{s}d_{s}}^{a_{s}}-C_{\ b_{s}d_{s}}^{e_{s}}C_{\ e_{s}c_{s}}^{a_{s}},%
\mbox{
spacetime components }.
\end{eqnarray*}
\end{theorem}

\begin{proof}
The coefficients (\ref{dcurv}) are computed for any $\widehat{\mathbf{D}}=\{%
\mathbf{\Gamma }_{\ \beta _{s}\gamma _{s}}^{\alpha _{s}}\}$ and N--adapted
frames (\ref{naders}) and (\ref{nadifs}) using a standard differential form
calculus with (\ref{dc}) (or, in operator form, using the formula (\ref%
{curvshell})). For $2+2+...$ shell decompositions, such formulae are similar
to the coefficients of curvature in higher dimensional spacetime but with
re--parameterized (in our case) nonholonomic $r$--jet variables.

$\square $
\end{proof}

\vskip5pt

\subsection{Sketch of proof of theorem \protect\ref{tcandist}}

\label{prooftcands} We can check by straightforward computations that the
conditions of metric compatibility and zero $h$- and $v$-- torsions are
satisfied by $\ ^{s}\widehat{\mathbf{D}}=\{\widehat{\mathbf{\Gamma }}_{\
\alpha _{s}\beta _{s}}^{\gamma _{s}}\}$ with coefficients computed
recurrently
\begin{eqnarray}
\widehat{L}_{jk}^{i} &=&\frac{1}{2}g^{ir}\left( \mathbf{e}_{k}g_{jr}+\mathbf{%
e}_{j}g_{kr}-\mathbf{e}_{r}g_{jk}\right) ,  \notag \\
\widehat{L}_{bk}^{a} &=&e_{b}(N_{k}^{a})+\frac{1}{2}h^{ac}\left( \mathbf{e}%
_{k}h_{bc}-h_{dc}\ e_{b}N_{k}^{d}-h_{db}\ e_{c}N_{k}^{d}\right) ,  \notag \\
\widehat{C}_{jc}^{i} &=&\frac{1}{2}g^{ik}e_{c}g_{jk},\ \widehat{C}_{bc}^{a}=%
\frac{1}{2}h^{ad}\left( e_{c}h_{bd}+e_{c}h_{cd}-e_{d}h_{bc}\right) ,
\label{candcon} \\
&&  \notag \\
\widehat{L}_{\beta \gamma }^{\alpha } &=&\frac{1}{2}g^{\alpha \tau }\left(
\mathbf{e}_{\gamma }g_{\beta \tau }+\mathbf{e}_{\beta }g_{\gamma \tau }-%
\mathbf{e}_{\tau }g_{\beta \gamma }\right) ,  \notag \\
\widehat{L}_{b_{1}\gamma }^{a_{1}} &=&\eth _{b_{1}}(N_{\gamma }^{a_{1}})+%
\frac{1}{2}h^{a_{1}c_{1}}\left( \mathbf{e}_{\gamma
}h_{b_{1}c_{1}}-h_{d_{1}c_{1}}\ \eth _{b_{1}}N_{\gamma
}^{d_{1}}-h_{d_{1}b_{1}}\ \eth _{c_{1}}N_{\gamma }^{d_{1}}\right) ,  \notag
\\
\widehat{C}_{\beta c_{1}}^{\alpha } &=&\frac{1}{2}g^{\alpha \tau }\eth
_{c_{1}}g_{\beta \tau },\ \widehat{C}_{b_{1}c_{1}}^{a_{1}}=\frac{1}{2}%
h^{a_{1}d_{1}}\left( \eth _{c_{1}}h_{b_{1}d_{1}}+\eth
_{c_{1}}h_{c_{1}d_{1}}-\eth _{d_{1}}h_{b_{1}c_{1}}\right) ,  \notag \\
&&...  \notag \\
&&  \notag \\
\widehat{L}_{\beta _{s-1}\gamma _{s-1}}^{\alpha _{s-1}} &=&\frac{1}{2}%
g^{\alpha _{s-1}\tau _{s-1}}\left( \mathbf{e}_{\gamma _{s-1}}g_{\beta
_{s-1}\tau _{s-1}}+\mathbf{e}_{\beta _{s-1}}g_{\gamma _{s-1}\tau _{s-1}}-%
\mathbf{e}_{\tau _{s-1}}g_{\beta _{s-1}\gamma _{s-1}}\right) ,  \notag \\
\widehat{L}_{b_{s}\gamma _{s-1}}^{a_{s}} &=&\eth _{b_{s}}(N_{\gamma
_{s-1}}^{a_{s}})+\frac{1}{2}h^{a_{s}c_{s}}\left( \mathbf{e}_{\gamma
_{s-1}}h_{b_{s}c_{s}}-h_{d_{s}c_{s}}\ \eth _{b_{s}}N_{\gamma
_{s-1}}^{d_{s}}-h_{d_{s}b_{s}}\ \eth _{c_{s}}N_{\gamma
_{s-1}}^{d_{s}}\right) ,  \notag \\
\widehat{C}_{\beta _{s-1}c_{s}}^{\alpha _{s-1}} &=&\frac{1}{2}g^{\alpha
_{s-1}\tau _{s-1}}\eth _{c_{s}}g_{\beta _{s-1}\tau _{s-1}},\ \widehat{C}%
_{b_{s}c_{s}}^{a_{s}}=\frac{1}{2}h^{a_{s}d_{s}}\left( \eth
_{c_{s}}h_{b_{s}d_{s}}+\eth _{c_{s}}h_{c_{s}d_{s}}-\eth
_{d_{s}}h_{b_{s}c_{s}}\right) .  \notag
\end{eqnarray}%
The torsion d--tensor (\ref{dt}) of $\ ^{s}\widehat{\mathbf{D}}$ is
completely defined by $\ ^{s}\mathbf{g}$ (\ref{dm}) for any chosen $\ ^{s}%
\mathbf{N=\{}N_{i_{s}}^{a_{s}}\}$ if the above coefficients (\ref{candcon})
are introduced "shell by shell" into formulae
\begin{eqnarray}
\widehat{T}_{\ jk}^{i} &=&\widehat{L}_{jk}^{i}-\widehat{L}_{kj}^{i},\widehat{%
T}_{\ ja}^{i}=\widehat{C}_{jb}^{i},\widehat{T}_{\ ji}^{a}=-\ ^{N}J_{\
ji}^{a},\ \widehat{T}_{aj}^{c}=\widehat{L}_{aj}^{c}-e_{a}(N_{j}^{c}),%
\widehat{T}_{\ bc}^{a}=\ \widehat{C}_{bc}^{a}-\ \widehat{C}_{cb}^{a},  \notag
\\
&&....  \label{dtors} \\
\widehat{T}_{\ \beta _{s}\gamma _{s}}^{\alpha _{s}} &=&\widehat{L}_{\ \beta
_{s}\gamma _{s}}^{\alpha _{s}}-\widehat{L}_{\ \gamma _{s}\beta _{s}}^{\alpha
_{s}},\widehat{T}_{\ \beta _{s}b_{s}}^{\alpha _{s}}=\widehat{C}_{\ \beta
_{s}b_{s}}^{\alpha _{s}},\widehat{T}_{\ \beta _{s}\gamma _{s}}^{a_{s}}=\
^{N}J_{\ \gamma _{s}\beta _{s}}^{a_{s}}.  \notag
\end{eqnarray}%
We can impose as additional nonholonomic conditions certain equations when
all coefficients (\ref{dtors}) are zero. In this case we extract from (\ref%
{candcon}) various LC--configurations. The coefficients of the
LC--connection $\ _{\shortmid }\Gamma _{\ \alpha _{s}\beta _{s}}^{\gamma
_{s}}$ can be computed in standard form in coordinate bases and/or with
respect to N--adapted frames. Taking differences between $\widehat{\mathbf{%
\Gamma }}_{\ \alpha _{s}\beta _{s}}^{\gamma _{s}}$ and $\ _{\shortmid
}\Gamma _{\ \alpha _{s}\beta _{s}}^{\gamma _{s}},$ we find the N--adapted
coefficients of the distortion d--tensor $\widehat{\mathbf{Z}}_{\ \beta
_{s}\gamma _{s}}^{\alpha _{s}}$ (similar formulae are given for Corollary
2.1 and (22) in Ref. \cite{vex3}, in extra dimension coordinates but without
jet configurations).

\subsection{Proof of theorem \protect\ref{tdecoupling}}

\label{assdecoup}Such a proof is possible by explicitly the computing of the
N--adapted coefficients of the canonical Ricci d--tensor on $\mathbf{J}^{r}(%
\mathbf{V},\mathbf{V}^{\prime })$ with 2-d shells. Let us consider an ansatz
(\ref{ansk}) with $\partial _{4}h_{a}\neq 0,\eth _{6}h_{a_{1}}\neq
0,...,\eth _{2s}h_{a_{s}}\neq 0,$ when the partial derivatives are denoted
in the forms $\partial _{1}h=\partial h/\partial x^{1},$ $\partial
_{4}h=\partial h/\partial y^{4},\partial _{44}^{2}h=\partial ^{2}h/\partial
y^{4}\partial y^{4}$ and $\eth _{66}^{2}=\partial ^{2}h/\partial \zeta
^{6}\partial \zeta ^{6}$, where the indices taking values $5,6,...$ are for $%
2+2+...$ jet parameterized variables. We can construct more special classes
of solutions when the conditions alluded to  are not satisfied which warrants the
analysis of more special classes of solutions. For simplicity, we suppose
that via frame transformations it is always possible to introduce the necessary type of
parameterizations for d--metrics whenever the necessary types of partial derivatives of
some coefficients are not zero.

\begin{lemma}
With respect to N--adapted frames (\ref{naders}) and (\ref{nadifs}), the
nonzero coefficients of the Ricci d--tensor $\mathbf{\hat{R}}_{\alpha
_{s}\beta _{s}}$ (\ref{dricci}) for ansatz (\ref{ansk}) with Killing
symmetry on $\partial _{3}$ possess symmetries determined by the following formulae: for
spacetime components with partial derivative operator $\partial ,$
\begin{eqnarray}
\widehat{R}_{1}^{1} &=&\widehat{R}_{2}^{2}=-\frac{1}{2g_{1}g_{2}}[\partial
_{11}^{2}g_{2}-\frac{(\partial _{1}g_{1})(\partial _{1}g_{2})}{2g_{1}}-\frac{%
\left( \partial _{1}g_{2}\right) ^{2}}{2g_{2}}+\partial _{22}^{2}g_{1}-\frac{%
(\partial _{2}g_{1})(\partial _{2}g_{2})}{2g_{2}}-\frac{\left( \partial
_{2}g_{1}\right) ^{2}}{2g_{1}}],  \label{equ1} \\
\widehat{R}_{3}^{3} &=&\widehat{R}_{4}^{4}=-\frac{1}{2h_{3}h_{4}}[\partial
_{44}^{2}h_{3}-\frac{\left( \partial _{4}h_{3}\right) ^{2}}{2h_{3}}-\frac{%
(\partial _{4}h_{3})(\partial _{4}h_{4})}{2h_{4}}],  \label{equ2} \\
\widehat{R}_{3k} &=&\frac{h_{3}}{2h_{4}}\partial _{44}^{2}n_{k}+\left( \frac{%
h_{3}}{h_{4}}\partial _{4}h_{4}-\frac{3}{2}\partial _{4}h_{3}\right) \frac{%
\partial _{4}n_{k}}{2h_{4}},  \label{equ3} \\
\widehat{R}_{4k} &=&\frac{w_{k}}{2h_{3}}[\partial _{44}^{2}h_{3}-\frac{%
\left( \partial _{4}h_{3}\right) ^{2}}{2h_{3}}-\frac{(\partial
_{4}h_{3})(\partial _{4}h_{4})}{2h_{4}}]+\frac{\partial _{4}h_{3}}{4h_{3}}(%
\frac{\partial _{k}h_{3}}{h_{3}}+\frac{\partial _{k}h_{4}}{h_{4}})-\frac{%
\partial _{k}(\partial _{4}h_{3})}{2h_{3}},  \label{equ4}
\end{eqnarray}%
and for $r$--jet components with partial derivative operator $\eth $ on jet
variables, on shell $s=1,\tau =1,2,3,4;$
\begin{eqnarray}
\widehat{R}_{5}^{5} &=&\widehat{R}_{6}^{6}=-\frac{1}{2h_{5}h_{6}}[\eth
_{66}h_{5}-\frac{\left( \eth _{6}h_{5}\right) ^{2}}{2h_{5}}-\frac{(\eth
_{6}h_{5})(\eth _{6}h_{6})}{2h_{6}}],  \label{equ5} \\
\widehat{R}_{5\tau } &=&\frac{h_{5}}{2h_{6}}\eth _{66}^{2}\ ^{1}n_{\tau
}+\left( \frac{h_{5}}{h_{6}}\eth _{6}h_{6}-\frac{3}{2}\eth _{6}h_{5}\right)
\frac{\eth _{6}\ ^{1}n_{\tau }}{2h_{6}},  \label{equ6} \\
\widehat{R}_{6\tau } &=&\frac{\ ^{1}w_{\tau }}{2h_{5}}[\eth _{66}^{2}h_{5}-%
\frac{\left( \eth _{6}h_{5}\right) ^{2}}{2h_{5}}-\frac{(\eth _{6}h_{5})(\eth
_{6}h_{6})}{2h_{6}}]+\frac{\eth _{6}h_{5}}{4h_{5}}(\frac{\eth _{\tau }h_{5}}{%
h_{5}}+\frac{\eth _{\tau }h_{6}}{h_{6}})-\frac{\eth _{\tau }(\eth _{6}h_{5})%
}{2h_{5}},  \label{equ7}
\end{eqnarray}%
and, for extra shells with number $s,$
\begin{eqnarray}
\widehat{R}_{2s-1}^{2s-1} &=&\widehat{R}_{2s}^{2s}=-\frac{1}{%
2h_{3+2s}h_{4+2s}}[\eth _{4+2s\ 4+2s}^{2}h_{3+2s}-\frac{\left( \eth
_{4+2s}h_{3+2s}\right) ^{2}}{2h_{3+2s}}-\frac{(\eth _{4+2s}h_{3+2s})(\eth
_{4+2s}h_{4+2s})}{2h_{4+2s}}],  \notag \\
\widehat{R}_{3+2s\ \tau _{1}} &=&\frac{h_{2s-1}}{2h_{2s}}\eth _{2s\ 2s}^{2}\
^{2}n_{\tau _{1}}+\left( \frac{h_{2s-1}}{h_{2s}}\eth _{4+2s}h_{4+2s}-\frac{3%
}{2}\eth _{4+2s}h_{3+2s}\right) \frac{\eth _{4+2s}\ ^{2}n_{\tau _{1}}}{%
2h_{3+2s}},  \notag \\
\widehat{R}_{4+2s\ \tau _{1}} &=&\frac{\ ^{2}w_{\tau _{1}}}{2h_{7}}[\eth
_{4+2s\ 4+2s}^{2}h_{3+2s}-\frac{\left( \eth _{4+2s}h_{3+2s}\right) ^{2}}{%
2h_{3+2s}}-\frac{(\eth _{4+2s}h_{3+2s})(\eth _{4+2s}h_{4+2s})}{2h_{4+2s}}]+
\label{equ4d2s} \\
&&\frac{\eth _{4+2s}h_{3+2s}}{4h_{3+2s}}(\frac{\eth _{\tau _{1}}h_{3+2s}}{%
h_{3+2s}}+\frac{\eth _{\tau _{1}}h_{4+2s}}{h_{4+2s}})-\frac{\eth _{\tau
_{1}}(\partial _{4+2s}h_{3+2s})}{2h_{3+2s}},  \notag
\end{eqnarray}%
when $\tau _{1}=1,2,3,4,5,6;$
\begin{equation*}
...
\end{equation*}
\end{lemma}

\begin{proof}
We introduce the coefficients of the canonical d--connection $\widehat{%
\mathbf{\Gamma }}_{\ \alpha _{s}\beta _{s}}^{\gamma _{s}}$(\ref{candcon})
for the d--metric ansatz (\ref{ansk}) and compute the N--adapted
d--curvature coefficients (\ref{dcurv}) and (\ref{dtorsj}). Then,
contracting the indices (following formulae (\ref{dricci}), (\ref{rdsc}) and
(\ref{einstdt})) we find the nontrivial values of the N--adapted
coefficientes for the Ricci d--tensor, scalar curvature and Einstein
d--tensor of \ $^{s}\widehat{\mathbf{D}}.$ Explicit proofs of the formulae (\ref%
{equ1})--(\ref{equ4}) for 4-d and extra dimensional indices are provided in
a series of our works, for instance, in \cite{vexsol2,vex3,vexsol1}. We do
not repeat the required calculus in this paper.

Introducing $r$--jet variables, we observe that on the first shell, with $%
s=1, $ the formulae (\ref{equ2})--(\ref{equ4}) are generalized in a similar
form but for the partial derivatives $\eth $ on jet variables, with
respective indices 5 and 6 for a nonholonomic 2+2+2+... splitting. To avoid
ambiguities, we put left labels $s=1$ on the necessary geometric objects and
coefficients. On this shell, the first four coordinates $\alpha =1,2,3,4$
are treated as "base type"  but take $a_{1},b_{1},...=5,6$ as conventional
"fiber/jet" ones. In symbolic form, the equations (\ref{equ5})--(\ref{equ7})
are constructed via  formally increasing by 2  respective values of 4-d
spacetime indices and introducing dependencies on all "base/ spacetime"
coordinates.

For shells $s=2,3,...,$ "fiber/jet" indices are labeled with values of type $%
3+2s$ and $4+2s$ and the previous (base type) indices take values $1,2,...2+2s.$ The
equations (\ref{equ4d2s}) present a "recurrent" generalizations for a finite
number of shells, $s,$ of the 1st jet shell when $s=1.$

$\square $
\end{proof}

\vskip5pt

Let us analyze some important nonholonomic symmetries of of the canonical
Ricci and Einstein d--tensors:

For $s=1$ and using the above formulae, we can compute the Ricci scalar (\ref%
{rdsc}) for $\ ^{1}\widehat{\mathbf{D}}$ $,$ $\ ^{1}\widehat{R}=2(\widehat{R}%
_{1}^{1}+\widehat{R}_{3}^{3}+\widehat{R}_{5}^{5}).$ There are certain
N--adapted symmetries of the Einstein d--tensor (\ref{einstdt}) for the
ansatz (\ref{ansk}), $\widehat{E}_{1}^{1}=\widehat{E}_{2}^{2}=-(\widehat{R}%
_{3}^{3}+\widehat{R}_{5}^{5}),\widehat{E}_{3}^{3}=\widehat{E}_{4}^{4}=-(%
\widehat{R}_{1}^{1}+\widehat{R}_{5}^{5}),\widehat{E}_{5}^{5}=\widehat{E}%
_{6}^{6}=-(\widehat{R}_{1}^{1}+\widehat{R}_{3}^{3}).$

In a similar form, we find symmetries for $s=2:$%
\begin{eqnarray*}
\widehat{E}_{1}^{1} &=&\widehat{E}_{2}^{2}=-(\widehat{R}_{3}^{3}+\widehat{R}%
_{5}^{5}+\widehat{R}_{7}^{7}),\widehat{E}_{3}^{3}=\widehat{E}_{4}^{4}=-(%
\widehat{R}_{1}^{1}+\widehat{R}_{5}^{5}+\widehat{R}_{7}^{7}), \\
\widehat{E}_{5}^{5} &=&\widehat{E}_{6}^{6}=-(\widehat{R}_{1}^{1}+\widehat{R}%
_{3}^{3}+\widehat{R}_{7}^{7}),\widehat{E}_{7}^{7}=\widehat{E}_{8}^{8}=-(%
\widehat{R}_{1}^{1}+\widehat{R}_{3}^{3}+\widehat{R}_{5}^{5}).
\end{eqnarray*}

We conclude that the nonholonomically jet modified Einstein equations (\ref%
{equ1})--(\ref{equ4d2s}) for $s=2$ jet shells with nontrivial $\Lambda $%
--sources can be written in N--adapted form as
\begin{equation}
\widehat{R}_{1}^{1} = \widehat{R}_{2}^{2}=-\Lambda (x^{k}),\ \widehat{R}%
_{3}^{3}=\widehat{R}_{4}^{4}=-\ ^{v}\Lambda (x^{k},y^{4}), \widehat{R}%
_{5}^{5} = \widehat{R}_{6}^{6}=-\ _{1}^{v}\Lambda (u^{\beta },\zeta ^{6}),\
\widehat{R}_{7}^{7}=\widehat{R}_{8}^{8}=-\ _{2}^{v}\Lambda (u^{\beta
_{1}},\zeta ^{8}), ....,  \label{sourc1}
\end{equation}%
which can be extended for any arbitrary finite number of jets' shells.

\subsection{Nonholonomic spacetime and $r$--jet vacuum solutions}

\subsubsection{ 4--d nonhlonomic vacuum configurations}

\label{assvacuum}To consider vacuum solutions for $\widehat{\mathbf{D}}$
with $\ ^{v}\Lambda =0$ in (\ref{e2}) we study configurations with
N--adapted coefficients when $\partial _{4}h_{3}=0$ and/or $\partial
_{4}\phi =0.$ The limits to the off--diagonal solutions with $\ \Lambda =\
^{v}\Lambda =0$ are not smooth because multiples $(\ ^{v}\Lambda )^{-1}$ are
considered in various coefficients and re--defined generating functions for
solutions (\ref{qnk4d}).

Let us analyze the conditions when the nontrivial coefficients of the Ricci
d--tensor (\ref{equ1})--(\ref{equ4}) are zero for ansatz (\ref{ansk}). The
first equation is a typical example of 2--d wave or Laplace equation. We
can express such solutions in a similar form $g_{i}=\epsilon _{i}e^{\psi
(x^{k},\Lambda =0)}(dx^{i})^{2}.$

There are three classes of off--diagonal metrics giving zero coefficients (%
\ref{equ2})--(\ref{equ4}).

\begin{enumerate}
\item We impose the condition $\partial _{4}h_{3}=0,h_{3}\neq 0,$ giving
only one nontrivial equation, see (\ref{equ3}), $\partial
_{44}^{2}n_{k}+\partial _{4}n_{k}\ \partial _{4}\ln |h_{4}|=0,$ where $%
h_{4}(x^{i},y^{4})\neq 0$ and $w_{k}(x^{i},y^{4})$ are arbitrary functions.
If $\partial _{4}h_{4}=0,$ we must take $\partial _{44}^{2}n_{k}=0.$ For $%
\partial _{4}h_{4}\neq 0,$ we get
\begin{equation}
n_{k}=\ _{1}n_{k}+\ _{2}n_{k}\int dy^{4}/h_{4}  \label{wsol}
\end{equation}%
with integration functions $\ _{1}n_{k}(x^{i})$ and $\ _{2}n_{k}(x^{i}).$
The corresponding class of nonholonomic vacuum solutions is defined by
quadratic line element
\begin{eqnarray*}
ds_{v1}^{2}&=&\epsilon _{i}e^{\psi (x^{k},\Lambda =0)}(dx^{i})^{2}+\
^{0}h_{3}(x^{k})[dy^{3}+(\ _{1}n_{k}(x^{i})+\ _{2}n_{k}(x^{i})\int
dy^{4}/h_{4})dx^{i}]^{2} \\
&& + h_{4}(x^{i},y^{4})[dy^{4}+w_{i}(x^{k},y^{4})dx^{i}].
\end{eqnarray*}

\item Let us assume $\partial _{4}h_{3}\neq 0$ and $\partial _{4}h_{4}\neq
0. $ We can solve (\ref{equ2}) and/or (\ref{e2}) for $\ ^{v}\Lambda =0$ if $%
\partial _{4}\phi =0$ for coefficients (\ref{c1}) and (\ref{ca1}). For $\phi
=\phi _{0}=const,$ we can consider arbitrary functions $w_{i}(x^{k},y^{4})$
as generating functions because $\beta =\alpha _{i}=0$ for such
configurations. The condition (\ref{ca1}) is satisfied by any
\begin{equation}
h_{4}=\ ^{0}h_{4}(x^{k})(\partial _{4}\sqrt{|h_{3}|})^{2},  \label{h34vacuum}
\end{equation}%
where $\ ^{0}h_{3}(x^{k})$ is an integration function and $%
h_{3}(x^{k},y^{4}) $ is any generating function. The coefficients $n_{k}$
are found from (\ref{equ3}), see (\ref{wsol}). The corresponding class of
nonholonomic vacuum metrics is defined by the quadratic line element
\begin{eqnarray}
ds_{v2}^{2} &=&\epsilon _{i}e^{\psi (x^{k},\Lambda
=0)}(dx^{i})^{2}+h_{3}(x^{i},y^{4})[dy^{3}+(\ _{1}n_{k}(x^{i})+\
_{2}n_{k}(x^{i})\int dy^{4}/h_{4})dx^{i}]^{2}+  \label{vs2} \\
&&\ ^{0}h_{4}(x^{k})(\partial _{4}\sqrt{|h_{3}|}%
)^{2}[dy^{4}+w_{i}(x^{k},y^{4})dx^{i}]^{2}.  \notag
\end{eqnarray}

\item Another type of configurations are generated by $\partial _{4}h_{3}\neq 0$
but $\partial _{4}h_{4}=0.$ The equation (\ref{equ2}) is $\partial
_{44}^{2}h_{3}-\frac{\left( \partial _{4}h_{3}\right) ^{2}}{2h_{3}}=0,$ with
general solution is $h_{3}(x^{k},y^{4})=\left[ c_{1}(x^{k})+c_{2}(x^{k})y^{4}%
\right] ^{2}$, \ where $c_{1}(x^{k}),c_{2}(x^{k})$ are generating functions
and $h_{4}=\ ^{0}h_{4}(x^{k}).$ For $\phi =\phi _{0}=const,$ we can choose any
values $w_{i}(x^{k},y^{4})$ because $\beta =\alpha _{i}=0.$ The coefficients
$n_{i}$ are determined by equation (\ref{equ3}) and/or, equivalently, (\ref%
{e3}) with $\gamma =\frac{3}{2}\partial _{4}|h_{3}|.$ We find
\begin{equation*}
n_{i}=\ _{1}n_{i}(x^{k})+\ _{2}n_{i}(x^{k})\int dy^{4}|h_{3}|^{-3/2}=\
_{1}n_{i}(x^{k})+\ _{2}\widetilde{n}_{i}(x^{k})
[c_{1}(x^{k})+c_{2}(x^{k})y^{4}] ^{-2},
\end{equation*}%
with integration functions $\ _{1}n_{i}(x^{k})$ and $\ _{2}n_{i}(x^{k}),$ or
re--defined $\ _{2}\widetilde{n}_{i}=-\ _{2}n_{i}/2c_{2}.$ The quadratic
line element for this class of vacuum nonholonomic solutions is given by
\begin{eqnarray}
ds_{v3}^{2} &=&\epsilon _{i}e^{\psi (x^{k},\Lambda =0)}(dx^{i})^{2}+\left[
c_{1}(x^{k})+c_{2}(x^{k})y^{4}] ^{2}[dy^{3}+(\ _{1}n_{i}(x^{k})+\ _{2}%
\widetilde{n}_{i}(x^{k})[c_{1}(x^{k})+c_{2}(x^{k})y^{4}] ^{-2})dx^{i}\right]%
^{2}  \notag \\
&&+\ ^{0}h_{4}(x^{k})[dy^{4}+w_{i}(x^{k},y^{4})dx^{i}]^{2}.  \label{vs3}
\end{eqnarray}
\end{enumerate}

Finally, we note that such solutions are with nontrivial induced torsion (%
\ref{dtors}) and that additional assumptions are necessary  to extract
vacuum LC--configurations.

\subsubsection{Nonholonomic $r$--jet prolongations of vacuum solutions:}

The quadratic line elements (\ref{qnk4d}), (\ref{qnk6d}), (\ref{qnk8d}),... for
off--diagonal jet prolongations of generic off--diagonal solutions have been
constructed for nontrivial sources $\ ^{v}\Lambda (x^{k},y^{4}),$ $\
_{1}^{v}\Lambda (u^{\tau },\zeta ^{6}),$ $\ _{2}^{v}\Lambda (u^{\tau },\zeta
^{8}),...$ In a similar manner, we can generate jet prolongations of vacuum
configurations with effective zero cosmological constants extending with $r$%
--jet variables the 4-d vacuum metrics of type $ds_{v1}^{2}$, $ds_{v2}^{2}$ (%
\ref{vs2}), $ds_{v3}^{2}$ (\ref{vs3}) etc. It is possible to generate
solutions when the sources are zero on some shells and nonzero on other
shells.

Let us consider an example of quadratic line element for jet prolongation of
effective 6--d gravity derived as a $s=1$ generalization of (\ref{vs2}). For
such solutions, $\partial _{4}h_{a}\neq 0,\eth _{6}h_{a_{1}}\neq 0,...$ and $%
\phi =\phi _{0}=const,$ $\ ^{1}\phi =\ ^{1}\phi _{0}=const,...$
\begin{eqnarray}
&&ds_{v2s3}^{2}=\epsilon _{i}e^{\psi (x^{k},\Lambda
=0)}(dx^{i})^{2}+h_{3}(x^{i},y^{4})[dy^{3}+\left( \ _{1}n_{k}(x^{i})+\
_{2}n_{k}(x^{i})\int dy^{4}/h_{4}\right) dx^{i}]^{2}+  \label{qe6dvacuum} \\
&&\ ^{0}h_{4}(x^{k})(\partial _{4}\sqrt{|h_{3}|}%
)^{2}[dy^{4}+w_{i}(x^{k},y^{4})dx^{i}]^{2}+h_{5}(u^{\tau },\zeta
^{6})[d\zeta ^{5}+\left( \ _{1}^{1}n_{\lambda }(u^{\tau })+\
_{2}^{1}n_{\lambda }(u^{\tau })\int d\zeta ^{6}/h_{6}\right) du^{\lambda
}]^{2}  \notag \\
&&+\ ^{0}h_{6}(u^{\tau })(\eth _{6}\sqrt{|h_{5}|})^{2}[d\zeta ^{6}+\
^{1}w_{\lambda }(u^{\tau },\zeta ^{6})du^{\lambda }]^{2},  \notag
\end{eqnarray}%
where $\ ^{0}h_{3}(x^{k}),\ ^{0}h_{5}(u^{\tau }),\ _{1}n_{k}(x^{i}),\
_{2}n_{k}(x^{i}),\ _{1}^{1}n_{\lambda }(u^{\tau }),\ _{2}^{1}n_{\lambda
}(u^{\tau })$ are integration functions. \ The values $h_{4}(x^{k},y^{4})$
and $h_{6}(u^{\tau },\zeta ^{6})$ are any generating functions depending on spacetime
and jet prolongation variables. We can consider arbitrary functions $%
w_{i}(x^{k},y^{4})$ and $\ ^{1}w_{\lambda }(u^{\tau },\zeta ^{6})$ because,
respectively, $\beta =\alpha _{i}=0$ and $\ ^{1}\beta =\ ^{1}\alpha _{\tau
}=0$ for such configurations, see formulas (\ref{c1}), (\ref{ca1}) and (\ref%
{c2}), (\ref{ca2}).

\subsection{ The LC--conditions}

\label{zt} We can consider nonholonomic frame deformations of the
N--coefficients and ansatz (\ref{ansk}) when all coefficients of a
nonholonomically induced torsion (\ref{dtors}) are zero and $\ _{\shortmid
}\Gamma _{\ \alpha _{s}\beta _{s}}^{\gamma _{s}}=\widehat{\mathbf{\Gamma }}%
_{\ \alpha _{s}\beta _{s}}^{\gamma _{s}}.$ For simplicity, we analyze such
conditions for  4--d spacetime (generalizations to extra jet shell can be
performed recurrently as we explained in section \ref{s3}).

The trivial coefficients of d--torsion (\ref{dtors}) are $\widehat{T}_{\
jk}^{i}=\widehat{L}_{jk}^{i}-\widehat{L}_{kj}^{i}=0,~\widehat{T}_{\ ja}^{i}=%
\widehat{C}_{jb}^{i}=0,~\widehat{T}_{\ bc}^{a}=\ \widehat{C}_{bc}^{a}-\
\widehat{C}_{cb}^{a}=0$ for any ansatz (\ref{ansk}). Let us compute the
nontrivial coefficients $\widehat{T}_{aj}^{c}=\widehat{L}%
_{aj}^{c}-e_{a}(N_{j}^{c})$ and $\widehat{T}_{\ ji}^{a}=-\ ^{N}J_{\ ji}^{a}.$
For a 2+2 spacetime splitting, the values
\begin{equation*}
\widehat{L}_{bi}^{a} =\partial _{b}N_{i}^{a}+\frac{1}{2}h^{ac}(\partial
_{i}h_{bc}-N_{i}^{e}\partial _{e}h_{bc}-h_{dc}\partial
_{b}N_{i}^{d}-h_{db}\partial _{c}N_{i}^{d}),\ \widehat{T}_{aj}^{c} = \frac{1%
}{2}h^{ac}(\partial _{i}h_{bc}-N_{i}^{e}\partial _{e}h_{bc}-h_{dc}\partial
_{b}N_{i}^{d}-h_{db}\partial _{c}N_{i}^{d}).
\end{equation*}%
are computed for $N_{i}^{3}=n_{i}(x^{k},y^{4}),N_{i}^{4}=w_{i}(x^{k},y^{4});$
$h_{bc}=diag[h_{3}(x^{k},y^{4}),h_{4}(x^{k},y^{4})];$ $\
h^{ac}=diag[(h_{3})^{-1},(h_{4})^{-1}].$ We write
\begin{eqnarray*}
\widehat{T}_{bi}^{3} &=&\frac{1}{2}h^{3c}(\partial
_{i}h_{bc}-N_{i}^{e}\partial _{e}h_{bc}-h_{dc}\partial
_{b}N_{i}^{d}-h_{db}\partial _{c}N_{i}^{d})=\frac{1}{2h_{3}}(\partial
_{i}h_{b3}-w_{i}\partial _{4}h_{b3}-h_{3}\partial _{b}n_{i}), \\
\mbox{ i.e. }\widehat{T}_{3i}^{3} &=&\frac{1}{2h_{3}}(\partial
_{i}h_{3}-w_{i}\partial _{4}h_{3}),\ \widehat{T}_{4i}^{3}=\frac{1}{2}%
\partial _{4}n_{i}.
\end{eqnarray*}%
In a similar form, we compute
\begin{eqnarray*}
&&\widehat{T}_{bi}^{4}=\frac{1}{2}h^{4c}(\partial
_{i}h_{bc}-N_{i}^{e}\partial _{e}h_{bc}-h_{dc}\partial
_{b}N_{i}^{d}-h_{db}\partial _{c}N_{i}^{d})=\frac{1}{2h_{4}}(\partial
_{i}h_{b4}-w_{i}\partial _{4}h_{b4}-h_{4}\partial _{b}w_{i}-h_{3b}\partial
_{4}n_{i}-h_{4b}\partial _{4}w_{i}) \\
&&\mbox{ i.e. }\widehat{T}_{3i}^{4}=-\frac{h_{3}}{2h_{4}}\partial
_{4}n_{i},\ \widehat{T}_{4i}^{4}=\frac{1}{2h_{4}}(\partial
_{i}h_{4}-w_{i}\partial _{4}h_{4})-\partial _{4}w_{i}.
\end{eqnarray*}

The coefficients of the N--connection curvature $\ ^{N}J_{ij}^{a}=\mathbf{e}%
_{j}\left( N_{i}^{a}\right) -\mathbf{e}_{i}(N_{j}^{a})$ are expressed as
\begin{equation*}
\ ^{N}J_{ij}^{a} =\mathbf{\partial }_{j}\left( N_{i}^{a}\right) -\partial
_{i}(N_{j}^{a})-N_{j}^{b}\partial _{b}N_{i}^{a}+N_{i}^{b}\partial
_{b}N_{j}^{a} = \mathbf{\partial }_{j}\left( N_{i}^{a}\right) -\partial
_{i}(N_{j}^{a})-w_{j}\partial _{4}N_{i}^{a}+w_{i}\partial _{4}N_{j}^{a}
\end{equation*}%
with  nontrivial values:
\begin{equation}
\ ^{N}J_{12}^{3} =-\ ^{N}J_{21}^{3}=\mathbf{\partial }_{2}n_{1}-\partial
_{1}n_{2}-w_{2}\partial _{4}n_{1}+w_{1}\partial _{4}n_{2}{}, \ \
^{N}J_{12}^{4} = -\ ^{N}J_{21}^{4}=\mathbf{\partial }_{2}w_{1}-\partial
_{1}w_{2}-w_{2}\partial _{4}w_{1}+w_{1}\partial _{4}w_{2}.  \label{omeg}
\end{equation}

Summarizing the above formulae for $\partial _{4}n_{i}=0$ and $\mathbf{\partial }%
_{2}n_{1}-\partial _{1}n_{2}=0,$ we get the condition for zero torsion for the
ansatz (\ref{ansk}) with $n_{k}=\partial _{k}n(x^{i}),$%
\begin{eqnarray}
\frac{1}{2h_{3}}(\partial _{i}h_{3}-w_{i}\partial _{4}h_{3}) =0,\ \frac{1}{%
2h_{4}}(\partial _{i}h_{4}-w_{i}\partial _{4}h_{4})&=&\partial _{4}w_{i},
\label{qa2} \\
\mathbf{\partial }_{2}w_{1}-\partial _{1}w_{2}-w_{2}\partial
_{4}w_{1}+w_{1}\partial _{4}w_{2} &=&0.  \label{qa3}
\end{eqnarray}%
From this, we can define a LC--configuration. The final step is to impose the
condition that the coefficients $n_{k}$ do not depend on $y^{4}.$ This can
be fixed for  $_{1}n_{k}(x^{i})=\partial _{k}n(x^{i})$ and $_{2}n_{k}=0,$
i.e. $n_{k}=\partial _{k}n(x^{i}).$

Finally, we note that the LC-conditions can be formulated recurrently, in
similar forms, for higher order shells of jet coordinates using the partial
derivative operator $\eth $ both for zero and non-zero sources.

\end{document}